\documentclass[acmsmall,screen]{acmart}
\AtBeginDocument{%
  }

\usepackage{proof}
\usepackage{tikz}
\usetikzlibrary{arrows,positioning,shapes,fit}
\usepackage{longtable}
\usepackage{pifont}
\usepackage{diagbox}
\usepackage{listings}
\usepackage{subcaption}
\usepackage{wrapfig}
\usepackage{pifont}
\usepackage[nameinlink,noabbrev]{cleveref}
\usepackage{mathpartir}

\makeatletter
\newcommand\tightparagraph{\def\@toclevel{4}%
  \@startsection{paragraph}{4}{\parindent}%
  {-.15\baselineskip \@minus -.05\baselineskip}%
  {-3.5\p@}%
  {\ACM@NRadjust{\@parfont}}}
\makeatother

\newcommand\Tstrut{\rule{0pt}{2.6ex}}         

\newcommand{\semspace}{\vspace{1.1\baselineskip}}
\newcommand{\newpremise}{\qquad}

\newcommand{\ruledef}[2]{\!{\textsc{\footnotesize#2}\phantomsection\label{rule:#1:#2}}}
\newcommand{\ruleref}[2]{\hyperref[rule:#1:#2]{\textsc{#2}}}
\newcommand{\Autoref}[1]{\Cref{#1}}
\renewcommand{\autoref}[1]{\Cref{#1}}
\newcommand{\appendixref}[1]{\hyperref[#1]{Appendix~\ref*{#1}}}

\newcommand{\dir}{\mathit{d}}
\newcommand{\forward}{\rightarrow}
\newcommand{\backward}{\leftarrow}
\newcommand{\both}{\leftrightarrow}
\newcommand{\ok}{\text{\ding{51}}}
\newcommand{\nok}{\text{\ding{55}}}
\newcommand{\notapp}{\text{n/a}}

\newcommand{\istree}[5]{(#1,#2,#3,#4)\mathbin{\Downarrow}#5}
\newcommand{\booltree}[4]{(#1,#2,#3)\mathbin{\overline{\Downarrow}}#4}

\newcommand{\ancsymb}{%
  \vcenter{\hbox{\oalign{
        \noalign{\kern-0.7ex}\hfil\rule{.4pt}{1.5ex}\hfil\cr
        \noalign{\kern-0.7ex}$\smile$\cr
        \noalign{\kern-1.8ex}\hfil{\scalebox{.8}{$-$}}\hfil\cr
        \noalign{\kern-1.5ex}\hfil{\scalebox{.5}{$\circ$}}\hfil\cr
      }}}}

\newcommand{\regbos}{\text{\textasciicircum}}
\newcommand{\regeos}{\text{\$}}

\newcommand{\inp}{\mathit{i}}
\newcommand{\inpcheck}{\mathit{i_{check}}}
\newcommand{\gm}{\mathit{gm}}
\newcommand{\gid}{\mathit{g}}
\newcommand{\gidl}{\mathit{gl}}
\newcommand{\cont}{\mathit{l}}
\newcommand{\actions}{\mathit{l}}
\newcommand{\cd}{\mathit{cd}}
\newcommand{\anc}{\mathit{a}}
\newcommand{\rchar}{\mathit{c}}
\newcommand{\subreg}{\mathit{r}}
\newcommand{\greedy}{\mathit{p}}
\newcommand{\strbr}{\mathit{s}}
\newcommand{\lk}{\mathit{lk}}
\newcommand{\result}{\mathit{res}}
\newcommand{\resulta}{\mathit{res}_{a}}
\newcommand{\resultb}{\mathit{res}_{b}}
\newcommand{\rmin}{\mathit{min}}
\newcommand{\rmax}{\mathit{max}}

\newcommand{\treecont}{\mathit{t}}
\newcommand{\treeiter}{\mathit{t_{iter}}}
\newcommand{\treelook}{\mathit{t_{look}}}

\newcommand{\str}[1]{\texttt{"#1"}}

\newcommand{\jsvar}[1]{\ensuremath{\mathit{#1}}}
\newcommand{\flag}[1]{\texttt{#1}}

\newcommand{\cdcomplement}{\hat{~}}

\newcommand{\reskip}{\mskip\medmuskip}
\newcommand{\resskip}{\mskip.5\thinmuskip}

\newcommand{\defineRegexSyntax}{%
\newcommand{\regepsilon}{\varepsilon}%
\newcommand{\regchar}[1]{##1}%
\newcommand{\charset}[1]{[##1]}%
\newcommand{\cdrange}[2]{##1\mathord{-}##2}%
\newcommand{\disjunction}[2]{##1|##2}%
\newcommand{\sequence}[2]{##1\reskip##2}%
\newcommand{\quant}[4]{##1\{##2,##3,##4\}}%
\newcommand{\anchor}[1]{##1}%
\newcommand{\backref}[1]{\backslash##1}%
\newcommand{\lookaround}[2]{(?{##1}\reskip##2)}%
\newcommand{\group}[2]{(_{##1}\resskip##2)}%
\newcommand{\noncap}[1]{\langle##1\rangle}%
\newcommand{\respecial}[1]{\mathord{\mathsf{##1}}}
\newcommand{\regqm}[1]{##1\resskip\respecial{?}}%
\newcommand{\regstar}[1]{##1\resskip\respecial{*}}%
\newcommand{\regplus}[1]{##1\resskip\respecial{+}}%
\newcommand{\lazystar}[1]{##1\resskip\respecial{*?}}%
\newcommand{\lazyplus}[1]{##1\resskip\respecial{+?}}%
\newcommand{\lazyqm}[1]{##1\resskip\respecial{??}}%
\newcommand{\regdot}{\cdot}%
\newcommand{\regall}{\odot}%
\newcommand{\esc}[1]{\backslash\respecial{##1}}%
\newcommand{\posahead}{\respecial{=}}%
\newcommand{\negahead}{\respecial{!}}%
\newcommand{\posbehind}{\respecial{<=}}%
\newcommand{\negbehind}{\respecial{<!}}%
\newcommand{\lookahead}[1]{\lookaround{\posahead}{##1}}%
\newcommand{\neglookahead}[1]{\lookaround{\negahead}{##1}}%
\newcommand{\lookbehind}[1]{\lookaround{\posbehind}{##1}}%
\newcommand{\neglookbehind}[1]{\lookaround{\negbehind}{##1}}%
}

\newcommand{\re}[1]{\ensuremath{{\defineRegexSyntax\mathrm{#1}}}}

\newcommand{\savere}[2]{%
  \expandafter\newsavebox\csname re:#1\endcsname%
  \sbox{\csname re:#1\endcsname}{{\small#2}}%
}
\newcommand{\usere}[1]{\usebox{\csname re:#1\endcsname}}

\newcommand{\ctor}[1]{\textsf{\textup{#1}}}
\newcommand{\fn}[1]{\textsf{\textup{#1}}}
\newcommand{\Some}[1]{\ctorone{Some}{#1}}
\newcommand{\None}{\ctor{None}}

\newcommand{\ctorone}[2]{\expandafter\ifx\expandafter\relax
  \detokenize{#2}\relax\ctor{#1}\else\ctor{#1}~#2\fi}
\newcommand{\ctortwo}[3]{\expandafter\ifx\expandafter\relax
  \detokenize{#2}\relax\ctor{#1}\else\ctor{#1}~#2~#3\fi}
\newcommand{\ctorthree}[4]{\expandafter\ifx\expandafter\relax
  \detokenize{#2}\relax\ctor{#1}\else\ctor{#1}~#2~#3~#4\fi}

\newcommand{\acheck}[1]{\ctorone{Acheck}{#1}}
\newcommand{\aclose}[1]{\ctorone{Aclose}{#1}}
\newcommand{\areg}[1]{#1}

\newcommand{\idx}[1]{\fn{idx}(#1)}
\newcommand{\strof}[1]{\fn{str}(#1)}
\newcommand{\inpadvance}[3]{\fn{advance}(#1, #2, #3)}
\newcommand{\defgroups}[1]{\mathcal{G}(#1)}
\newcommand{\checkanchor}[2]{\fn{check\_anchor}(#1,#2)}
\newcommand{\readbackref}[4]{\fn{advance\_bkrf}(#1,#2,#3,#4)}
\newcommand{\lkdir}[1]{\fn{dir}(#1)}
\newcommand{\lkresult}[4]{\fn{lk\_result}(#1,#2,#3,#4)}

\newcommand{\gmclose}[3]{\fn{GM\textsubscript{close}}(#1,#2,#3)}
\newcommand{\gmopen}[3]{\fn{GM\textsubscript{open}}(#1,#2,#3)}
\newcommand{\gmreset}[2]{\fn{GM\textsubscript{reset}}(#1,#2)}
\newcommand{\gmempty}{\fn{GM\textsubscript{$\emptyset$}}}

\newcommand{\treematch}{\ctor{Match}}
\newcommand{\treemismatch}{\ctor{Mismatch}}
\newcommand{\treeprogress}[1]{\ctorone{Progress}{#1}}
\newcommand{\treeopen}[2]{\ctortwo{Open}{#1}{#2}}
\newcommand{\treeclose}[2]{\ctortwo{Close}{#1}{#2}}
\newcommand{\treereset}[2]{\ctortwo{Reset}{#1}{#2}}
\newcommand{\treeread}[2]{\ctortwo{Read}{\re{#1}}{#2}}
\newcommand{\treechoice}[2]{\ctortwo{Choice}{#1}{#2}}
\newcommand{\treeanchor}[2]{\ctortwo{AnchorPass}{#1}{#2}}
\newcommand{\treebackref}[2]{\ctortwo{BackrefPass}{#1}{#2}}
\newcommand{\treelk}[3]{\ctorthree{LK}{#1}{#2}{#3}}
\newcommand{\treelkfail}[2]{\ctortwo{LKMismatch}{#1}{#2}}

\newcommand{\computetreefuel}[5]{\mathcal{T}(#1,#2,#3,#4,#5)}
\newcommand{\computetree}[4]{\mathcal{T}(#1,#2,#3,#4)}
\newcommand{\fuel}[3]{\lvert\lvert#1\rvert\rvert^{#2}_{#3}}
\newcommand{\firstbranch}[2]{\mathcal{L}_0(#1,#2)}
\newcommand{\sinput}{\ctor{Input}}
\newcommand{\worstinp}[2]{\fn{worst}(#1,#2)}
\newcommand{\inpsize}[2]{|#1|_{#2}}

\newcommand{\encodes}[3]{(#1,#2) \vdash #3}
\newcommand{\canexit}{\top}
\newcommand{\cannotexit}{\bot}

\newcommand{\pikesub}{\mathcal{P}}
\newcommand{\insub}[1]{#1 \in \pikesub}

\newcommand{\pts}{\mathit{pts}}
\newcommand{\ptsinit}[2]{\ctor{PT\textsubscript{init}}(#1,#2)}
\newcommand{\ptsfinal}[1]{\ctor{PT\textsubscript{end}}(#1)}
\newcommand{\piketreeinv}[2]{#1 \downarrow\downarrow #2}
\newcommand{\ptres}[4]{#1 \downarrow_{#3}^{#4} #2}
\newcommand{\piketreestep}[2]{#1 \rightarrow #2}
\newcommand{\piketreestar}[2]{#1 \rightarrow^{*} #2}
\newcommand{\ptstate}[5]{(#1,#2,#3,#4,#5)}
\newcommand{\intseen}[2]{#1\in#2}
\newcommand{\addtseen}[2]{#1\cup\{#2\}}

\newcommand{\pvs}{\mathit{pvs}}
\newcommand{\pvsinit}[1]{\ctor{VM\textsubscript{init}}(#1)}
\newcommand{\pvsfinal}[1]{\ctor{VM\textsubscript{end}}(#1)}
\newcommand{\code}{\mathit{code}}
\newcommand{\compilation}[1]{\fn{NFA}(#1)}
\newcommand{\pikeinv}[3]{#1\mathbin{\sim_{#3}}#2}
\newcommand{\pikevmstep}[3]{#1\mathbin{\rightarrow_{#3}}#2}
\newcommand{\trc}{\rightarrow^{*}}
\newcommand{\pikevmstar}[3]{#1\mathbin{\trc_{#3}}#2}

\newcommand{\best}{\mathit{best}}
\newcommand{\blocked}{\mathcal{B}}
\newcommand{\pactive}{\mathcal{A}}
\newcommand{\seen}{\mathcal{S}}
\newcommand{\pvstate}[5]{(#1,#2,#3,#4,#5)}
\newcommand{\nextinp}[1]{\fn{next}(#1)}
\newcommand{\nextinpdir}[2]{\nextinp{#1}_{#2}}
\newcommand{\thread}[3]{(#1,#2,#3)}
\newcommand{\pc}{\mathit{pc}}
\newcommand{\bo}{\mathit{b}}
\newcommand{\inseen}[3]{(#1,#2)\in#3}
\newcommand{\notseen}[3]{(#1,#2)\notin#3}
\newcommand{\addseen}[3]{#1\cup\{(#2,#3)\}}
\newcommand{\getpc}[2]{#1~\#~#2}

\newcommand{\ctx}{\mathit{C}}
\newcommand{\fwdctx}{\overrightarrow{\mathit{C}}}
\newcommand{\bwdctx}{\overleftarrow{\mathit{C}}}
\newcommand{\samectx}{\mathit{C}^{=}}
\newcommand{\leafeq}[2]{#1\mathrel{\sim}#2}
\newcommand{\regeq}[2]{#1\mathrel{\approx}#2}
\newcommand{\plug}[2]{#1[#2]}
\newcommand{\leafeqdir}[3]{#1\mathrel{\sim_{#3}}#2}
\newcommand{\leaves}[3]{\mathcal{L}(#1,#2,#3)}
\newcommand{\leaveseq}[2]{#1 \equiv #2}
\newcommand{\lea}{\mathit{leaves}}
\newcommand{\fwdeq}{\sim_{\forward}}
\newcommand{\bwdeq}{\sim_{\backward}}
\newcommand{\botheq}{\sim_{\both}}

\newcommand{\bytecodeinstr}[1]{\texttt{#1}}
\newcommand{\bcone}[2]{\expandafter\ifx\expandafter\relax
  \detokenize{#2}\relax\bytecodeinstr{#1}\else\bytecodeinstr{#1}~#2\fi}
\newcommand{\bctwo}[3]{\expandafter\ifx\expandafter\relax
  \detokenize{#2}\relax\bytecodeinstr{#1}\else\bytecodeinstr{#1}~#2~#3\fi}

\newcommand{\lbl}{\mathit{l}}
\newcommand{\instr}{\mathit{instr}}
\newcommand{\accept}{\bytecodeinstr{Accept}}
\newcommand{\consume}[1]{\bcone{Consume}{#1}}
\newcommand{\jmp}[1]{\bcone{Jump}{#1}}
\newcommand{\fork}[2]{\bctwo{Fork}{#1}{#2}}
\newcommand{\setregopen}[1]{\bcone{RegOpen}{#1}}
\newcommand{\setregclose}[1]{\bcone{RegClose}{#1}}
\newcommand{\resetregs}[1]{\bcone{ResetRegs}{#1}}
\newcommand{\beginloop}{\bytecodeinstr{BeginLoop}}
\newcommand{\iendloop}[1]{\bcone{EndLoop}{#1}}

\newcommand\defequal{~\ensuremath{\raisebox{0pt}{$\overset{\scriptscriptstyle\smash{\Delta}}{=}$}}~} 

\newcommand\seqop[2]{#1 \circ #2}

\newcommand{\inpgt}[3]{#2 <_{#3} #1}

\newcommand{\treethread}[4]{#1 \sim_{#3}^{#4} #2}
\newcommand{\rep}[3]{\mathit{rep}_{#1}~#2~#3}
\newcommand{\seenincl}[3]{#1 \subseteq_{#3} #2}
\newcommand{\app}{\mathbin{++}}

\newcommand{\warblrecompile}[1]{\fn{compile}(#1)}

\newcommand{\wellformed}{\mathcal{W}}
\newcommand{\iswf}[1]{#1 \in \wellformed}

\newcommand{\negskip}{\mskip-.5\thinmuskip}
\newcommand{\tolinden}[1]{\mathopen{\downharpoonleft\negskip}#1\mathclose{\negskip\downharpoonright}}
\newcommand{\towarblre}[1]{\mathopen{\upharpoonleft\negskip}#1\mathclose{\negskip\upharpoonright}}

\newcommand{\jscode}[1]{\lstinline{#1}}
\newcommand{\counterexample}[6]{
  \noindent\textbf{Counter-example of #1:}\\
  #2\hfill
  $\ctx = $ #3 \hfill
  Input: \str{#4}\\
  \textit{JavaScript syntax:}\\
  #5\\
  #6\\
}

\newcommand{\regexample}{\group{1}{\disjunction{\regstar{a}}{a}}b}
\newcommand{\ttree}[1]{\ensuremath{t_{#1}}}
\newcommand{\tnode}[1]{\ttree{#1}$:~$}
\newcommand{\pike}[4]{\ensuremath{[#1]} & \ensuremath{[#2]} & \ensuremath{[#3]} & \ensuremath{[#4]}\\}

\newcommand{\annotation}[1]{{\color{ACMPurple}#1}}

\newcommand{\tskip}{{\tiny skip}}

\definecolor{ex1}{RGB}{201,241,255} 
\definecolor{ex2}{RGB}{255,241,180} 
\definecolor{ex3}{RGB}{198,248,205} 
\definecolor{ex4}{RGB}{255,230,251} 

\newcommand{\reghl}[2]{\colorbox{#1}{\ensuremath{\mathrm{\vphantom{\mid}#2}}}}

\pgfdeclarelayer{background}
\pgfdeclarelayer{foreground}
\pgfsetlayers{background,foreground}

\newtheorem{theorem}{Theorem}

\lstdefinelanguage{JavaScript}{
  keywords={typeof, new, true, false, catch, function, return, null, catch, switch, var, if, in, while, do, else, case, break},
  keywordstyle=\color{blue}\bfseries,
  ndkeywords={class, export, boolean, throw, implements, import, this, match},
  ndkeywordstyle=\color{darkgray}\bfseries,
  identifierstyle=\color{black},
  sensitive=false,
  comment=[l]{//},
  morecomment=[s]{/*}{*/},
  commentstyle=\color{purple}\ttfamily,
  stringstyle=\color{purple}\ttfamily,
  morestring=[b]',
  morestring=[b]",
}

\begin{document}

\title[Formal Verification for JavaScript Regular Expressions]
      {Formal Verification for JavaScript Regular Expressions:\\
        A Proven Mechanized Semantics and Its Applications (Extended Version)}
\pdfstringdefDisableCommands{\def\\#1{ #1}}



\author{Aurèle Barrière}
\orcid{https://orcid.org/0000-0002-7297-2170}
\email{aurele.barriere@epfl.ch}
\affiliation{%
  \institution{EPFL}
  \city{Lausanne}
  \country{Switzerland}
}
\author{Victor Deng}
\orcid{https://orcid.org/0000-0002-7871-0147}
\email{victor.deng@epfl.ch}
\affiliation{%
  \institution{EPFL}
  \city{Lausanne}
  \country{Switzerland}
}
\affiliation{%
  \institution{DI ENS, École Normale Supérieure, PSL, CNRS}
  \city{Paris}
  \country{France}
}
\author{Clément Pit-Claudel}
\orcid{https://orcid.org/0000-0002-1900-3901}
\email{clement.pit-claudel@epfl.ch}
\affiliation{%
  \institution{EPFL}
  \city{Lausanne}
  \country{Switzerland}
}


\begin{abstract}
  We present the first mechanized, succinct, practical, complete, and proven-faithful semantics for a modern regular expression language with backtracking semantics.
  We ensure its faithfulness by proving it equivalent to a preexisting line-by-line embedding of the official ECMAScript specification of JavaScript regular expressions.
  We demonstrate its practicality by presenting two real-world applications.
  First, a new notion of contextual equivalence for modern regular expressions, which we use to prove or disprove rewrites drawn from previous work.
  Second, the first formal proof of the PikeVM algorithm used in many real-world engines.
  In contrast with the specification and other formalization work,
  our semantics captures not only the top-priority match, but a full \textit{backtracking tree} recording all possible matches and their respective priority.
  All our definitions and results have been mechanized in the Rocq proof assistant.
\end{abstract}




\maketitle

\section{Introduction}
\label{sec:intro}

Despite sharing a name and some features, the modern regular expressions found in many programming languages and libraries are fundamentally different from traditional regular expressions.
For one, with the addition of new features such as backreferences, modern regular expressions (which we call \textit{regexes} in the rest of this paper) no longer correspond to regular languages, nor even to context-free languages~\cite{backref_expr}.
More importantly, with features such as capture groups, the matching problem changes.
It no longer amounts to a recognition problem (checking that a string belongs to the language of a regular expression), but instead to a segmentation problem (extracting parts of a string that match individual subexpressions of a regex).
For instance, matching the regex \texttt{a(b)} on string \str{abc} in JavaScript returns that the whole regex matched substring \str{ab}, and that the subexpression in parentheses matched substring \str{b}.
From these features arises a new problem: when there are several ways to match a regex, which substring should be returned?
One possible answer is given by the POSIX standard, favoring the \textit{leftmost-longest} match of the regex.
Another more popular answer is to prioritize the match that would be found by a backtracking algorithm.
Many modern regex languages, such as the ones of Perl, JavaScript, Python, Go, Rust, Java, or .NET follow that disambiguation policy, called \textit{backtracking semantics} (sometimes also called leftmost-greedy, or PCRE semantics~\cite{pcre}).

Between the new features and this disambiguation policy, the semantics of these modern regexes has become surprisingly large and complex.
ECMAScript~\cite{ecma_2025}, the official JavaScript standard, dedicates more than 30 pages to the semantics of JavaScript regexes.
This specification is so complex that semantics bugs are still regularly found in regex engines~\cite{incorrect_result_v8,lookbehind_priority_v8}; that the JavaScript implementation used in Firefox gave up on writing a regex engine and uses Google Chrome's instead~\cite{spidermonkey_irregexp}; and in this work we exhibit semantics errors in the optimizer of a real-world JavaScript regex processor, \texttt{regexp-tree}~\cite{regexp_tree}.

In this work, we present a new formal semantics for JavaScript regexes, that is \textbf{mechanized, succinct, complete, practical, and proven} to be correct.
Despite decades of modern regexes being in wide use (including in more than 30\% of JavaScript and Python packages~\cite{redos_impact}), and multiple previous formalization efforts~\cite{expose,regex_repair,regex_repair2,psst,lean_lookarounds,verified_lookarounds,warblre_icfp}, this is the first time that a semantics enables practical mechanized verification for a real-world modern regex language.
Our semantics is \textbf{mechanized} in the Rocq proof assistant (formerly known as Coq).
It is \textbf{proven} to be faithful to the JavaScript semantics with a verified equivalence to Warblre~\cite{warblre_icfp}, an existing line-by-line embedding of the ECMAScript regex specification into Rocq that is easily auditable but verbose and impractical.
It is \textbf{succinct}: the core rules of matching fit in a single page.
It is \textbf{complete}: it supports all regex features of the 14\textsuperscript{th} edition of ECMAScript.
Finally, it is \textbf{practical} for formal verification.  We demonstrate this with two applications.
First, we present a new formalized definition of contextual equivalence for regexes with backtracking semantics.
Second, we verify the PikeVM algorithm, which simulates a NFA (\textit{non-deterministic finite automaton}).
This is the first time that this widely used matching algorithm is formally verified in full, i.e. proven to return the top priority match according to backtracking semantics. 
Interestingly, beyond the usual considerations of mechanized semantics design (convenient induction principles, strict positive occurrences etc.), we find that a key ingredient for a practical regex backtracking semantics is to not follow a backtracking algorithm.  Instead, we formalize not only the one match returned by a backtracking algorithm, but instead all possible matches, ordered by priority.
We claim the following novel contributions:
\begin{itemize}
\item A new succinct, complete, practical \textbf{formal semantics}, mechanized in Rocq.
\item Evidence that this semantics is \textbf{faithful} to the ECMAScript regex standard, by proving it equivalent to the Warblre embedding~\cite{warblre_icfp}.
\item A new \textbf{formal definition of equivalence} for JavaScript regexes. We use it to prove the correctness of regex equations from the literature, as well as prove and disprove optimizations used in a real-world JavaScript regex processor.
\item The first \textbf{formal verification of the PikeVM algorithm}, a linear-time regex matching algorithm used in many real-world regex engines. This is the first time that NFA simulation with capture groups and backtracking semantics comes with a formal proof.
\end{itemize}
The theorem relating our semantics model to the official ECMAScript standard makes this the first formal verification work directly connected to a real-world regex specification.

\Autoref{sec:background} informally introduces JavaScript regexes.
\Autoref{sec:semantics} presents our new inductive semantics and its properties.
\Autoref{sec:warblre_equiv} shows its faithfulness by proving it equivalent to the Warblre embedding.
We then present two formally verified applications.
First, \autoref{sec:rewrite} defines our new regex equivalence and our case studies of correct and incorrect regex rewrites.
Second, \autoref{sec:pikevm} presents our formal verification of the PikeVM algorithm.
Finally, \autoref{sec:related} discusses adapting and extending our work to other languages and applications, as well as related and future work.

Our Rocq development, consisting of over 15k lines of definitions and proofs, is available as supplementary material.
\appendixref{app:correspondence} links each individual paper definition to its Rocq counterpart.
Some definitions have been simplified in the paper for presentation purposes.

\section{Background on JavaScript Regexes}
\label{sec:background}

\subsection{Syntax}

The abstract syntax of JavaScript regexes we use in this work is summarized on \Autoref{fig:syntax}, with some modifications from concrete JavaScript syntax for presentation and convenience purposes.
These differences are as follows.
We use \re{\regepsilon} instead of no string at all for the empty regex.
In concrete JavaScript syntax, capture groups are pairs of parentheses that are implicitly associated with an index depending on their position in the regex. Here, we directly annotate each group with its index, $\gid$. For instance, \texttt{(a(b))} in JavaScript syntax corresponds to \re{\group{1}{a\group{2}{b}}}.
JavaScript also uses syntactic sugar \texttt{(?<\jsvar{name}>\re{\subreg})} (\textit{named group}) and \texttt{\textbackslash k<\jsvar{name}>} (\textit{named backreference}) to refer to groups using names instead of indices. Group names can refer to one group at most, so there is a direct correspondence from names to indices.
To group subexpressions together without defining a capture group, JavaScript uses non-capturing groups, denoted as \texttt{(?:$\subreg$)} and for which we use the notation \re{\noncap{\subreg}}.
These non-capturing groups are here to help parsing, and are not reflected in the AST.
Disjunction is right-associative and has the lowest parsing precedence, meaning for instance that \re{\disjunction{a}{\disjunction{ab}{abc}}} is exactly the same regex as \re{\disjunction{a}{\noncap{\disjunction{\noncap{ab}}{\noncap{abc}}}}}.

Quantifiers in concrete JavaScript syntax are written \texttt{\jsvar{r}\{\jsvar{min},\jsvar{max}\}} where \jsvar{max} is either a natural number or infinity.
This represents any number of repetitions of \jsvar{r} between \jsvar{min} and \jsvar{max}.
Since every valid regex quantifier in JavaScript has $\jsvar{min}\leq\jsvar{max}$, we instead represent quantifier upper bounds with a $\Delta$ representing \texttt{\jsvar{max}-\jsvar{min}}, simplifying the semantics.
Quantifiers can either be greedy or lazy (the latter by appending a \texttt{?} in concrete syntax); we represent this with a $\top$ for greedy or a $\bot$ for lazy.
JavaScript also defines useful syntactic sugar for some common quantifiers:
\re{\regstar{\subreg}} for \re{\quant{\subreg}{0}{\infty}{\top}},
\re{\lazystar{\subreg}} for \re{\quant{\subreg}{0}{\infty}{\bot}},
\re{\regplus{\subreg}} for \re{\quant{\subreg}{1}{\infty}{\top}},
\re{\lazyplus{\subreg}} for \re{\quant{\subreg}{1}{\infty}{\bot}},
\re{\regqm{\subreg}} for \re{\quant{\subreg}{0}{1}{\top}},
and \re{\lazyqm{\subreg}} for \re{\quant{\subreg}{0}{1}{\bot}}.

This syntax covers all regexes from the 14\textsuperscript{th} edition of ECMAScript~\cite{ecma_2023}.

\begin{figure}
  \raggedright
  \begin{minipage}[t]{.5\textwidth}
  Regular Expressions:\\
  \begin{tabular}{l r l l}
    \re{\subreg} & $::=$ & \re{\regepsilon} & Empty \\
    &$\mid$& \re{\regchar{\cd}} & Character \\
    &$\mid$& \re{\disjunction{\subreg_1}{\subreg_2}} & Disjunction \\
    &$\mid$& \re{\sequence{\subreg_1}{\subreg_2}} & Sequence \\
    &$\mid$& \re{\group{\gid}{\subreg}} & Capture group \\
    &$\mid$& \re{\noncap{\subreg}} & Non-capturing group \\
    &$\mid$& \re{\anchor{\anc}} & Anchor \\
    &$\mid$& \re{\backref{\gid}} & Backreference \\
    &$\mid$& \re{\quant{\subreg}{\rmin}{\Delta}{\top}} & Greedy Quantifier \\
    &$\mid$& \re{\quant{\subreg}{\rmin}{\Delta}{\bot}} & Lazy Quantifier \\
    &$\mid$& \re{\lookaround{\lk}{\subreg}} & Lookaround \\
  \end{tabular}\\
  \end{minipage}%
  \begin{minipage}[t]{.5\textwidth}
  
  Character descriptors:\\
  \begin{tabular}{l r l}
    $\cd$ & $::=$ & \re{\rchar} ~$\mid$~ \re{\regdot} $\mid$ \re{\charset{\rchar_0\cdrange{\rchar_1}{\rchar_2}}} ~$\mid$~ \re{\charset{\cdcomplement{}\rchar_0\cdrange{\rchar_1}{\rchar_2}}} ~$\mid$~\\
    & & \re{\esc{w}} ~$\mid$~ \re{\esc{W}} ~$\mid$~ \re{\esc{d}} ~$\mid$~ \re{\esc{D}} ~$\mid$~
    \re{\esc{s}} ~$\mid$~ \re{\esc{S}} ~$\mid$~ \dots
  \end{tabular}\\
  
  Anchors:\\
  \begin{tabular}{l r l c l l}
    $\anc$ & $::=$ & \re{\regbos} &$\mid$& \re{\regeos} & Input boundary\\
    &$\mid$& \re{\esc{b}} &$\mid$& \re{\esc{B}} & Word boundary\\
  \end{tabular}\\
  
  Lookarounds:\\
  \begin{tabular}{l r l l}
    $\lk$ & $::=$ & \re{\posahead} & Positive lookahead\\
    &$\mid$& \re{\negahead} & Negative lookahead\\
    &$\mid$& \re{\posbehind} & Positive lookbehind\\
    &$\mid$& \re{\negbehind} & Negative lookbehind\\
  \end{tabular}\\
  \end{minipage}
  \caption{Abstract JavaScript regex syntax}
  \Description{} 
  \label{fig:syntax}
\end{figure}

\subsection{Informal Semantics}

\newcommand{\feature}[1]{\tightparagraph{#1.}}
We summarize here how JavaScript regexes differ from traditional regular expressions.
\feature{Character descriptors}
There are many ways to describe a set of characters in JavaScript regexes: the following list is not exhaustive.
One can write the character directly: any character \re{\rchar} matches itself.
Then, \re{\regdot} matches all characters (except line terminator characters depending on the \texttt{s} flag, described below).
Character classes match several ranges of characters. For instance \re{\charset{a\cdrange{f}{z}}} matches either \re{a}, or any letter in between \re{f} and \re{z}.
They can also be negated, for instance \re{\charset{\cdcomplement{}a\cdrange{f}{z}}} matches any character that cannot be matched by the previous example.
Finally, character class escapes match common sets of characters, for instance \re{\esc{d}} matches numerical digits.
\feature{Capture groups}
Capture groups record the last match of each subexpression inside parentheses.
For instance, matching \re{a\group{1}{\regstar{\regdot}}c} on the string \str{abcd} returns a match on the substring \str{abc}, and captures the substring \str{b} between indices 1 and 2 in group 1.
\feature{Backreferences}
A backreference \re{\backref{\gid}} matches the content of capture group $\gid$ again.
For instance, the regex \re{\group{1}{\esc{d}}\backref{1}} matches any digit repeated twice, like \str{33}. 
\feature{Anchors}
Anchors enforce a local condition on the current input position, for instance checking for being at the end of the input (with \re{\regeos}), or being at a word boundary (with \re{\esc{b}}).
For instance, \re{a\regeos} does not match anything in \str{ab} but matches the substring \str{a} in \str{ba}.
\feature{Lookarounds}
Lookarounds describe a zero-width condition described as a regex.
For instance, the regex \re{a\lookahead{b}} matches any \str{a} only if it is followed by a \str{b}, but the \str{b} itself is not part of the returned substring.
Lookarounds can be either positive or negative; negative lookarounds require that there is no match for the subexpression at the current input position.
Positive lookarounds can define capture groups. For instance, matching \re{a\lookahead{\group{1}{\disjunction{b}{c}}}} on string \str{ab} matches the substring \str{a} and sets capture group 1 to \str{b}.
Lookbehinds allow to match the condition backwards.
For instance, matching \re{\lookbehind{Year\mathord{:}}\esc{d}+} on string \str{Year:2026} matches the substring \str{2026}.
\feature{Backtracking semantics}
As ECMAScript uses a backtracking algorithm to describe their semantics, JavaScript regexes follow backtracking semantics.
As a result, when there are several ways to match a regex, the leftmost match has priority (for instance, matching \re{\disjunction{a}{b}} on \str{dba} returns \str{b}).
Then, the left branch of each disjunction has priority over the right one (for instance, matching \re{\disjunction{a}{ab}} on \str{ab} returns \str{a}).
Finally, greedy quantifiers give priority to a match that has iterated as many times as possible, while lazy ones try to iterate as few times as possible (within what is allowed by the $\rmin$ and $\Delta$ parameters).
\feature{Unanchored and anchored matching}
The kind of regex matching we have discussed so far, where a match can start at any position in the input string (but priority is given to the leftmost-starting match), is also known as \textit{unanchored matching}.
However, the problem of unanchored matching is typically reduced to the problem of \textit{anchored matching}, where given a regex, an input string and an input position, one needs to find the top-priority match of the regex that starts precisely at this position (but may end at any position).
In the ECMAScript specification, unanchored matching is performed by repeatedly trying an anchored matcher at every possible starting position.
In implementations, the top-priority unanchored match of a regex $r$ is sometimes obtained by computing the top-priority anchored match of \re{\lazystar{\regall}\group{0}{\subreg}}, where \re{\regall} is a regex which matches any character~\cite{regexp_vm_approach}. The laziness of the star ensures we find the leftmost-starting match, and the capture group with index 0 extracts the start and end positions of that match.
As a result, our semantics focuses on anchored matches; in the rest of this paper, \textit{matching} refers to anchored matching.
\feature{Top-level API}
JavaScript provides several top-level functions to use regexes: we can find an unanchored match with \texttt{match}, find all matches with \texttt{matchAll}, replace the result of matching in a string with another substring with \texttt{replace} etc.
All of these functions can be implemented using an anchored matcher and common string manipulation operations.
\feature{Regex flags}
Regex flags can be used to change the behavior of matching.
The Unicode flag \texttt{u} changes the underlying alphabet of characters.
Then, the three flags \texttt{i} (case-insensitivity), \texttt{m} (multiline mode for anchors) and \texttt{s} (make \re{\regdot} match line terminators) affect the way some characters descriptors and anchors are matched.
Finally, the other flags affect the behavior of top-level functions and are not directly relevant to our semantics for anchored matching.
The \texttt{g} flag enables global matching (finding all matches of a regex in a string).
With the \textit{sticky} flag \texttt{y}, matching a regex on a string starts from the position where a match of that same regex was previously found.
When the flag \texttt{d} is active, an engine should not only return the substring matched by each capture group, but also the corresponding string indices.

\feature{JavaScript semantic peculiarities}
While the rules above hold for any regex language with backtracking semantics, JavaScript semantics also differs from other such languages in two ways.
First, to prevent infinite repetitions in regexes like \re{\regstar{\regepsilon}}, the \textbf{Nullable Quantifier} property of JavaScript means that, when $\rmin=0$, iterations of a quantifier cannot match the empty string.
Other languages use different criteria to prevent infinite repetitions. This unique property requires adapting standard matching algorithms (see section~\ref{sec:pikevm}).

Second, the \textbf{Capture Reset} property means that at each quantifier iteration, the values of capture groups inside that quantifier are reset.
For instance, matching \re{\regstar{\noncap{\disjunction{\group{1}{a}}{b}}}} on the string \str{ab} matches the whole string, but the value of capture group 1 is set to $\None$ after the match: the regex is iterated twice, the first iteration defines group 1, but the second iteration resets it.

\section{Tree Semantics}
\label{sec:semantics}

\subsection{An Inductive Relation for a Backtracking Execution Trace}

A key insight of our semantics is to depart from the backtracking algorithm used to specify JavaScript regex semantics in ECMAScript.
In contrast, our semantics is an inductive relation that connects a regex and a string to an execution trace of a backtracking algorithm that \textit{would not stop} at the first accepting result.
This execution trace is represented with a tree, where each backtracking decision (for disjunctions and quantifiers) is represented as a branching node, and other nodes correspond to some operations of the algorithm, such as reading a character or a backreference, opening or closing a capture group, or checking the validity of an anchor or lookaround.

\Autoref{fig:tree_example} showcases the backtracking tree of the regex \re{\noncap{\disjunction{a}{\disjunction{a\group{1}{b}}{a}}}bc} on the input \str{abbc}.
The tree starts by branching on two possibilities: either matching the first or the second branch of the disjunction.
On the left, we see the operations corresponding to exploring the first alternative. First, an \str{a} is read, then the disjunction has finished matching. We resume by matching a \str{b}, but then matching fails because in this string it is not possible to read a \str{c} at the current position.
The other two branches correspond to the deepest disjunction, and the middle branch does find a match in the string, opening and closing a capture group along the way.
Observe that the backtracking tree includes the full third branch, even though a backtracking algorithm would stop after reaching the match in the second branch.

\begin{figure}
\centering
\begin{minipage}{.4\textwidth}
  \centering
  \raggedright
  \begin{tikzpicture}[%
          every node/.style={rectangle,minimum size=6pt,minimum height=4pt, inner sep=1pt, align=center},
          node distance=.3cm, >=latex
    ]
    \begin{pgfonlayer}{foreground}
    \node [] (t0) [] {$\treechoice{}{}$};
    \begin{scope}[local bounding box=bb1]
      \node [] (tab) [below left=of t0] {$\treeread{a}{}$};
    \end{scope}
    \begin{scope}[local bounding box=bb2]
      \node [] (tb) [below=of tab] {$\treeread{b}{}$};
      \node [] (tm) [below=of tb] {$\treemismatch$};
    \end{scope}
    \node [] (tchoice) [below right=of t0] {$\treechoice{}{}$};
    \begin{scope}[local bounding box=bb3]
      \node [] (tabbc) [below left=of tchoice, xshift=.5cm] {$\treeread{a}{}$};
      \node [] (topen) [below=of tabbc] {$\treeopen{1}{}$};
      \node [] (tbbc) [below=of topen] {$\treeread{b}{}$};
      \node [] (tclose) [below=of tbbc] {$\treeclose{1}{}$};
    \end{scope}
    \begin{scope}[local bounding box=bb4]
      \node [] (tbc) [below=of tclose] {$\treeread{b}{}$};
      \node [] (tc) [below=of tbc] {$\treeread{c}{}$};
      \node [] (tmatch) [below=of tc] {$\treematch$};
    \end{scope}
    \begin{scope}[local bounding box=bb5]
      \node [] (tab') [below right=of tchoice, xshift=-.5cm] {$\treeread{a}{}$};
    \end{scope}
    \begin{scope}[local bounding box=bb6]
      \node [] (tb') [below=of tab'] {$\treeread{b}{}$};
      \node [] (tm') [below=of tb'] {$\treemismatch$};
    \end{scope}
    \path [draw] (t0) edge[->] node {} (tab);
    \path [draw] (tab) edge[->] node {} (tb);
    \path [draw] (tb) edge[->] node {} (tm);
    \path [draw] (t0) edge[->] node {} (tchoice);
    \path [draw] (tchoice) edge[->] node {} (tabbc);
    \path [draw] (tabbc) edge[->] node {} (topen);
    \path [draw] (topen) edge[->] node {} (tbbc);
    \path [draw] (tbbc) edge[->] node {} (tclose);
    \path [draw] (tclose) edge[->] node {} (tbc);
    \path [draw] (tbc) edge[->] node {} (tc);
    \path [draw] (tc) edge[->] node {} (tmatch);
    \path [draw] (tchoice) edge[->] node {} (tab');
    \path [draw] (tab') edge[->] node {} (tb');
    \path [draw] (tb') edge[->] node {} (tm');
    \end{pgfonlayer}
    \begin{pgfonlayer}{background}
      \fill [ex1] (bb2.north west) rectangle (bb2.south east);
      \fill [ex1] (bb4.north west) rectangle (bb4.south east);
      \fill [ex3] (bb5.north west) rectangle (bb5.south east);
      \fill [ex4] (bb3.north west) rectangle (bb3.south east);
      \fill [ex2] (bb1.north west) rectangle (bb1.south east);
      \fill [ex1] (bb6.north west) rectangle (bb6.south east);
    \end{pgfonlayer}
      \end{tikzpicture}
  \savere{aaba}{\re{\noncap{\disjunction{\reghl{ex2}{a}}{\noncap{\disjunction{\reghl{ex4}{a\group{1}{b}}}{\reghl{ex3}{a}}}}}\reghl{ex1}{bc}}}
  \captionof{figure}{Backtracking tree \(\treecont\) of the regex\\\usere{aaba} on input \str{abbc}}
  \Description{The root of the tree is a binary `Choice' node.  To the left is a branch with unary `Read a', `Read b', and `Mismatch' nodes.  To the right is another `Choice' node, with `Read a', `Open 1', `Read b', `Close 1', `Read b', `Read c', `Match' on the left and `Read a', `Read b', `Mismatch' on the right.}
  \label{fig:tree_example}
\end{minipage}%
\begin{minipage}{.6\textwidth}
  \centering
  \raggedright
  Tree:\\
  \begin{tabular}{l r l l}
    $\treecont$ & $::=$ & $\treematch$ & Successful match\\
    &$\mid$& $\treemismatch$ & Match failure\\
    &$\mid$& $\treechoice{\treecont_1}{\treecont_2}$ & Branching\\
    &$\mid$& $\treeread{\rchar}{\treecont}$ & Character read success\\
    &$\mid$& $\treebackref{\strbr}{\treecont}$ & Backreference success\\
    &$\mid$& $\treeprogress{\treecont}$ & Progress check success\\
    &$\mid$& $\treeanchor{\anc}{\treecont}$ & Anchor success\\
    &$\mid$& $\treeopen{\gid}{\treecont}$ & Group opening\\ 
    &$\mid$& $\treeclose{\gid}{\treecont}$ & Group closing\\ 
    &$\mid$& $\treereset{\gidl}{\treecont}$ & Group resetting\\ 
    &$\mid$& $\treelk{\lk}{\treelook}{\treecont}$ & Lookaround success\\
    &$\mid$& $\treelkfail{\lk}{\treelook}$ & Lookaround failure\\
  \end{tabular}\\
  \captionof{figure}{Backtracking Trees}
  \Description{} 
  \label{fig:tree_syntax}
\end{minipage}%
\end{figure}

The syntax of these backtracking trees is defined in \Autoref{fig:tree_syntax}.
A $\treemismatch$ node means that the algorithm has failed to find a match, for instance because the next character in the input does not correspond to the character descriptor found in the regex, or because some anchor was not satisfied, or because a progress check failed, etc.
For $\treechoice{}{}$ nodes, the subtrees are ordered by priority, meaning that the result of a backtracking algorithm corresponds to the leftmost $\treematch$ leaf of the corresponding tree.
Finally, nodes $\treelk{\lk}{\treelook}{\treecont}$ and $\treelkfail{\lk}{\treelook}$ include the entire tree $\treelook$ corresponding to matching the lookaround $\lk$ at the current input position.

\begin{figure}
  \mbox{\infer[\ruledef{tree}{Match}]{\istree{[]}{\inp}{\gm}{\dir}{\treematch}}{}}\hfill%
  \mbox{\infer[\ruledef{tree}{Close}]{\istree{\aclose{\gid} :: \cont}{\inp}{\gm}{\dir}{\treeclose{\gid}{\treecont}}}
    {\istree{\cont}{\inp}{\gmclose{\gm}{\gid}{\idx{\inp}}}{\dir}{\treecont}}}%
  
  \semspace
  \mbox{\infer[\ruledef{tree}{Check}]{\istree{\acheck{\inpcheck} :: \cont}{\inp}{\gm}{\dir}{\treeprogress{\treecont}}}
    {\inpgt{\inp}{\inpcheck}{\dir} \newpremise \istree{\cont}{\inp}{\gm}{\dir}{\treecont}}}\hfill%
  \mbox{\infer[\ruledef{tree}{CheckFail}]{\istree{\acheck{\inpcheck} :: \cont}{\inp}{\gm}{\dir}{\treemismatch}}
    {\neg(\inpgt{\inp}{\inpcheck}{\dir})}}%

  \semspace
  \mbox{\infer[\ruledef{tree}{Read}]{\istree{\areg{\re{\regchar{\cd}}} :: \cont}{\inp}{\gm}{\dir}{\treeread{\rchar}{\treecont}}}
    {\begin{array}{c}
        \inpadvance{\cd}{\inp}{\dir} = \Some{(\rchar, \inp')}\\
        \istree{\cont}{\inp'}{\gm}{\dir}{\treecont}
      \end{array}}}\hfill%
  \mbox{\infer[\ruledef{tree}{ReadFail}]{\istree{\areg{\re{\regchar{\cd}}} :: \cont}{\inp}{\gm}{\dir}{\treemismatch}}
    {\inpadvance{\cd}{\inp}{\dir} = \None}}%

  \semspace
  \mbox{\infer[\ruledef{tree}{Disj}]{\istree{\areg{\re{(\disjunction{\subreg_1}{\subreg_2})}} :: \cont}{\inp}{\gm}{\dir}{\treechoice{\treecont_1}{\treecont_2}}}
    {\istree{\areg{\re{\subreg_1}} :: \cont}{\inp}{\gm}{\dir}{\treecont_1} \newpremise \istree{\areg{\re{\subreg_2}} :: \cont}{\inp}{\gm}{\dir}{\treecont_2}}}%

  \semspace
  \mbox{\infer[\ruledef{tree}{SeqForward}]{\istree{\areg{\re{\sequence{\subreg_1}{\subreg_2}}} :: \cont}{\inp}{\gm}{\forward}{\treecont}}
    {\istree{\areg{\re{\subreg_1}} :: \areg{\re{\subreg_2}} :: \cont}{\inp}{\gm}{\forward}{\treecont}}}\hfill%
  \mbox{\infer[\ruledef{tree}{SeqBackward}]{\istree{\areg{\re{\sequence{\subreg_1}{\subreg_2}}} :: \cont}{\inp}{\gm}{\backward}{\treecont}}
    {\istree{\areg{\re{\subreg_2}} :: \areg{\re{\subreg_1}} :: \cont}{\inp}{\gm}{\backward}{\treecont}}}%

  \semspace
  \mbox{\infer[\ruledef{tree}{Epsilon}]{\istree{\areg{\re{\regepsilon}} :: \cont}{\inp}{\gm}{\dir}{\treecont}}
    {\istree{\cont}{\inp}{\gm}{\dir}{\treecont}}}\hfill%
  \mbox{\infer[\ruledef{tree}{Group}]{\istree{\areg{\re{\group{\gid}{\subreg}}} :: \cont}{\inp}{\gm}{\dir}{\treeopen{\gid}{\treecont}}}
    {\istree{\areg{\re{\subreg}} :: \aclose{\gid} :: \cont}{\inp}{\gmopen{\gm}{\gid}{\idx{\inp}}}{\dir}{\treecont}}}%
      
  \semspace
  \mbox{\infer[\ruledef{tree}{Anchor}]{\istree{\areg{\re{\anchor{\anc}}} :: \cont}{\inp}{\gm}{\dir}{\treeanchor{\anc}{\treecont}}}
    {\begin{array}{c}
        \checkanchor{\anc}{\inp} = \top\\
        \istree{\cont}{\inp}{\gm}{\dir}{\treecont}
      \end{array}}}\hfill%
  \mbox{\infer[\ruledef{tree}{AnchorFail}]{\istree{\areg{\re{\anchor{\anc}}} :: \cont}{\inp}{\gm}{\dir}{\treemismatch}}
    {\checkanchor{\anc}{\inp} = \bot}}%

  \semspace
  \mbox{\infer[\ruledef{tree}{Backref}]{\istree{\areg{\re{\backref{\gid}}} :: \cont}{\inp}{\gm}{\dir}{\treebackref{\strbr}{\treecont}}}
    {\begin{array}{c}
        \readbackref{\gm}{\gid}{\inp}{\dir} = \Some{(\strbr, \inp')} \\
        \istree{\cont}{\inp'}{\gm}{\dir}{\treecont}
      \end{array}}}\hfill%
  \mbox{\infer[\ruledef{tree}{BackrefFail}]{\istree{\areg{\re{\backref{\gid}}} :: \cont}{\inp}{\gm}{\dir}{\treemismatch}}
    {\readbackref{\gm}{\gid}{\inp}{\dir} = \None}}%

  \semspace
  \mbox{\infer[\ruledef{tree}{Forced}]{\istree{\areg{\re{\quant{\subreg}{\rmin+1}{\Delta}{\greedy}}} :: \cont}{\inp}{\gm}{\dir}{\treereset{\defgroups{\subreg}}{\treecont}}}
    {\istree{\areg{\re{\subreg}} :: \areg{\re{\quant{\subreg}{\rmin}{\Delta}{\greedy}}} :: \cont}{\inp}{\gmreset{\gm}{\defgroups{\subreg}}}{\dir}{\treecont}}}\hfill%
  \mbox{\infer[\ruledef{tree}{Done}]{\istree{\areg{\re{\quant{\subreg}{0}{0}{\greedy}}} :: \cont}{\inp}{\gm}{\dir}{\treecont}}
    {\istree{\cont}{\inp}{\gm}{\dir}{\treecont}}}%

  \semspace
  \mbox{\infer[\ruledef{tree}{Greedy}]{\istree{\areg{\re{\quant{\subreg}{0}{\Delta+1}{\top}}} :: \cont}{\inp}{\gm}{\dir}{\treechoice{(\treereset{\defgroups{\subreg}}{\treeiter})}{\treecont_{skip}}}}
    {\istree{\cont}{\inp}{\gm}{\dir}{\treecont_{skip}} \newpremise \istree{\areg{\re{\subreg}} :: \acheck{\inp} :: \areg{\re{\quant{\subreg}{0}{\Delta}{\top}}} :: \cont}{\inp}{\gmreset{\gm}{\defgroups{\subreg}}}{\dir}{\treeiter}}}%

  \semspace
  \mbox{\infer[\ruledef{tree}{Lazy}]{\istree{\areg{\re{\quant{\subreg}{0}{\Delta+1}{\bot}}} :: \cont}{\inp}{\gm}{\dir}{\treechoice{\treecont_{skip}}{(\treereset{\defgroups{\subreg}}{\treeiter})}}}
    {\istree{\cont}{\inp}{\gm}{\dir}{\treecont_{skip}} \newpremise \istree{\areg{\re{\subreg}} :: \acheck{\inp} :: \areg{\re{\quant{\subreg}{0}{\Delta}{\bot}}} :: \cont}{\inp}{\gmreset{\gm}{\defgroups{\subreg}}}{\dir}{\treeiter}}}%

  \semspace
  \mbox{\infer[\ruledef{tree}{Lookaround}]{\istree{\areg{\re{\lookaround{\lk}{\subreg}}} :: \cont}{\inp}{\gm}{\dir}{\treelk{\lk}{\treelook}{\treecont}}}
    {\begin{array}{c}
        \lkdir{\lk} = ~\dir' \newpremise \lkresult{\lk}{\treelook}{\gm}{\idx{\inp}} = \Some{\gm'} \\
        \istree{[\areg{\re{\subreg}}]}{\inp}{\gm}{\dir'}{\treelook}
        \newpremise \istree{\cont}{\inp}{\gm'}{\dir}{\treecont}
      \end{array}}}%

  \semspace
  \mbox{\infer[\ruledef{tree}{LookaroundFail}]{\istree{\areg{\re{\lookaround{\lk}{\subreg}}} :: \cont}{\inp}{\gm}{\dir}{\treelkfail{\lk}{\treelook}}}
    {\begin{array}{c}
        \lkdir{\lk} = ~\dir' \newpremise \lkresult{\lk}{\treelook}{\gm}{\idx{\inp}} = \None \\
        \istree{[\areg{\re{\subreg}}]}{\inp}{\gm}{\dir'}{\treelook}
      \end{array}}}%

\caption{Inductive tree semantics}
\Description{} 
\label{fig:tree_semantics}
\end{figure}

\paragraph{Advantages of our semantics}
The core of our semantics comprises only 21 rules shown in \Autoref{fig:tree_semantics}, achieving a convenient succinctness compared to the ECMAScript specification, while providing a complete formalization of JavaScript regexes.
To match individual characters, we reuse some definitions from the Warblre mechanization, but reasoning on the core matching algorithm can be done without unfolding them.
The rules are straightforward to mechanize, in particular we materialize the full tree of lookarounds and scan it for a match, thereby avoiding workarounds for non-strict positivity that are needed in other work (see \autoref{sec:related}).
Our semantics also supports the relevant regex flags: although we omit them from the paper definitions, our Rocq definitions are parameterized by the value of the flags \texttt{i}, \texttt{m} and \texttt{s}.
To support the Unicode mode (flag \texttt{u}), our semantics is parameterized by an alphabet of characters, like Warblre.
Although we focus on anchored matching, our semantics could also support the flags used for top-level functions.
Essentially, the flags \texttt{g} and \texttt{y} require changing the inital input position, and we already compute the group indices required by the \texttt{d} flag.

More importantly, our semantics also formalizes all possible matches of a regex on an input, not just the top-priority one.
Later, we show that this is not only helpful to verify algorithms that explore more paths than a backtracking algorithm (\Autoref{sec:pikevm}), but also to reason about contextual equivalence in backtracking semantics (\Autoref{sec:rewrite}).

\paragraph{Semantics statement}
\Autoref{fig:tree_semantics} presents the rules of our semantics.
This relation naturally provides an induction principle that we use extensively in our mechanized proofs. 
We define the big-step relation $\istree{\cont}{\inp}{\gm}{\dir}{\treecont}$, to mean that $\treecont$ is the backtracking tree for the list of actions $\cont$, the input $\inp$, the group map $\gm$ and the reading direction $\dir$.

With lookbehinds, it is possible to read the input in reverse.
As a result, there are two possible reading directions $\dir$: forward ($\forward$) and backward ($\backward$).
To support bidirectional reading, the input string $\inp$ is represented with a zipper (a pair of the list of next characters to consume, and the reversed list of characters already consumed). The relation $\inpgt{\inp_1}{\inp_2}{\dir}$ on inputs means that input $\inp_1$ is the same as $\inp_2$ after having read one or more characters following the direction $\dir$.
We define $\idx{\inp}$ to be the number of characters already read in $\inp$, and $\nextinpdir{\inp}{\dir}$ corresponds to $\inp$ after reading one character.

The list $\cont$ represents a stack of remaining actions, each of which can take one of three shapes.
First, regexes are actions: initially for a regex $\subreg$, the list of actions is simply $[\areg{\re{\subreg}}]$.
Second, $\aclose{\gid}$ is the action that closes the group $\gid$.
For instance, in \Autoref{fig:tree_example}, after reading the first character $b$ in the middle branch, one needs to first close group 1 before matching the rest of the regex.
Third, $\acheck{\inp}$ means that we need to check that the current input has made some progress compared to the input $\inp$; this is used in quantifier semantics to check that each optional iteration has not matched the empty string.

The values of all capture groups are stored in a group map, $\gm$, associating to each capture group index $\gid$ the optional range that the group last matched.
In the initial group map $\gmempty$, none of the groups are defined.
The function $\gmopen{\gm}{\gid}{n}$ updates the group map $\gm$ to record that capture group $\gid$ is open at position $n$.
Then, $\gmclose{\gm}{\gid}{n}$ and $\gmreset{\gm}{\gidl}$ respectively close the group $\gid$ and reset the value of each group in the list $\gidl$.

\paragraph{Explanations of rules}
We discuss here the most interesting rules of the semantics.
To pass a progress check (\ruleref{tree}{Check}), the semantics ensures that some progress has been made in the current input $\inp$ compared to the input $\inpcheck$ at which the progress check was generated.
For character descriptors (\ruleref{tree}{Read}), the semantics checks with the function $\inpadvance{\cd}{\inp}{\dir}$ that the next character $\rchar$ (for direction $\dir$) in the input $\inp$ can be matched by the character descriptor $\cd$, and returns the next input zipper $\inp'$.
In contrast with traditional regular expressions, matching the sequence $\subreg_1\subreg_2$ does not mean that we can split the input into disjoint parts (one that matches $\subreg_1$ and the other that matches $\subreg_2$): a lookahead in $\subreg_1$ might explore the same part of the input that is matched by $\subreg_2$. In \ruleref{tree}{SeqForward}, we instead match $\subreg_1$ at the current input, and add $\subreg_2$ to the list of actions to match it afterwards.
Handling the sequence also depends on the current direction of matching (\ruleref{tree}{SeqBackward}).
For anchors (\ruleref{tree}{Anchor}), $\checkanchor{\anc}{\inp}$ returns $\top$ when the input $\inp$ satisfies the anchor $\anc$ (for instance, when the input is in the final position and the anchor is $\regeos$), and $\bot$ otherwise.
For backreferences (\ruleref{tree}{Backref}), the function $\readbackref{\gm}{\gid}{\inp}{\dir}$ looks up the value of group $\gid$ in $\gm$ and checks that the next characters of $\inp$ match. When the group $\gid$ is not defined, matching the backreference corresponds to reading the empty substring.

The rules for quantifiers come in three shapes.
When a quantifier has forced repetitions ($\rmin>0$, \ruleref{tree}{Forced}), the corresponding tree consists in resetting the groups inside the subregex, then matching the regex followed by the quantifier where the $\rmin$ has decreased. The function $\defgroups{\subreg}$ returns the list of groups defined in $\subreg$.
When both the $\rmin$ and the $\Delta$ of a quantifier are 0, the quantifier is forced to be skipped (\ruleref{tree}{Done}).
Finally, when $\rmin=0$ and $\Delta>0$, there is a $\treechoice{}{}$ in the tree, between doing one more iteration or skipping the quantifier.
For a greedy quantifier (\ruleref{tree}{Greedy}), the extra iteration has higher priority, while for a lazy quantifier (\ruleref{tree}{Lazy}), skipping has higher priority.

For lookarounds (\ruleref{tree}{Lookaround}) the semantics considers a tree $\treelook$ of the regex inside the lookaround from the current input position and the direction of $\lk$ ($\forward$ for lookaheads, $\backward$ for lookbehinds).
Then, the function $\lkresult{\lk}{\treelook}{\gm}{i}$ checks that this tree $\treelook$ contains an accepting branch if the lookaround is positive, in which case it returns an updated group map with the groups defined in the first accepting branch of that tree.
If the lookaround is negative, the function checks that the tree has no accepting branch, and then returns the group map unchanged.

\paragraph{Semantics result}
We call \textit{accepting branches} the branches of a tree that end with a $\treematch$.
The \textit{leaf} of an accepting branch (a final input position and a final group map), is obtained by replaying the group operations in the branch.
The result of matching a regex on a string corresponds to the leaf of the first accepting branch in the corresponding tree.
When there is no such accepting branch, there is no match for the regex on the string.
We define the function $\firstbranch{\treecont}{\inp}$ to return either the first leaf of $\treecont$ for input $\inp$, or $\None$.\\
In \Autoref{fig:tree_example}, $\firstbranch{\treecont}{\sinput [\texttt{a};\texttt{b};\texttt{b};\texttt{c}] []} = \Some{(\sinput [] [\texttt{c};\texttt{b};\texttt{b};\texttt{a}], \gmclose{\gmopen{\gmempty}{1}{1}}{1}{2})}$,
where $\sinput~l_1~l_2$ denotes an input zipper, with $l_1$ being the list of next characters and $l_2$ the list of previously seen characters.

\subsection{Semantics Properties}
\label{subsec:props}
In this section, we show that for any list of actions, input, group map and direction, there exists a unique backtracking tree.
We first prove that our semantics is deterministic.

\begin{theorem}[Determinism]
  \label{thm:tree_det}
  $\forall\: \actions,\: \inp,\: \gm,\: \dir,\: \treecont_1,\: \treecont_2.\;$
  $\istree{\actions}{\inp}{\gm}{\dir}{\treecont_1} \wedge \istree{\actions}{\inp}{\gm}{\dir}{\treecont_2} \implies \treecont_1 = \treecont_2$.
\end{theorem}
\begin{proof}
  By induction over a derivation of $\istree{\actions}{\inp}{\gm}{\dir}{\treecont_1}$.
\end{proof}
We then prove its productivity.
To do so, we define a function that computes a backtracking tree, and show that this tree respects the relation defined in \Autoref{fig:tree_semantics}.
However, the termination of this function is not straightforward: the termination of quantifiers relies on the fact that each non-forced iteration must consume at least one character, and that the input itself is finite.
This is made even harder with lookarounds allowing to change the direction of matching.

We first define a function $\computetreefuel{\actions}{\inp}{\gm}{\dir}{n}$ with a fuel $n$, allowing only $n$ recursive calls to ensure termination, but which may return $\None$ when not provided with enough fuel.
Next, we define a function $\fuel{\actions}{\inp}{\dir}$ which computes an upper bound on the minimal amount of recursive calls needed for the function to terminate.
The fuel of a list of actions $\actions$ and a regex $\subreg$, respectively $\fuel{\actions}{\inp}{\dir}$ (left)  and $\fuel{\subreg}{\inp}{\dir}$ (right), is defined as follows:

\smallskip
\noindent\begin{tabular*}{\textwidth}{@{}l@{\extracolsep{\fill}}r@{}}
\(\begin{aligned}
  \fuel{[]}{\inp}{\dir}
  &\defequal 1
  \\
  \fuel{\areg{\re{\subreg}}::\cont}{\inp}{\dir}
  &\defequal \fuel{\subreg}{\inp}{\dir} + \fuel{\cont}{\inp}{\dir}
  \\
  \fuel{\aclose{\gid}::\cont}{\inp}{\dir}
  &\defequal 1 + \fuel{\cont}{\inp}{\dir}
  \\
  \fuel{\acheck{\inpcheck}::\cont}{\inp}{\dir}
  &\defequal 0
  \\
  &\hspace{-2em}\text{(when $\inpcheck$ is at the end for $\dir$)}
  \\
  \fuel{\acheck{\inpcheck}::\cont}{\inp}{\dir}
  &\defequal 1 + \fuel{\cont}{\nextinpdir{\inpcheck}{\dir}}{\dir}
  \\
  &\hspace{-2em}\text{(otherwise)}
\end{aligned}\)
&
\(\begin{aligned}
  \fuel{\re{\regepsilon}}{\inp}{\dir} =
  \fuel{\re{\regchar{\cd}}}{\inp}{\dir}
  &\defequal 1
  \\
  \fuel{\re{\anchor{\anc}}}{\inp}{\dir} =
  \fuel{\re{\backref{\gid}}}{\inp}{\dir}
  &\defequal 1
  \\
  \fuel{\re{\disjunction{\subreg_1}{\subreg_2}}}{\inp}{\dir} =
  \fuel{\re{\sequence{\subreg_1}{\subreg_2}}}{\inp}{\dir}
  &\defequal 1 + \fuel{\re{\subreg_1}}{\inp}{\dir} + \fuel{\re{\subreg_2}}{\inp}{\dir}
  \\
  \fuel{\re{\group{\gid}{\subreg}}}{\inp}{\dir}
  &\defequal 2 + \fuel{\re{\subreg}}{\inp}{\dir}
  \\
  \fuel{\re{\lookaround{\lk}{\subreg}}}{\inp}{\dir}
  &\defequal 2 + \fuel{\re{\subreg}}{\worstinp{\lk}{\inp}}{\lkdir{\lk}}
  \\
  \fuel{\re{\quant{\subreg}{\rmin}{\Delta}{\greedy}}}{\inp}{\dir}
  &\defequal (2 + \fuel{\re{\subreg}}{\inp}{\dir}) \times (1 + \rmin + \inpsize{\inp}{\dir})
  \\
  &\text{}
\end{aligned}\)
\end{tabular*}

\noindent{}There are three interesting cases.
For quantifiers, we know that each non-forced repetition has to make progress.
As a result, there can be at most $1 + min + \inpsize{\inp}{\dir}$ iterations of the quantifiers, where $\inpsize{\inp}{\dir}$ is the number of characters left to be read in the string for direction $\dir$.
For lookarounds, we cannot know ahead-of-time where they will be matched from, so we define  $\worstinp{\lk}{\inp}$ to be the worst possible input for the lookaround direction (at the beginning for a lookahead, at the very end for a lookbehind).
For $\acheck{\inpcheck}$ actions, we know that if the input being compared to is at the end position (for instance, when the current direction is $\forward$ and the input zipper is of the form $\sinput~[]~l_2$), then the progress check cannot succeed.
And otherwise, we know that in the worst case, the input after passing the check will be $\nextinpdir{\inpcheck}{\dir}$, after having read exactly one character.

With these definitions, we can prove that this fuel computation is indeed an upper bound on the number of recursive calls sufficient to finish the computation of the tree.

\begin{theorem}[Termination]
  \label{thm:termination}
  $\forall\: \actions,\: \inp,\: \gm,\: \dir,\: n.\;$
  $n > \fuel{\actions}{\inp}{\dir} \implies \computetreefuel{\actions}{\inp}{\gm}{\dir}{n} \neq \None$.
\end{theorem}
\begin{proof}
  By induction over $n$. In the inductive case, we prove that each recursive call strictly decreases the fuel computation.
\end{proof}

From here on we omit the fuel argument.
Finally, we prove that the tree we compute is a correct backtracking tree for its arguments.
Combined with the determinism of \Autoref{thm:tree_det}, this shows that for any list of actions, input, and direction, there exists a unique backtracking tree.

\begin{theorem}[Functional Semantics Correctness]
  \label{thm:functional_correctness}
  $\forall\: \actions,\: \inp,\: \gm,\: \dir.\;
  \istree{\actions}{\inp}{\gm}{\dir}{\computetree{\actions}{\inp}{\gm}{\dir}}$.
\end{theorem}
\begin{proof}
  By induction over the number of recursive calls.
\end{proof}

\section{Semantics Faithfulness: Equivalence to Warblre}
\label{sec:warblre_equiv}

In this section, we prove that our tree semantics is faithful to the 14\textsuperscript{th} edition of the ECMAScript specification~\cite{ecma_2023}.
To do so, we provide a Rocq proof that it is equivalent to Warblre~\cite{warblre_icfp}, a faithful, manually audited, line-by-line translation of the ECMAScript regex chapter into Rocq.
The ECMAScript regex chapter consists of more than 30 pages of a pseudocode algorithm.
The main function of that algorithm, \textit{CompilePattern} (that we name $\warblrecompile{\subreg}$ for short), compiles a regex into a pseudocode \textit{matcher function} that implements a backtracking search.
This matcher function takes as input a string and a starting index position, and returns a \textit{match result}, which either indicates that no match was found at that starting position, or otherwise returns the final string position and the set of capture group values. 
We prove that the first accepting branch of our backtracking tree semantics corresponds to the result of the Warblre matcher function.

\paragraph{Regex equivalence}
First, to facilitate proofs and provide concise rules, the regex type we use in our semantics differs slightly from the AST used in ECMAScript and Warblre.
While our quantifiers use convenient $\rmin$ and $\Delta$ parameters, in Warblre each quantifier has a $\rmin$ and a $\rmax$ parameter, which corresponds to $\rmin+\Delta$.
At first glance, our type appears less expressive, but ECMAScript starts compilation by checking that the regex is well-formed, which includes that in all quantifiers, $\rmax\geq\rmin$.
We denote by $\wellformed$ the set of well-formed Warblre regexes (those that pass \textit{Early Errors} in ECMAScript parlance).
We define a function, $\tolinden{\subreg_w}$ translating a well-formed regex $\iswf{\subreg_w}$ from the Warblre AST into ours.
Another interesting difference is the indexing of capture groups.
In Warblre, capture groups have no indices in the AST.
Instead, at each manipulation of a group, its index is recomputed by counting the number of open parentheses that occurred before in the original regex.
This process requires carrying around a cumbersome compilation context in Warblre definitions.
To facilitate definitions and proofs instead, our $\tolinden{\subreg_w}$ function annotates in the AST the index of each group.
Similarly, it translates every named group or named backreference into the corresponding indexed group and indexed backreference.
Finally, there are other minor differences in the way characters are represented. For instance, to facilitate formal proofs, we use the same type to represent both single characters and Unicode or identity escapes.

Finally, we can prove the equivalence below (\Autoref{thm:warblre_equiv}).
The Warblre definitions use strings (list of characters) and natural numbers to encode position.
Our zipper inputs are more convenient to express the tree semantics, and we note $\strof{\inp}$ the original string of the zipper input $\inp$.
Finally, as our group map type slightly differs from the Warblre definition, we define a function $\towarblre{\result}$ that transforms our result into a Warblre one.

\begin{theorem}[Faithfulness]
  \label{thm:warblre_equiv}
  $\forall\: \inp,\: \iswf{\subreg_w},\:$\\
  $\warblrecompile{\subreg_w}(\strof{\inp},\idx{\inp}) = \towarblre{\firstbranch{\computetree{[\areg{\tolinden{\re{\subreg_w}}}]}{\inp}{\gmempty}{\forward}{}}{\inp}}$
\end{theorem}
\begin{proof}
  We summarize a few key insights of this considerable proof (over 6000 lines of Rocq).
  In order to define the function $\warblrecompile{\subreg}$, ECMAScript generates \textit{matcher continuations}, functions that can match a regex in a particular direction, open and close a capture group, or check for progress in the string.
  Our proof starts by defining an equivalence relation, parameterized by a direction, between these continuations and our lists of actions.
  We prove an intermediate theorem: the matcher continuation of a regex $\subreg_w$ and a direction $\dir$ is equivalent, for $\dir$, to the list of actions $[\tolinden{\subreg_w}]$.
  \Autoref{thm:warblre_equiv} then follows as a corollary.

  The proof of that intermediate theorem proceeds by induction.
  For quantifiers, we then proceed by induction on their maximum number of iterations (the sum of $\rmin$ and the size of the remaining string for the current direction).
  Groups are also handled differently in the two semantics. In our semantics, we generate an $\treeopen{}{}$ node before matching the subregex, but in ECMAScript, the corresponding matcher continuation remembers the initial input position, matches the regex, then opens and closes the group at the same time.
  As a result, we prove as an invariant that whenever our group maps have partially defined groups (opened but not yet closed), the corresponding matcher continuations in Warblre will later open the same groups at the same index.
\end{proof}

\section{Formally Verified Regex Equivalence}
\label{sec:rewrite}

To demonstrate the practicality of our semantics, we use it to define a new notion of contextual regex equivalence for backtracking semantics.  We use it to prove and disprove some seemingly intuitive equivalences, and even highlight bugs in a real-world JavaScript regex processor.

Reasoning about modern regex equivalence has many real-world uses.
Derivative-based engines compute regex equivalence classes among derivatives during matching~\cite{derivatives_reexamined, dotnet_pldi,resharp}.
Regex-optimization libraries (like \texttt{regexp-tree}~\cite{regexp_tree} for JavaScript regexes) rewrite regexes into equivalent, \textit{optimized} regexes that are expected to be faster to match.
The topic of traditional regular expression equivalence is well studied, even in proof assistants~\cite{decision_procedure, unified_decision_procedure, deciding_equivalence_coq, equivalence_compact_proof, pearl_equivalence}, but these definitions are not applicable to modern regexes with backtracking semantics and nontraditional features.
For these, the literature is much more sparse.
\citet{decidability} has shown that backreferences make regex equivalence undecidable.
Recent work has also established a few results for regex languages with restricted feature sets.
For instance, \citet{lean_lookarounds} show that some anchors can be expressed in terms of lookarounds and
\citet{efficient_lookarounds} establish regex equivalence formulae about lookarounds.

However, even these new equivalences may not hold for a full real-world language with backtracking semantics.
For instance, we found that some lookaround equivalences from~\citet{efficient_lookarounds}, while correct in a language without backreferences, do not hold in JavaScript.
This is because lookarounds do not commute: the regex \re{\lookahead{\group{1}{a}}\lookahead{a\backref{1}}a} does not match the string \str{a}, as the first lookaround sets group 1 to \str{a}, but then the second one fails to read \str{a} twice.
In contrast, \re{\lookahead{a\backref{1}}\lookahead{\group{1}{a}}a} does match the string \str{a}: since group 1 is not defined yet, the backreference matches the empty string.

We present a new way to formally verify regex equivalence using our backtracking tree semantics.  To facilitate comparison to previous work, we restrict this discussion to the default JavaScript semantics, leaving out the \flag{i} (\texttt{ignoreCase}), \flag{m} (\texttt{multiline}), and \flag{s} (\texttt{dotAll}) flags.

\subsection{A New Definition of Regex Equivalence}

\paragraph{Observational equivalence and its limitations}
The most natural way to define regex equivalence is to say that $\subreg_1$ and $\subreg_2$ are equivalent when for any input $\inp$, matching $\subreg_1$ on $\inp$ returns the same result as matching $\subreg_2$ on $\inp$.
We refer to this as \textit{observational equivalence}, and denote it as $\regeq{\subreg_1}{\subreg_2}$.
While observational equivalence is enough to replace the matching of a regex with another, it is not strong enough to be used locally, to replace a subregex within a bigger context.

For instance, consider the regexes \re{\disjunction{\regepsilon}{a}} and \re{\disjunction{\regepsilon}{b}}.
These two regexes are observationally equivalent: on any input, they will both match the empty substring since this corresponds to the match with the highest priority.
However, we cannot replace \re{\disjunction{\regepsilon}{a}} with \re{\disjunction{\regepsilon}{b}} in a bigger context and preserve equivalence.
For instance, \re{c\noncap{\disjunction{\regepsilon}{a}}c} matches the input \str{cac}, but  \re{c\noncap{\disjunction{\regepsilon}{b}}c} does not.

\paragraph{Directional contextual equivalence}
Instead, we need a notion of equivalence that allows local rewriting of regexes.
We refer to this as \textit{contextual equivalence}, and use the notation $\leafeq{\subreg_1}{\subreg_2}$.
We would like this definition to imply observational equivalence in a bigger context: $\leafeq{\subreg_1}{\subreg_2} \implies \forall\: \ctx.\; \regeq{\plug{\ctx}{\subreg_1}}{\plug{\ctx}{\subreg_2}}$,\: where $\ctx$ is a regex context (a regex with a hole).

But in fact, we show that we can define a more precise equivalence with the following novel observation: some regex equivalences only hold in a given matching direction.
As a result, in the rest of this section we define two notions of equivalence, $\leafeqdir{\re{\subreg_1}}{\re{\subreg_2}}{\forward}$ and $\leafeqdir{\re{\subreg_1}}{\re{\subreg_2}}{\backward}$, for equivalences that can be rewritten inside different contexts.

We distinguish three types of contexts.
\textit{Forward contexts}, denoted as $\fwdctx$, have their hole inside a lookahead, meaning that we know the regex will be matched in the forward direction.
\textit{Backwards contexts}, denoted as $\bwdctx$, have their hole inside a lookbehind.
In cases where lookarounds are nested, the deepest lookaround sets the direction of the context.
Finally, \textit{bidirectional} contexts, denoted as $\samectx$, are contexts in which the hole does not appear in a lookaround.

\paragraph{Using leaves to define contextual equivalence}
We now show that \textit{leaves} of backtracking trees of regexes are sufficient to define this directional contextual equivalence.
Crucially, the top priority match of \re{c\noncap{\disjunction{\regepsilon}{a}}c} on \str{cac} (our previous example) uses the \textit{second} priority match of subexpression \re{\disjunction{\regepsilon}{a}}.
To define a contextual equivalence, we must reason on several possible matches of a regex, which our tree semantics was designed to formalize.
We first define $\leaves{\treecont}{\inp}{\dir}$ to return the order-preserving list of leaves (pairs of a final input and a final group map at a $\treematch$ node) of the tree $\treecont$ for input $\inp$ and direction $\dir$.
We say that two lists of leaves are equivalent, noted $\leaveseq{\lea_1}{\lea_2}$, when they are equal after removing lower-priority duplicates in each list.
For instance, $\leaveseq{[(\inp_1,\gm_1);(\inp_2,\gm_2);(\inp_1,\gm_1)]}{[(\inp_1,\gm_1);(\inp_2,\gm_2)]}$.
Finally, we define directional contextual equivalence as follows:\\
$\leafeqdir{\re{\subreg_1}}{\re{\subreg_2}}{\dir}$ $\defequal$ $\defgroups{\subreg_1}=\defgroups{\subreg_2} \wedge \forall\: \inp,\: \gm.\; \leaveseq{\leaves{\computetree{[\subreg_1]}{\inp}{\gm}{\dir}}{\inp}{\dir}}{\leaves{\computetree{[\subreg_2]}{\inp}{\gm}{\dir}}{\inp}{\dir}}$

We require the two regexes to define exactly the same groups, so that replacing $\subreg_1$ with $\subreg_2$ will not turn a well-formed regex into an ill-formed one with duplicated groups.
We use notation $\leafeqdir{\re{\subreg_1}}{\re{\subreg_2}}{\both}$ for equivalence in both directions.
Removing duplicates in the list of leaves allows to prove equivalences like $\leafeqdir{\re{\disjunction{\rchar}{\rchar}}}{\re{\rchar}}{\both}$~: even though the first regex has more leaves, any result that can be obtained from the second branch could also be obtained from the first one.

Using these definitions, we were able to prove the following theorems:
\begin{theorem}[Bidirectional Equivalence]
  \label{thm:same_equiv}
  $\forall\: \subreg_1,\: \subreg_2,\: \dir.\;$
  $\leafeqdir{\re{\subreg_1}}{\re{\subreg_2}}{\dir} \implies \forall \samectx.\; \leafeqdir{\re{\plug{\samectx}{\subreg_1}}}{\re{\plug{\samectx}{\subreg_2}}}{\dir}$
\end{theorem}
\begin{theorem}[Forward Equivalence]
  \label{thm:forward_equiv}
  $\forall\: \subreg_1,\: \subreg_2.\;$
  $\leafeqdir{\re{\subreg_1}}{\re{\subreg_2}}{\forward} \implies \forall\: \fwdctx,\: \dir.\; \leafeqdir{\re{\plug{\fwdctx}{\subreg_1}}}{\re{\plug{\fwdctx}{\subreg_2}}}{\dir}$
\end{theorem}
\begin{theorem}[Backward Equivalence]
  \label{thm:backward_equiv}
  $\forall\: \subreg_1,\: \subreg_2.\;$
  $\leafeqdir{\re{\subreg_1}}{\re{\subreg_2}}{\backward} \implies \forall\: \bwdctx,\: \dir.\; \leafeqdir{\re{\plug{\bwdctx}{\subreg_1}}}{\re{\plug{\bwdctx}{\subreg_2}}}{\dir}$
\end{theorem}
\begin{theorem}[Contextual to Observational Equivalence]
  \label{thm:ctx_observ_equiv}
  $\forall\: \subreg_1,\: \subreg_2.\;$
  $\leafeqdir{\re{\subreg_1}}{\re{\subreg_2}}{\forward} \implies \regeq{\subreg_1}{\subreg_2}$
\end{theorem}
\begin{proof}
  Most proofs proceed by induction over the context. In the case of quantifiers for the bidirectional equivalence, we proceed by induction over the length of the input remaining to match (giving an upper-bound for the number of remaining iterations). 
\end{proof}

\subsection{Formally Verified and Invalid Equivalences}
Using these theorems, we can now locally rewrite subregexes, using the forward equivalence when inside a lookahead or outside of lookarounds, and the backward equivalence when inside a lookbehind.
As a case study, we have proved or disproved rewrites from the literature.
We consider three sets of regex rewrites: rewriting anchors as lookarounds, quantifier merging, and the traditional associativity and distributivity of sequence and disjunction.

\begin{figure}
  \begin{tabular}[t]{cc}
    \begin{tabular}[t]{c}
      Associativity\\
      \hline
      \(\begin{aligned}\Tstrut
        \re{\disjunction{\subreg_1}{\noncap{\disjunction{\subreg_2}{\subreg_3}}}} & \botheq \re{\disjunction{\noncap{\disjunction{\subreg_1}{\subreg_2}}}{\subreg_3}} \\
        \re{\subreg_1\noncap{\subreg_2\subreg_3}} & \botheq \re{\noncap{\subreg_1\subreg_2}\subreg_3}
      \end{aligned}\)\\
      Distributivity (when $\subreg_1$ has no group)\Tstrut\\
      \hline
      \(\begin{aligned}\Tstrut
        \re{\subreg_1\noncap{\disjunction{\subreg_2}{\subreg_3}}} & \bwdeq \re{\disjunction{\noncap{\subreg_1\subreg_2}}{\noncap{\subreg_1\subreg_3}}} \\
        \re{\noncap{\disjunction{\subreg_2}{\subreg_3}}\subreg_1} & \fwdeq \re{\disjunction{\noncap{\subreg_2\subreg_1}}{\noncap{\subreg_3\subreg_1}}} \\
      \end{aligned}\)
    \end{tabular}
    &
    \begin{tabular}[t]{c}
      Anchors as lookarounds\\
      \hline
      \(\begin{aligned}\Tstrut
        \re{\regbos} & \botheq \re{\neglookbehind{\regall}} \\
        \re{\regeos} & \botheq \re{\neglookahead{\regall}} \\
        \re{\esc{b}} & \botheq \re{\disjunction{\neglookbehind{\esc{w}}\lookahead{\esc{w}}}{\lookbehind{\esc{w}}\neglookahead{\esc{w}}}} \\
        \re{\esc{B}} & \botheq \re{\noncap{\disjunction{\lookbehind{\esc{w}}}{\neglookahead{\esc{w}}}}\noncap{\disjunction{\neglookbehind{\esc{w}}}{\lookahead{\esc{w}}}}}
      \end{aligned}\)
    \end{tabular}
  \end{tabular}
  \caption{Correct regex equivalences --- associativity, distributivity and anchors}
  \Description{} 
\label{fig:correct_rewrites}
\end{figure}

\paragraph{Associativity and distributivity}
The regex equivalences we have proved correct are shown on \Autoref{fig:correct_rewrites}.
The missing directions for distributivity are not shown because they are not correct: we provide counter-examples in \appendixref{app:counterex} and \Autoref{fig:distr_counterex}.
We prove associativity for both the sequence and the disjunction.
These proofs illustrate that our semantics eliminates some unnecessary complexity from the Warblre definitions.
Consider for instance the associativity of disjunction. 
In Warblre, we would need to prove that even though $\subreg_2$ is compiled in different contexts, these two contexts are equivalent in the sense that they do not change the ordering of capture groups, requiring a new general property for matcher functions compiled in equivalent contexts. 
In contrast, with our equivalence definition, this associativity boils down to the associativity of the concatenation of lists of leaves, and the proof concludes in five lines.
We also find that contrary to traditional regular expressions, concatenation and disjunction can only distribute under certain conditions.
First, the distributed regex must not define any groups, to preserve the property that each capture group is defined at most once in JavaScript regexes.
Second, depending on the context direction, distributing can change the top priority result.
For instance, in \Autoref{fig:distr_counterex}, we can see that distributing \re{\disjunction{a}{ab}} over \re{\disjunction{c}{b}} changes which disjunction is considered first in their respective backtracking trees, as a result the two middle branches are inverted and the regexes are not equivalent.

\begin{figure}
\begin{tikzpicture}[%
          every node/.style={rectangle,minimum size=6pt,minimum height=4pt, inner sep=1pt, align=center},
          node distance=.2cm, >=latex
              ]
    \begin{pgfonlayer}{foreground}
    \node [] (t0) [] {$\treechoice{}{}$};
    \node [] (t1) [below left=of t0, xshift=.5cm] {$\treeread{a}{}$};
    \node [] (t2) [below=of t1] {$\treechoice{}{}$};
    \node [] (t3) [below left=of t2] {$\treemismatch$};
    \begin{scope}[local bounding box=bb2]
    \node [] (t4) [below=of t2] {$\treeread{b}{}$};
    \node [] (t5) [below=of t4] {$\treematch$};
    \end{scope}
    \begin{scope}[local bounding box=bb1]
    \node [] (t6) [below right=of t0, xshift=-.5cm] {$\treeread{a}{}$};
    \node [] (t7) [below=of t6] {$\treeread{b}{}$};
    \end{scope}
    \node [] (t8) [below=of t7] {$\treechoice{}{}$};
    \begin{scope}[local bounding box=bb3]
    \node [] (t9) [below=of t8] {$\treeread{c}{}$};
    \node [] (t10) [below=of t9] {$\treematch$};
    \end{scope}
    \node [] (t11) [below right=of t8] {$\treemismatch$};
    \node [] (s0) [right=of t0, xshift=5.5cm] {$\treechoice{}{}$};
    \begin{scope}[local bounding box=bb6]
    \node [] (s1) [below left=of s0, xshift=.5cm] {$\treechoice{}{}$};
    \node [] (s4) [below=of s1] {$\treeread{a}{}$};
    \node [] (s5) [below=of s4] {$\treeread{b}{}$};
    \end{scope}
    \node [] (s2) [below left=of s1] {$\treeread{a}{}$};
    \node [] (s3) [below=of s2] {$\treemismatch$};
    \begin{scope}[local bounding box=bb4]
    \node [] (s6) [below=of s5] {$\treeread{c}{}$};
    \node [] (s7) [below=of s6] {$\treematch$};
    \end{scope}
    \begin{scope}[local bounding box=bb8]
    \node [] (s8) [below right=of s0, xshift=-.5cm] {$\treechoice{}{}$};
    \node [] (s9) [below=of s8] {$\treeread{a}{}$};
    \end{scope}
    \begin{scope}[local bounding box=bb5]
    \node [] (s10) [below=of s9] {$\treeread{b}{}$};
    \node [] (s11) [below=of s10] {$\treematch$};
    \end{scope}
    \begin{scope}[local bounding box=bb7]
    \node [] (s12) [below right=of s8] {$\treeread{a}{}$};
    \node [] (s13) [below=of s12] {$\treeread{b}{}$};
    \end{scope}
    \node [] (s14) [below=of s13] {$\treemismatch$};
    \path [draw] (t0) edge[->] node {} (t1);
    \path [draw] (t0) edge[->] node {} (t6);
    \path [draw] (t1) edge[->] node {} (t2);
    \path [draw] (t2) edge[->] node {} (t3);
    \path [draw] (t2) edge[->] node {} (t4);
    \path [draw] (t4) edge[->] node {} (t5);
    \path [draw] (t6) edge[->] node {} (t7);
    \path [draw] (t7) edge[->] node {} (t8);
    \path [draw] (t8) edge[->] node {} (t9);
    \path [draw] (t8) edge[->] node {} (t11);
    \path [draw] (t9) edge[->] node {} (t10);
    \path [draw] (s0) edge[->] node {} (s1);
    \path [draw] (s0) edge[->] node {} (s8);
    \path [draw] (s1) edge[->] node {} (s2);
    \path [draw] (s1) edge[->] node {} (s4);
    \path [draw] (s2) edge[->] node {} (s3);
    \path [draw] (s4) edge[->] node {} (s5);
    \path [draw] (s5) edge[->] node {} (s6);
    \path [draw] (s6) edge[->] node {} (s7);
    \path [draw] (s8) edge[->] node {} (s9);
    \path [draw] (s8) edge[->] node {} (s12);
    \path [draw] (s9) edge[->] node {} (s10);
    \path [draw] (s10) edge[->] node {} (s11);
    \path [draw] (s12) edge[->] node {} (s13);
    \path [draw] (s13) edge[->] node {} (s14);
    \end{pgfonlayer}
    \begin{pgfonlayer}{background}
    \fill [ex1] (t0.north west) rectangle ([yshift=-.47cm]t0.south east);
    \fill [ex1] (t1.north west) rectangle (t1.south east);
    \fill [ex1] (bb1.north west) rectangle (bb1.south east);
    \fill [ex2] (t3.north west) rectangle (t3.south east);
    \fill [ex2] (bb3.north west) rectangle (bb3.south east);
    \fill [ex3] (t11.north west) rectangle (t11.south east);
    \fill [ex3] (bb2.north west) rectangle (bb2.south east);
    \fill [ex2] (s3.north west) rectangle (s3.south east);
    \fill [ex2] (bb4.north west) rectangle (bb4.south east);
    \fill [ex3] (bb5.north west) rectangle (bb5.south east);
    \fill [ex3] (s14.north west) rectangle ([yshift=-.05cm]s14.south east);
    \fill [ex1] (s2.north west) rectangle ([xshift=1cm]s2.south east);
    \fill [ex1] (bb6.north west) rectangle (bb6.south east);
    \fill [ex1] ([xshift=-.075cm]bb7.north west) rectangle (bb7.south east);
    \fill [ex1] (bb8.north west) rectangle ([xshift=.075cm]bb8.south east);
    \end{pgfonlayer}
    \end{tikzpicture}
    \savere{aabcb}{\re{\noncap{\reghl{ex1}{\disjunction{a}{ab}}}\noncap{\disjunction{\reghl{ex2}{c}}{\reghl{ex3}{b}}}}}
    \savere{aabcaabb}{\re{\disjunction{\noncap{\reghl{ex1}{\disjunction{a}{ab}}}\reghl{ex2}{c}}{\noncap{\reghl{ex1}{\disjunction{a}{ab}}}\reghl{ex3}{b}}}}
    \caption{Backtracking trees of \usere{aabcb}
      and \usere{aabcaabb} on string \str{abc}}
    \Description{Left tree: `Choice' root, with `Read a' on both sides.  To the left, a `Choice' with `Mismatch' on the left and `Read b', `Match' on the right.  To the right, a `Read b', `Choice' with `Read c', `Match' on the left, `Mismatch' on the right.  Right tree: `Choice' root, with `Choice' on both sides.  To the left, the choice has `Read a', `Match' on the left, and `Read a', `Read b', `Read c', `Match' on the right.  To the right, the choice has `Read a', `Read b', `Match' on the left, and `Read a', `Read b', `Mismatch' on the right.}
    \label{fig:distr_counterex}
\end{figure}

\paragraph{Anchors as lookarounds}
We also proved that anchors can be rewritten as lookarounds, where \re{\regall} is a regex that matches all characters and \re{\esc{w}} matches all word characters.
These equivalences had been formulated by~\citet{derivatives_extended} and the first two had been mechanized in Lean by~\citet{lean_lookarounds}, but for a regex language without backtracking semantics.
Now, using our semantics and equivalence definitions, we can prove that these rewrites hold for JavaScript regexes as well.

\paragraph{Quantifier merging}
Finally, the JavaScript regex processor library \texttt{regexp-tree}~\cite{regexp_tree} optimizes repeated quantifiers in a regex.
It replaces every instance of \re{\quant{\subreg}{\rmin_1}{\Delta_1}{\greedy}\quant{\subreg}{\rmin_2}{\Delta_2}{\greedy}} with \re{\quant{\subreg}{\rmin_1+\rmin_2}{\Delta_1+\Delta_2}{\greedy}}, where $+$ is defined such that $n+\infty\defequal\infty+n\defequal\infty+\infty\defequal\infty$.
While trying to prove the correctness of this equivalence, we have found counter-examples.
But we also found that, when $\subreg$ defines no group, depending on the $\rmin$ and $\Delta$ parameters, their order, and the direction of the context, it can be correct to perform the optimization.
First, we prove the following equivalence, allowing to change the greediness of a quantifier with only forced iterations:
$\leafeqdir{\re{\quant{\subreg}{\rmin}{0}{\top}}}{\re{\quant{\subreg}{\rmin}{0}{\bot}}}{\both}$.
Then, we prove the eight correct equivalences shown on \Autoref{fig:quantifier_merge_correct}.
We give counter-examples for incorrect ones in \appendixref{app:counterex}.
The $\ok$ and $\nok$ labels indicate correct and incorrect rewrite directions, and $\notapp$ indicates that merging is not applicable since it mixes non-forced repetitions of different priorities.  Since \texttt{regexp-tree} performed all rewrites in both directions, all $\nok$ marks are \texttt{regexp-tree} bugs.
We have reported the issue to the \texttt{regexp-tree} maintainers.\footnote{\url{https://github.com/DmitrySoshnikov/regexp-tree/issues/267}}

\begin{figure}[h]
\begin{tabular}{c | c c c}
\diagbox[width=11em]{\small1\textsuperscript{st} quantifier}{\small2\textsuperscript{nd} quantifier}
                                          & \re{\quant{\subreg}{\rmin_2}{0}{\greedy}} & \re{\quant{\subreg}{0}{\Delta_2}{\top}} & \re{\quant{\subreg}{0}{\Delta_2}{\bot}}\\
\hline
\re{\quant{\subreg}{\rmin_1}{0}{\greedy}} & $\both^\ok$                                & $\forward^\ok~\backward^\nok$ & $\forward^\ok~\backward^\nok$ \\
\re{\quant{\subreg}{0}{\Delta_1}{\top}}           & $\backward^\ok~\forward^\nok$               & $\both^\ok$                  & $\notapp$ \\
\re{\quant{\subreg}{0}{\Delta_1}{\bot}}           & $\backward^\ok~\forward^\nok$               & $\notapp$                   & $\both^\nok$ \\
\end{tabular}
\caption{Correct and incorrect regex quantifier merging equivalences}
\Description{} 
\label{fig:quantifier_merge_correct}
\end{figure}

To provide an insight into these equivalence proofs, consider for instance proving one direction of the middle case of \Autoref{fig:quantifier_merge_correct}: $\leafeqdir{\re{\quant{\subreg}{0}{\Delta_1}{\top} \quant{\subreg}{0}{\Delta_2}{\top}}}{\re{\quant{\subreg}{0}{\Delta_1 + \Delta_2}{\top}}}{\forward}$.
In the case where both $\Delta_1$ and $\Delta_2$ are natural numbers, this is proved by induction on $\Delta_1$.
The simplified trees of each regex are shown on \Autoref{fig:atmost_atmost}.
By induction hypothesis, the two framed subtrees have equivalent leaves.
However, when $\Delta_2 > 0$, the tree on the left has an extra subtree when the first quantifier is skipped, but the second one iterates.
Since $\Delta_2-1 \leq \Delta_1+\Delta_2$, we can prove that each leaf of this extra subtree is a leaf of the subtree of \re{\quant{\subreg}{0}{\Delta_1+\Delta_2}{\top}}, and then by induction hypothesis that this is a leaf of \re{\quant{\subreg}{0}{\Delta_1}{\top} \quant{\subreg}{0}{\Delta_2}{\top}} as well.
As a result, all leaves of this extra subtree are duplicates of the leaves of the leftmost subtree, and the leaf equivalence between the two regexes hold.

\begin{figure}[h]
\begin{tikzpicture}[%
          every node/.style={rectangle,minimum size=6pt,minimum height=4pt, inner sep=1pt, align=center},
          node distance=.5cm, >=latex, shorten >=0.1cm
    ]
  \begin{pgfonlayer}{foreground}
  \node [] (m1tn) [] {\annotation{\re{\quant{\subreg}{0}{\Delta_1+1}{\top}\quant{\subreg}{0}{\Delta_2}{\top}}}\\$\treechoice{}{}$};
  \node [draw] (mtn) [below left=of m1tn, xshift=1.4cm] {\annotation{\re{\quant{\subreg}{0}{\Delta_1}{\top}\quant{\subreg}{0}{\Delta_2}{\top}}}};
  \node [] (n) [below right=of m1tn, xshift=-1.4cm] {\annotation{\re{\quant{\subreg}{0}{\Delta_2}{\top}}}\\$\treechoice{}{}$};
  \node [] (n1) [below left=of n] {\annotation{\re{\quant{\subreg}{0}{\Delta_2-1}{\top}}}};
  \node [] (match1) [below right=of n] {$\treematch$};
  \node [] (mn1) [right=of m1tn, xshift=4cm] {\annotation{\re{\quant{\subreg}{0}{\Delta_1+\Delta_2+1}{\top}}}\\$\treechoice{}{}$};
  \node [draw] (mn) [below left=of mn1, xshift=1cm] {\annotation{\re{\quant{\subreg}{0}{\Delta_1+\Delta_2}{\top}}}};
  \node [] (match2) [below right=of mn1, xshift=-1cm] {$\treematch$};
  \path [draw] (m1tn) edge[->,above left] node {{\tiny iterate}} (mtn);
  \path [draw] (m1tn) edge[->,above right] node {\tskip} (n);
  \path [draw] (n) edge[->,above left] node {{\tiny iterate}} (n1);
  \path [draw] (n) edge[->,above right] node {\tskip} (match1);
  \path [draw] (mn1) edge[->,above left] node {{\tiny iterate}} (mn);
  \path [draw] (mn1) edge[->,above right] node {\tskip} (match2);
  \node [] (eq) [left=of mn] {$\fwdeq$};
  \node [] (extra) [left=of n1,xshift=.4cm] {{\tiny extra subtree:}};
  \end{pgfonlayer}
\end{tikzpicture}
\savere{quanteq}{$\leafeqdir{\re{\quant{\subreg}{0}{\Delta_1}{\top} \quant{\subreg}{0}{\Delta_2}{\top}}}{\re{\quant{\subreg}{0}{\Delta_1 + \Delta_2}{\top}}}{\forward}$}
\caption{Inductive case when proving \usere{quanteq}}
\Description{Two side-by-side trees.  Left tree: a `Choice' leads on the left to the same regex with Delta1+1 replaced with Delta1 (iteration), and on the right to the same regex with the first quantifier skipped (skip).  Below the right regex is another choice leading to Delta2 being decremented on the left (iteration), and a Match on the right (skip).  Right tree: a `Choice' leads to the same regex with Delta1+Delta2 (iteration) or a Match (skip).}
\label{fig:atmost_atmost}
\end{figure}

These rewrites can be combined.
For instance, the following chain of forward equivalences holds:\\
\begin{tabular}{l c l}
\re{\quant{\subreg}{\rmin_0}{0}{\bot}\quant{\subreg}{\rmin_1}{\Delta_1}{\top}\quant{\subreg}{0}{\Delta_2}{\top}} &$\fwdeq$&
\re{\quant{\subreg}{\rmin_0}{0}{\top}\quant{\subreg}{\rmin_1}{0}{\top}\quant{\subreg}{0}{\Delta_1}{\top}\quant{\subreg}{0}{\Delta_2}{\top}}\\
&$\fwdeq$& \re{\quant{\subreg}{\rmin_0+\rmin_1}{0}{\top}\quant{\subreg}{0}{\Delta_1+\Delta_2}{\top}}\\
&$\fwdeq$& \re{\quant{\subreg}{\rmin_0+\rmin_1}{\Delta_1+\Delta_2}{\top}}\\
\end{tabular}

For each invalid rewrite, we conjecture that restrictions either on the subregex $\subreg$ or on the contexts could restore correctness.
For traditional regular expressions, previous work explored transformations valid within an optional repetition or a quantifier~\cite{simplifying_regular}.

\section{Formal Verification of the PikeVM Matching Algorithm}
\label{sec:pikevm}

To implement modern backtracking semantics, most engines use a backtracking algorithm.
For instance, both V8 and SpiderMonkey (the JavaScript implementations of Google Chrome and Firefox, respectively) use Irregexp, a backtracking regex engine.
However, bactracking algorithms suffer from string-size exponential complexity, leading to
the \textit{REgex Denial Of Service} (ReDoS) vulnerability, a complexity attack harming many real-world programs~\cite{freezing_the_web}.

To sidestep this complexity vulnerability, some modern engines instead use a more restrictive flavor of regexes and implement linear-time matching algorithms.
For instance, by removing backreferences (which make the matching problem NP-hard~\cite{nphard}) and lookarounds, libraries like RE2~\cite{re2}, HyperScan~\cite{hyperscan}, the Rust \texttt{Regex} crate~\cite{rust_regex} or the Go \texttt{regexp} package~\cite{go_regexp} all achieve linear-time guarantees.
For languages with non-linear features like .NET~\cite{dotnet_pldi} or JavaScript~\cite{non_backtrack_v8}, implementers sometimes provide both a complete backtracking engine and a restricted linear engine for a subset of regexes.

The PikeVM algorithm~\cite{regexp_vm_approach} is a popular linear-time matching algorithm for modern regexes with backtracking semantics.
This algorithm is a descendant of NFA simulation~\cite{thompson68}, where the NFA is extended to encode priority and is represented as interpreted bytecode.
With its encoding of priority, the PikeVM algorithm supports capture groups, and greedy and lazy quantifiers, and returns the same result as a backtracking algorithm but without the exponential complexity.
In the examples above, all the linear engines supporting capture groups implement a PikeVM (among other linear algorithms): the RE2 library\footnote{\url{https://github.com/google/re2/blob/main/re2/nfa.cc}}, the Rust \texttt{Regex} crate\footnote{\url{https://github.com/rust-lang/regex/blob/master/regex-automata/src/nfa/thompson/pikevm.rs}}, the Go engine\footnote{\url{https://cs.opensource.google/go/go/+/master:src/regexp/exec.go}} and the linear engine of V8 for JavaScript regexes.\footnote{\url{https://github.com/v8/v8/tree/main/src/regexp/experimental}}

However, the execution of the PikeVM algorithm is vastly different from the execution of a backtracking algorithm.
It explores several paths in parallel, and discards some paths to avoid the exponential complexity.
As a result, a PikeVM can be challenging to implement.
Worse: small variations in semantics can affect the correctness of the base algorithm, and \citet{regelk_pldi} have shown that the PikeVM implementation deployed in V8 used to contain a semantic bug.
While traditional NFA simulation itself has been formalized and verified many times (see \autoref{sec:related}), handling priority requires new formal arguments, and to the best of our knowledge the PikeVM has never been formally verified despite its widespread use in modern engines.

In this section, we provide the first formal verification of the PikeVM algorithm.
The proof has been fully mechanized in the Rocq proof assistant.
Using an inductive backtracking tree semantics greatly facilitates the proof, as several crucial properties can be conveniently expressed on trees.
We prove that the algorithm returns the result specified by the ECMAScript standard, making it the first time a modern regex matching algorithm is formally verified against a real-world specification.

\subsection{The PikeVM Algorithm}

\begin{figure}
\centering
\begin{minipage}[t]{.6\textwidth}
  \centering
  \raggedright
  \begin{tabular}{l r l l}
    \re{\subreg} & $::=$ & \re{\regepsilon} & Empty \\
    &$\mid$& \re{\regchar{\cd}} & Character\\
    &$\mid$& \re{\disjunction{\subreg_1}{\subreg_2}} & Disjunction\\
    &$\mid$& \re{\sequence{\subreg_1}{\subreg_2}} & Sequence\\
    &$\mid$& \re{\group{\gid}{\subreg}} & Capture group\\
    &$\mid$& \re{\regstar{\subreg}} & Greedy star\\
    &$\mid$& \re{\lazystar{\subreg}} & Lazy star\\
  \end{tabular}\\
  \captionof{figure}{The subset of regexes, $\pikesub$,\\ supported by the PikeVM algorithm}
  \Description{} 
  \label{fig:syntax_pikevm}
\end{minipage}%
\begin{minipage}[t]{.4\textwidth}
  \centering
  \raggedright
  \begin{tabular}{l r l}
    $\instr$ & $::=$ & \accept\\
    &$\mid$& \consume{$\cd$}\\
    &$\mid$& \jmp{$\lbl$}\\
    &$\mid$& \fork{$\lbl_1$}{$\lbl_2$}\\
    &$\mid$& \setregopen{$\gid$}\\
    &$\mid$& \setregclose{$\gid$}\\
    &$\mid$& \resetregs{$\gidl$}\\
    &$\mid$& \beginloop\\
    &$\mid$& \iendloop{$\lbl$}\\
  \end{tabular}\\
  \captionof{figure}{NFA bytecode instructions}
  \Description{} 
  \label{fig:bytecode}
\end{minipage}
\end{figure}

In this section, we present a formal model of the PikeVM algorithm, first introduced by Rob Pike in the text editor \texttt{sam}~\cite{pike_sam}.
We formalize and prove a version of the PikeVM that slightly differs from the original algorithm to account for JavaScript-specific semantics: we generate instructions for capture reset, and we use the extension described in~\citet{regelk_pldi} to support JavaScript nullable quantifiers.
Both of these extensions have been used in the PikeVM engine of V8 (we discuss adapting these definitions for other languages in \Autoref{sec:related}).
This algorithm supports a subset of JavaScript regexes, $\pikesub$, depicted on \Autoref{fig:syntax_pikevm}.
This corresponds to the traditional features of regular expressions, extended with capture groups, and with the star being either greedy or lazy.
We extend this subset to lists of actions (all $\acheck{}$ and $\aclose{}$ actions belong to the subset).
We also extend this subset to trees (rejecting trees containing nodes corresponding to features that are not in $\pikesub$, like $\treebackref{}{}$ or $\treelk{}{}{}$).

\begin{wrapfigure}{L}{5.3cm}
  \begin{tabular}{@{}l r l@{}}
    \textbf{Regex} \re{\subreg} & \textbf{Label} & $\compilation{\re{\subreg}}$\\
    \hline
    \re{\regepsilon} & & \textit{no instruction}\\
    \hline
    \re{\regchar{\cd}} & & \consume{$\cd$}\\
    \hline
    \re{\sequence{\subreg_1}{\subreg_2}} & & $\compilation{\subreg_1}$\\
    & & $\compilation{\subreg_2}$\\
    \hline
    \re{\disjunction{\subreg_1}{\subreg_2}} & & \fork{$\lbl_1$}{$\lbl_2$}\\
    & $\lbl_1:$ & $\compilation{\subreg_1}$\\
    & & \jmp{$\lbl_{3}$}\\
    & $\lbl_2:$ & $\compilation{\subreg_2}$\\
    & $\lbl_{3}:$ & \dots\\
    \hline
    \re{\group{\gid}{\subreg}} & & \setregopen{$\gid$}\\
    & & $\compilation{\subreg}$\\
    & & \setregclose{$\gid$}\\
    \hline
    \re{\regstar{\subreg}} & $\lbl_{s}:$ & \fork{$\lbl_{in}$}{$\lbl_{out}$}\\
    & $\lbl_{in}:$ & \beginloop\\
    & & \resetregs{$\defgroups{\subreg}$}\\
    & & $\compilation{\subreg}$\\
    & & \iendloop{$\lbl_{s}$}\\
    & $\lbl_{out}:$ & \dots\\
    \hline
    \re{\lazystar{\subreg}} & $\lbl_{s}:$ & \fork{$\lbl_{out}$}{$\lbl_{in}$}\\
    & $\lbl_{in}:$ & \beginloop\\
    & & \resetregs{$\defgroups{\subreg}$}\\
    & & $\compilation{\subreg}$\\
    & & \iendloop{$\lbl_{s}$}\\
    & $\lbl_{out}:$ & \dots\\
    \hline
  \end{tabular}
  \caption{Compiling a regex to its\\ bytecode extended NFA}
  \Description{} 
  \label{fig:pikevm_compile}
\end{wrapfigure}

\paragraph{Compilation}
The first step of the PikeVM algorithm consists in computing a bytecode representing the NFA of the regex.
An NFA is encoded as a list of bytecode instructions.
Each bytecode instruction corresponds to a state of the NFA; these instructions are represented on \Autoref{fig:bytecode}.
All instructions are labeled with a natural number $\lbl$ indicating their position in the list.
Edges of the NFA are represented in two ways: either the instructions directly contain the labels of their successor states, or there is an implicit edge between each instruction without a label and the following instruction in the list.
The NFA bytecode compilation function, shown on \Autoref{fig:pikevm_compile}, is an extension of the traditional Thompson NFA construction~\cite{thompson68}.
The recursive compilation function, $\compilation{\subreg}$, transforms a regex into a list of bytecode instructions.
The dots ``\dots'' indicate a fresh label, to be used for the next instruction.
At the end of compilation, an \accept~ instruction (the accepting state of the NFA) is appended to the end of the list of instructions.
To encode priority, the labels in the \fork{}{}~ instruction are ordered: the first label corresponds to the top priority branch to explore.
As a result, the compilation of \re{\regstar{\subreg}} and \re{\lazystar{\subreg}} differ in their first instruction, the greedy star giving more priority to doing one more iteration, and the lazy one giving more priority to exiting.

\begin{figure}
  \mbox{\infer[\ruledef{pikevm}{Final}]{\pikevmstep{\pvstate{\inp}{\best}{[]}{[]}{\seen}}{\pvsfinal{\best}}{\code}}{}}\hspace{0.4cm}%
  \mbox{\infer[\ruledef{pikevm}{NextChar}]
    {\pikevmstep
      {\pvstate{\inp}{\best}{[]}{\blocked}{\seen}}
      {\pvstate{\inp'}{\best}{\blocked}{[]}{\emptyset}}{\code}}
    {\nextinp{\inp} = \Some{\inp'} \newpremise \blocked \neq []}}

  \semspace
  \mbox{\infer[\ruledef{pikevm}{Skip}]
    {\pikevmstep
      {\pvstate{\inp}{\best}{\thread{\pc}{\gm}{\bo} :: \pactive}{\blocked}{\seen}}
      {\pvstate{\inp}{\best}{\pactive}{\blocked}{\seen}}{\code}}
    {\inseen{\pc}{\bo}{\seen}}}
  
  \semspace
  \mbox{\infer[\ruledef{pikevm}{Match}]
    {\pikevmstep
      {\pvstate{\inp}{\best}{\thread{\pc}{\gm}{\bo} :: \pactive}{\blocked}{\seen}}
      {\pvstate{\inp}{\Some{(\inp,\gm)}}{[]}{\blocked}{\seen'}}{\code}}
    {\getpc{\code}{\pc} = \accept \newpremise \notseen{\pc}{\bo}{\seen} \newpremise \seen' = \addseen{\seen}{\pc}{\bo}}}

  \semspace
  \mbox{\infer[\ruledef{pikevm}{Block}]
    {\pikevmstep
      {\pvstate{\inp}{\best}{\thread{\pc}{\gm}{\bo} :: \pactive}{\blocked}{\seen}}
      {\pvstate{\inp}{\best}{\pactive}{\blocked \app [\thread{\pc+1}{\gm}{\canexit}]}{\seen'}}{\code}}
    {\begin{array}{r@{\newpremise} l}
        \getpc{\code}{\pc} = \consume{\cd} & \inpadvance{\cd}{\inp}{\forward} \neq \None \\
        \notseen{\pc}{\bo}{\seen} & \seen' = \addseen{\seen}{\pc}{\bo}
      \end{array}}}

  \semspace
  \mbox{\infer[\ruledef{pikevm}{FailBlock}]
    {\pikevmstep
      {\pvstate{\inp}{\best}{\thread{\pc}{\gm}{\bo} :: \pactive}{\blocked}{\seen}}
      {\pvstate{\inp}{\best}{\pactive}{\blocked}{\seen'}}{\code}}
    {\begin{array}{r@{\newpremise} l}
        \getpc{\code}{\pc} = \consume{\cd} & \inpadvance{\cd}{\inp}{\forward} = \None\\
        \notseen{\pc}{\bo}{\seen} & \seen' = \addseen{\seen}{\pc}{\bo}
      \end{array}}}

  \semspace
  \mbox{\infer[\ruledef{pikevm}{Jump}]
    {\pikevmstep
      {\pvstate{\inp}{\best}{\thread{\pc}{\gm}{\bo} :: \pactive}{\blocked}{\seen}}
      {\pvstate{\inp}{\best}{\thread{\pc'}{\gm}{\bo} :: \pactive}{\blocked}{\seen'}}{\code}}
    {\getpc{\code}{\pc} = \jmp{\pc'} \newpremise \notseen{\pc}{\bo}{\seen} \newpremise \seen' = \addseen{\seen}{\pc}{\bo}}}

  \semspace
  \mbox{\infer[\ruledef{pikevm}{Fork}]
    {\pikevmstep
      {\pvstate{\inp}{\best}{\thread{\pc}{\gm}{\bo} :: \pactive}{\blocked}{\seen}}
      {\pvstate{\inp}{\best}{\thread{\pc_1}{\gm}{\bo} :: \thread{\pc_2}{\gm}{\bo} :: \pactive}{\blocked}{\seen'}}{\code}}
    {\getpc{\code}{\pc} = \fork{\pc_1}{\pc_2} \newpremise \notseen{\pc}{\bo}{\seen} \newpremise \seen' = \addseen{\seen}{\pc}{\bo}}}

    \semspace
  \mbox{\infer[\ruledef{pikevm}{Open}]
    {\pikevmstep
      {\pvstate{\inp}{\best}{\thread{\pc}{\gm}{\bo} :: \pactive}{\blocked}{\seen}}
      {\pvstate{\inp}{\best}{\thread{\pc+1}{\gm'}{\bo} :: \pactive}{\blocked}{\seen'}}{\code}}
    {\begin{array}{r@{\newpremise} l}
        \getpc{\code}{\pc} = \setregopen{\gid} & \gmopen{\gm}{\gid}{\idx{\inp}} = \gm'\\
        \notseen{\pc}{\bo}{\seen} & \seen' = \addseen{\seen}{\pc}{\bo}
      \end{array}}}

  \semspace
  \mbox{\infer[\ruledef{pikevm}{Close}]
    {\pikevmstep
      {\pvstate{\inp}{\best}{\thread{\pc}{\gm}{\bo} :: \pactive}{\blocked}{\seen}}
      {\pvstate{\inp}{\best}{\thread{\pc+1}{\gm'}{\bo} :: \pactive}{\blocked}{\seen'}}{\code}}
    {\begin{array}{r@{\newpremise} l}
        \getpc{\code}{\pc} = \setregclose{\gid} & \gmclose{\gm}{\gid}{\idx{\inp}} = \gm'\\
        \notseen{\pc}{\bo}{\seen} & \seen' = \addseen{\seen}{\pc}{\bo}
      \end{array}}}

  \semspace
  \mbox{\infer[\ruledef{pikevm}{Reset}]
    {\pikevmstep
      {\pvstate{\inp}{\best}{\thread{\pc}{\gm}{\bo} :: \pactive}{\blocked}{\seen}}
      {\pvstate{\inp}{\best}{\thread{\pc+1}{\gm'}{\bo} :: \pactive}{\blocked}{\seen'}}{\code}}
    {\begin{array}{r@{\newpremise} l}
        \getpc{\code}{\pc} = \resetregs{\gidl} & \gmreset{\gm}{\gidl} = \gm'\\
        \notseen{\pc}{\bo}{\seen} & \seen' = \addseen{\seen}{\pc}{\bo}
      \end{array}}}
  
  \semspace
  \mbox{\infer[\ruledef{pikevm}{Begin}]
    {\pikevmstep
      {\pvstate{\inp}{\best}{\thread{\pc}{\gm}{\bo} :: \pactive}{\blocked}{\seen}}
      {\pvstate{\inp}{\best}{\thread{\pc+1}{\gm}{\cannotexit} :: \pactive}{\blocked}{\seen'}}{\code}}
    {\getpc{\code}{\pc} = \beginloop \newpremise \notseen{\pc}{\bo}{\seen} \newpremise \seen' = \addseen{\seen}{\pc}{\bo}}}

  \semspace
  \mbox{\infer[\ruledef{pikevm}{End}]
    {\pikevmstep
      {\pvstate{\inp}{\best}{\thread{\pc}{\gm}{\canexit} :: \pactive}{\blocked}{\seen}}
      {\pvstate{\inp}{\best}{\thread{\pc'}{\gm}{\canexit} :: \pactive}{\blocked}{\seen'}}{\code}}
    {\getpc{\code}{\pc} = \iendloop{\pc'} \newpremise \notseen{\pc}{\canexit}{\seen} \newpremise \seen' = \addseen{\seen}{\pc}{\canexit}}}

  \semspace
  \mbox{\infer[\ruledef{pikevm}{EndStuck}]
    {\pikevmstep
      {\pvstate{\inp}{\best}{\thread{\pc}{\gm}{\cannotexit} :: \pactive}{\blocked}{\seen}}
      {\pvstate{\inp}{\best}{\pactive}{\blocked}{\seen'}}{\code}}
    {\getpc{\code}{\pc} = \iendloop{\pc'} \newpremise \notseen{\pc}{\cannotexit}{\seen} \newpremise \seen' = \addseen{\seen}{\pc}{\cannotexit}}}
  
\caption{PikeVM small-step semantics}%
\Description{} 
\label{fig:pikevm_smallstep}
\end{figure}

\paragraph{Execution}
After the regex is compiled to bytecode, PikeVM interprets this bytecode on
a string of characters.
The algorithm reads each character of the string one at a time.
In between each character, it builds a list of possible incomplete paths (ordered by priority) in the bytecode NFA.

To explore different parts of the bytecode simultaneously, PikeVM keeps track of several threads.
In a thread $\thread{\pc}{\gm}{\bo}$, $\pc$ is the current label, $\gm$ is a group map, and $\bo$ is a Boolean indicating whether the thread has consumed a character since entering its last quantifier.
In this section, we describe the execution algorithm with small-step semantics, shown on \Autoref{fig:pikevm_smallstep}.\footnote{For a more traditional imperative description of the algorithm, we refer the reader to~\citet{regexp_vm_approach}.}
States of the PikeVM algorithm are either of the form $\pvsfinal{\best}$ or tuples $\pvstate{\inp}{\best}{\pactive}{\blocked}{\seen}$.
In these states, $\inp$ is the current input (although PikeVM explores several threads in parallel, they are all synchronized at the same position in the original string).
$\best$ is the top priority result found so far if any (a match with even higher priority may be found later when not in a $\pvsfinal{}$ state).
$\pactive$ is a list of active threads to explore, ordered by priority.
$\blocked$ is an ordered list of blocked threads that reached a \consume{}~ instruction for the current input, which will become the list of active threads when the algorithm advances to the next input.
Finally, $\seen$ is a set of pairs $(\pc,\bo)$ that have already been visited in the execution for the current input.
The initial state of the algorithm is defined to be $\pvsinit{\inp} \defequal \pvstate{\inp}{\None}{[\thread{0}{\gmempty}{\canexit}]}{[]}{\emptyset}$.

\paragraph{Small-step rules}
A final state is reached when there are no more active or blocked threads (\ruleref{pikevm}{Final}).
When a thread reaches an already explored state (where $\inseen{\pc}{\bo}{\seen}$), it is discarded entirely (\ruleref{pikevm}{Skip}).
This mechanism is what makes the algorithm linear: each NFA state is explored at most once per input character (the $\seen$ set is reset each time PikeVM moves to the next input (\ruleref{pikevm}{NextChar})).
Since threads are always ordered by priority, when two threads eventually reach the same configuration $(\pc,\bo)$, the lower priority one is discarded.
The other rules explain how to handle each instruction.
When reaching an \accept~instruction (\ruleref{pikevm}{Match}), a new top-priority match is found and stored in $\best$. In that case, the algorithm discards the remaining, lower-priority active threads, but keeps the blocked threads that could lead to a higher-priority match.
When reaching a \consume{}~ instruction, the current active thread is added to the bottom of the blocked list (\ruleref{pikevm}{Block}), and the Boolean is set to $\canexit$ to indicate that a character has been read since entering the last star.
Conversely, when reaching a \beginloop~ instruction, the Boolean is set to $\cannotexit$.
Finally, when reaching an \iendloop{}~ instruction, the thread is either kept or discarded depending on its Boolean (\ruleref{pikevm}{End}, \ruleref{pikevm}{EndStuck}).

An example of PikeVM bytecode for the regex \re{\regexample} is shown on \Autoref{subfig:bytecode}.
An example execution for the input \str{ab} is shown on the right column of \Autoref{subfig:exec}, showing the evolution of the lists $\pactive$ and $\blocked$ where each thread is represented only with its $\pc$.
The algorithm uses an unconventional exploring order of branches: it alternates between a breadth-first search (computing all possible reachable states for a given input position), and a depth-first search within a given input position.
In that example, we can also see PikeVM skipping a branch: it visits $\pc$ 9 twice for the same input position (once per branch of the disjunction) and discards the second visit.

\paragraph{Correctness}
To formally verify the correctness of PikeVM algorithm, one needs to prove that if a final state can be reached using the small-step semantics of \Autoref{fig:pikevm_smallstep}, then this result corresponds to the ECMAScript specification.
As PikeVM is vastly different from a backtracking algorithm, the proof needs to address the following verification challenges:
\begin{itemize}
\item \textbf{Different progress checks:} Instead of comparing inputs, PikeVM uses a Boolean to encode progress in star iterations, as suggested in~\cite{regelk_pldi}. One needs to verify that this Boolean correctly mirrors the outcomes of each progress check.
\item \textbf{Different exploration schemes:} PikeVM explores several possible paths in parallel, while the backtracking semantics only explores one at a time in a depth-first search. One needs to prove that these two exploration schemes return the same result.
\item \textbf{Skipping branches for linearity:} To ensure linearity, PikeVM skips entire branches when they correspond to previously visited states $(\pc,\bo)$ of the extended NFA. One needs to prove that skipping these paths does not change the result.
\item \textbf{Compilation correctness:} The PikeVM compiles the regex to a bytecode. This compilation needs to be formally verified.
\end{itemize}

We have designed a proof methodology in three steps, addressing most of these challenges separately.
In all steps, the backtracking tree semantics allows to write elegant invariants, and offers a convenient induction principle.
In a first step, in \autoref{subsec:boolean}, we handle the first verification challenge by presenting an alternative tree semantics, resembling the rules of \Autoref{fig:tree_semantics}, but encoding progress with a Boolean instead of comparing the current input with the input in an $\acheck{}$ action.
Then, to handle the two next challenges, we present an intermediate algorithm in \autoref{subsec:piketree}, which follows the path exploration order and the skipping of PikeVM, but without compiling to bytecode.
Finally, in \autoref{subsec:pikevm_piketree}, we tackle the final challenge, and show that the PikeVM algorithm working on the compiled bytecode returns the same result as the intermediate algorithm.

\subsection{Correctness of Encoding Progress}
\label{subsec:boolean}

First, we prove that encoding progress with a Boolean correctly allows PikeVM to predict the behavior of the progress checks.
This property can be expressed purely using trees, without even considering the PikeVM compilation process.

We present a new semantics, the Boolean tree semantics, of which a few selected rules are depicted on \Autoref{fig:bool_semantics}.
This semantics is identical to the previous tree semantics of \autoref{sec:semantics}, with two exceptions.
First, we add a Boolean parameter indicating whether a check would succeed.
Rules \ruleref{bool}{Check} and \ruleref{bool}{CheckFail} use this parameter to decide whether to fail or not, independently from the input inside the $\acheck{}$ action.
The Boolean is updated in the premises of the rules that read a character (\ruleref{bool}{Read}) or enter a non-forced iteration of a quantifier (\ruleref{bool}{Greedy}).
Second, since this semantics is used to verify the correctness of PikeVM, we make simplifications to reflect the restricted subset of supported regexes $\pikesub$.
We remove irrelevant parameters and rules, for instance the direction (no lookarounds) and the group map (no backreferences).

Linear algorithms rely on the uniform-futures property~\cite{regelk_pldi}, that states that while matching a regex, the future of a path is independent from the current values of each group. With the removal of the group map parameter, the Boolean tree semantics of \Autoref{fig:bool_semantics} has this property by construction.

\begin{figure}
  \mbox{\infer[\ruledef{bool}{Check}]{\booltree{\acheck{\inpcheck} :: \cont}{\inp}{\canexit}{\treeprogress{\treecont}}}
    {\booltree{\cont}{\inp}{\canexit}{\treecont}}}\hspace{.2cm}%
  \mbox{\infer[\ruledef{bool}{CheckFail}]{\booltree{\acheck{\inpcheck} :: \cont}{\inp}{\cannotexit}{\treemismatch}}
    {}}%

  \semspace
  \mbox{\infer[\ruledef{bool}{Read}]{\booltree{\areg{\re{\regchar{\cd}}} :: \cont}{\inp}{\bo}{\treeread{\rchar}{\treecont}}}
    {\inpadvance{\cd}{\inp}{\dir} = \Some{(\rchar, \inp')} \newpremise \booltree{\cont}{\inp'}{\canexit}{\treecont}}}%

    \semspace
  \mbox{\infer[\ruledef{bool}{Greedy}]{\booltree{\areg{\re{\quant{\subreg}{0}{\Delta+1}{\top}}} :: \cont}{\inp}{\bo}{\treechoice{(\treereset{\defgroups{\subreg}}{\treeiter})}{\treecont_{skip}}}}
    {\booltree{\cont}{\inp}{\bo}{\treecont_{skip}} \newpremise \booltree{\areg{\re{\subreg}} :: \acheck{\inp} :: \areg{\re{\quant{\subreg}{0}{\Delta}{\top}}} :: \cont}{\inp}{\cannotexit}{\treeiter}}}%

  \caption{Boolean tree semantics --- selected rules}%
\Description{} 
\label{fig:bool_semantics}
\end{figure}

\newcommand{\encodespace}{\hspace{0.2cm}}
\begin{figure}
  \mbox{\infer[]
    {\encodes{[]}{\inp}{\bo}}
    {}}\encodespace%
  \mbox{\infer[]
    {\encodes{\areg{\re{\subreg}}::\actions}{\inp}{\bo}}
    {\encodes{\actions}{\inp}{\bo}}}\encodespace%
  \mbox{\infer[]
    {\encodes{\aclose{\gid}::\actions}{\inp}{\bo}}
    {\encodes{\actions}{\inp}{\bo}}}\encodespace%
  \mbox{\infer[]
    {\encodes{\acheck{\inpcheck}::\actions}{\inp}{\canexit}}
    {\encodes{\actions}{\inp}{\canexit}\newpremise\inpgt{\inp}{\inpcheck}{\forward}}}\encodespace%
  \mbox{\infer[]
    {\encodes{\acheck{\inp}::\actions}{\inp}{\cannotexit}}
    {\encodes{\actions}{\inp}{\bo}}}

  \caption{Encoding actions with a Boolean}%
  \Description{} 
  \label{fig:encodes}
\end{figure}

To relate the two semantics, we define an invariant, $\encodes{\actions}{\inp}{\bo}$ on \Autoref{fig:encodes}.
To understand how it can be used to prove the correctness of the Boolean encoding, consider a backward proof of $\encodes{\actions}{\inp}{b}$.
After using the fourth rule of \Autoref{fig:encodes}, the hypothesis requires the Boolean to be $\canexit$.
Consequently, the rest of the proof cannot use the fifth rule anymore.
In other words, $\encodes{\actions}{\inp}{b}$ implies that $\actions$ can be split into two lists:
first, a list where each $\acheck{\inpcheck}$ action has its input $\inpcheck$ equal to the current input $\inp$ (corresponding to stars that have just been entered),
followed by a list where each $\acheck{\inpcheck}$ action has its input $\inpcheck$ strictly smaller than $\inp$ (for stars in which we have made progress already).
When $b=\canexit$, it means that the first list contains no $\acheck{}$ action: all stars have progressed, and all progress checks will succeed.
When $b=\cannotexit$, there is at least one check for which progress has not been made, and the next check will fail unless we advance in the input.

Finally, we can prove that for the same regex and the same input, the two semantics build the same tree.
With this proof, we have showed that the progress strategy used by PikeVM (encoding with a Boolean the validity of future checks), matches the specification.

\begin{theorem}[Correctness of the Boolean semantics]
  \label{thm:boolean}
  $\forall\: \insub{\actions},\: \inp,\: \gm,\: \treecont,\: \bo.\;$\\
  $\encodes{\actions}{\inp}{\bo} \wedge \istree{\actions}{\inp}{\gm}{\forward}{\treecont} \implies$ 
  $\booltree{\actions}{\inp}{\bo}{\treecont}$.
\end{theorem}
\begin{theorem}
  \label{thm:boolean_correct}
  $\forall\: \insub{\subreg},\: \inp,\: \treecont.\;$ 
  $\istree{[\subreg]}{\inp}{\gmempty}{\forward}{\treecont} \iff$ 
  $\booltree{[\subreg]}{\inp}{\canexit}{\treecont}$.
\end{theorem}
\begin{proof}
  \Autoref{thm:boolean} is proved by induction on a derivation of $\istree{\actions}{\inp}{\gm}{\forward}{\treecont}$.
  The left to right direction of \Autoref{thm:boolean_correct} is a direct consequence of \Autoref{thm:boolean}.
  The right to left direction is a consequence of the productivity of the tree semantics (\Autoref{thm:functional_correctness}), the determinism of the Boolean tree semantics (similar to \Autoref{thm:tree_det}) and \Autoref{thm:boolean}.
\end{proof}

\subsection{Correctness of the PikeVM Exploration Scheme}
\label{subsec:piketree}

\savere{example}{{\re{\regexample}}}

\begin{figure}
\begin{minipage}{.55\textwidth}

\begin{subfigure}{\linewidth}
\centering\begin{tabular}{r | l}
\multicolumn{2}{c}{$\compilation{\re{\regexample}}$}\\
\hline
0 & $\setregopen{1}$\\
1 & $\fork{2}{8}$\\
2 & $\fork{3}{7}$\\
3 & $\beginloop$\\
4 & $\resetregs{[]}$\\
5 & $\consume{\re{a}}$\\
6 & $\iendloop{2}$\\
7 & $\jmp{9}$\\
8 & $\consume{\re{a}}$\\
9 & $\setregclose{1}$\\
10 & $\consume{\re{b}}$\\
11 & $\accept$\\
\end{tabular}
\caption{Bytecode NFA of \usere{example}}
\Description{} 
\label{subfig:bytecode}
\end{subfigure}
\vspace{.3cm}

\begin{subfigure}{\linewidth}
  \begin{tikzpicture}[%
          every node/.style={rectangle,minimum size=6pt,minimum height=4pt, inner sep=2pt, align=center},
          node distance=.2cm, >=latex
    ]
    \begin{pgfonlayer}{foreground}
    \node [] (t0) [] {\tnode{0}$\treeopen{1}{}$};
    \node [] (t1) [below=of t0] {\tnode{1}$\treechoice{}{}$};
    \node [] (t2) [below left=of t1] {\tnode{2}$\treechoice{}{}$};
    \node [] (t3) [below left=of t2, xshift=.75cm] {\tnode{3}$\treereset{[]}{}$};
    \node [] (t4) [below=of t3] {\tnode{4}$\treeread{a}{}$};
    \node [] (t5) [below=of t4] {\tnode{5}$\treeprogress{}$};
    \node [] (t6) [below=of t5] {\tnode{6}$\treechoice{}{}$};
    \node [] (t7) [below left= of t6, xshift=0.75cm] {\tnode{7}$\treereset{[]}{}$};
    \node [] (t8) [below=of t7] {\tnode{8}$\treemismatch$};
    \node [] (t9) [below right=of t6, xshift=-0.75cm] {\tnode{9}$\treeclose{1}{}$};
    \node [] (t10) [below=of t9] {\tnode{10}$\treeread{b}{}$};
    \node [] (t11) [below=of t10] {\tnode{11}$\treematch$};
    \node [] (t12) [below right=of t2, xshift=-0.75cm] {\tnode{12}$\treeclose{1}{}$};
    \node [] (t13) [below=of t12] {\tnode{13}$\treemismatch$};
    \node [] (t14) [below right=of t1] {\tnode{14}$\treeread{a}{}$};
    \node [] (t15) [below=of t14] {\tnode{15}$\treeclose{1}{}$};
    \node [] (t16) [below=of t15] {\tnode{16}$\treeread{b}{}$};
    \node [] (t17) [below=of t16] {\tnode{17}$\treematch$};
    \path [draw] (t0) edge[->] node {} (t1);
    \path [draw] (t1) edge[->] node {} (t2);
    \path [draw] (t1) edge[->] node {} (t14);
    \path [draw] (t2) edge[->] node {} (t3);
    \path [draw] (t2) edge[->] node {} (t12);
    \path [draw] (t3) edge[->] node {} (t4);
    \path [draw] (t4) edge[->] node {} (t5);
    \path [draw] (t5) edge[->] node {} (t6);
    \path [draw] (t6) edge[->] node {} (t7);
    \path [draw] (t6) edge[->] node {} (t9);
    \path [draw] (t7) edge[->] node {} (t8);
    \path [draw] (t9) edge[->] node {} (t10);
    \path [draw] (t10) edge[->] node {} (t11);
    \path [draw] (t12) edge[->] node {} (t13);
    \path [draw] (t14) edge[->] node {} (t15);
    \path [draw] (t15) edge[->] node {} (t16);
    \path [draw] (t16) edge[->] node {} (t17);
    \draw [draw=black] (t9.north west) rectangle (t11.south east);
    \draw [draw=black] (t15.north west) rectangle (t17.south east);
    \path [draw] (t9.east) edge[to-to, bend right,below right,dashed] node {{\tiny Equal subtrees}} ([yshift=-.05cm]t17.south);
    \end{pgfonlayer}
    \end{tikzpicture}
    \caption{Backtracking tree of \usere{example} and \str{ab}}
    \Description{The tree is rooted at `Open 1' followed by `Choice'. To the right is `Read a', `Close 1', `Read b', `Match'; to the left is a `Choice'.  That `Choice' has two subtrees: one with `Close 1', `Mismatch' to the right, and one with  `Reset []'  `Read a', `Progress', `Choice' to the left.  Under that final choice is `Close 1', `Read b', `Match' to the right, and `Reset []', `Mismatch' to the left.  An arrow points out the repeated `Close 1', `Read b', `Match' sequence.}
    \label{subfig:tree}
\end{subfigure}

\end{minipage}%
\hfill
\begin{minipage}{.42\textwidth}

\begin{subfigure}{\linewidth}
\bgroup
\def\arraystretch{1.1}
\begin{tabular}{l l | l l}
\multicolumn{2}{c}{\textbf{PikeTree}} & \multicolumn{2}{c}{\textbf{PikeVM}}\\
$\pactive$ & $\blocked$ & $\pactive$ & $\blocked$\\
\hline
\pike{\ttree{0}}{}{0}{}
\pike{\ttree{1}}{}{1}{}
\pike{\ttree{2};\ttree{14}}{}{2;8}{}
\pike{\ttree{3};\ttree{12};\ttree{14}}{}{3;7;8}{}
& & $[4;7;8]$ & []\\
\pike{\ttree{4};\ttree{12};\ttree{14}}{}{5;7;8}{}
\pike{\ttree{12};\ttree{14}}{\ttree{5}}{7;8}{6}
& & $[9;8]$ & $[6]$\\
\pike{\ttree{13};\ttree{14}}{\ttree{5}}{10;8}{6}
\pike{\ttree{14}}{\ttree{5}}{8}{6}
\pike{}{\ttree{5};\ttree{15}}{}{6;9}
\hline
\multicolumn{4}{c}{Advance after reading \str{a}}\\
\hline
\pike{\ttree{5};\ttree{15}}{}{6;9}{}
\pike{\ttree{6};\ttree{15}}{}{2;9}{}
\pike{\ttree{7};\ttree{9};\ttree{15}}{}{3;7;9}{}
& & $[4;7;9]$ & []\\
\pike{\ttree{8};\ttree{9};\ttree{15}}{}{5;7;9}{}
\pike{\ttree{9};\ttree{15}}{}{7;9}{}
& & $[9;9]$ & []\\
\pike{\ttree{10};\ttree{15}}{}{10;9}{}
\pike{\ttree{15}}{\ttree{11}}{9}{11}
\multicolumn{2}{c|}{\textbf{Skip} \ttree{15}$=$\ttree{9}} & \multicolumn{2}{c}{\textbf{Skip} 9}\\
\pike{}{\ttree{11}}{}{11}
\hline
\multicolumn{4}{c}{Advance after reading \str{b}}\\
\hline
\pike{\ttree{11}}{}{11}{}
\multicolumn{2}{c|}{\textbf{Found result}} & \multicolumn{2}{c}{\textbf{Found result}}\\
\end{tabular}
\egroup
\caption{Execution of PikeTree and PikeVM}
\Description{} 
\label{subfig:exec}
\end{subfigure}
\end{minipage}

\caption{Example execution of the PikeTree and PikeVM algorithms for regex \usere{example} and input \str{ab}}
\Description{} 
\label{fig:pike_ex}
\end{figure}

As the second step of our proof, we present an intermediate algorithm, resembling PikeVM but without bytecode compilation.
This algorithm, which we name PikeTree, is given a backtracking tree, and finds its top priority accepting branch.
Instead of exploring the tree in depth-first order like a backtracking algorithm, the algorithm emulates the
mixed breadth-first and depth-first
exploration order of PikeVM: it explores several paths in parallel, all synchronized on the same input position.
Just like PikeVM avoids exploring the same NFA state twice, this algorithm also skips some subtrees that have already been explored.
Our strategy consists in first showing that this algorithm always computes the first accepting branch of a given tree, then showing that a PikeVM execution of a regex and a string corresponds to a PikeTree execution of the corresponding tree (\autoref{subsec:pikevm_piketree}).
The key advantage of this strategy is that we can prove two crucial properties without even considering compilation to a bytecode: the correctness of the exploration order, and the possibility of skipping entire branches.
Here we also make use of a key property of our backtracking trees: since the PikeVM algorithm explores paths that are not explored by a backtracking algorithm that stops at the first priority match, it is crucial for our semantics to materialize these extra paths.

An example execution of the PikeTree algorithm is shown on \Autoref{fig:pike_ex}.
\Autoref{subfig:tree} shows the backtracking tree of \re{\regexample} on input \str{ab}, which is used to initialize the algorithm.
Instead of maintaining lists of threads like PikeVM, PikeTree maintains lists of trees to explore.
It can also skip trees already explored for the current input position.
For instance, the subtree \ttree{9} is equal to the subtree \ttree{15} (the two rectangles in \Autoref{subfig:tree}), which allows the PikeTree algorithm to skip \ttree{15} in its execution on \Autoref{subfig:exec}.
On this diagram, we can see how the PikeTree algorithm is designed to match the behavior of PikeVM: each active tree corresponds to an active thread (but note that PikeTree sometimes takes fewer steps than PikeVM, see \Autoref{subsec:pikevm_piketree}).

Selected rules of the PikeTree small-step semantics are presented in \Autoref{fig:piketree_smallstep}.
The full semantics is available in \appendixref{app:smallstep_piketree}.
States of the PikeTree algorithm are either of the form $\ptsfinal{\best}$ when the algorithm terminates, or $\ptstate{\inp}{\best}{\pactive}{\blocked}{\seen}$ otherwise.
Given a tree $\treecont$ and an input $\inp$, the initial state of the PikeTree algorithm is defined to be
$\ptsinit{\treecont}{\inp} \defequal \ptstate{\inp}{\None}{[(\treecont,\gmempty)]}{[]}{\emptyset}$.
These states closely resemble those of the PikeVM semantics of \autoref{fig:pikevm_smallstep}, with the following differences.
First, $\pactive$ and $\blocked$ are now lists of pairs $(\treecont,\gm)$ of a subtree to explore and a group map.
Second, the $\seen$ set now contains subtrees that have already been explored.

More importantly, the way in which the PikeTree algorithm skips branches is different from PikeVM.
In PikeVM, different threads, at different program counters, can correspond to the execution of the same subtree.
For instance, when compiling the regex $ab~|~ab$, the PikeVM algorithm will compile $ab$ twice in different locations.
Two threads of PikeVM executing these two pieces of bytecode will always correspond to the same tree, but since they are located at different places in the bytecode, PikeVM will not skip either of them.
As our goal is to relate a PikeVM execution to a PikeTree execution, we make the PikeTree algorithm non-deterministic: when it sees a tree that is already in the $\seen$ set, it can either skip it (\ruleref{piketree}{Skip}), or not skip it and process it anyways (for instance \ruleref{piketree}{Blocked} does not require that $\treeread{\rchar}{\treecont} \notin \seen$).
While the PikeTree algorithm is non-deterministic, we prove that each possible execution produces the correct result, and then show that the PikeVM execution corresponds to one specific PikeTree execution.

\begin{figure}
  \mbox{\infer[\ruledef{piketree}{Final}]{\piketreestep{\ptstate{\inp}{\best}{[]}{[]}{\seen}}{\ptsfinal{\best}}}{}}\hfill%
  \mbox{\infer[\ruledef{piketree}{Skip}]{\piketreestep
      {\ptstate{\inp}{\best}{(\treecont,\gm) :: \pactive}{\blocked}{\seen}}
      {\ptstate{\inp}{\best}{\pactive}{\blocked}{\seen}}}
    {\intseen{\treecont}{\seen}}}

  \semspace
  \mbox{\infer[\ruledef{piketree}{Blocked}]{\piketreestep
      {\ptstate{\inp}{\best}{(\treeread{\rchar}{\treecont},\gm)::\pactive}{\blocked}{\seen}}
      {\ptstate{\inp}{\best}{\pactive}{\blocked\app[(\treecont,\gm)]}{\seen'}}}
    {\seen' = \addtseen{\seen}{\treeread{\rchar}{\treecont}}}}
  
\caption{PikeTree small-step semantics --- selected rules (full version in \appendixref{app:smallstep_piketree})}%
\Description{} 
\label{fig:piketree_smallstep}
\end{figure}

To prove the correctness of our new algorithm we define a relation, $\piketreeinv{\pts}{\result}$, between a state of the PikeTree semantics ($\pts$) and an optional match result ($\result$), then show that this relation is an invariant of the execution.
This invariant is a combination of several other relations.

First, we define a non-deterministic relation $\ptres{(\treecont,\gm)}{\result}{\seen}{\inp}$, relating a tree $\treecont$ to an optional result $\result$.
Informally, $\ptres{(\treecont,\gm)}{\result}{\seen}{\inp}$ holds when $\result$ is a leftmost accepting leaf of $\treecont$
after removing any number of subtrees found in $\seen$,
for input $\inp$ and group map $\gm$ .
This means that executing the PikeTree algorithm on tree $\treecont$ with the set $\seen$ can produce result $\result$ (recall that PikeTree  may, but does not \textit{have to}, skip trees in $\seen$).
For instance, if $\result_1$ is the first accepting branch of $\treecont_1$ and $\result_2$ of $\treecont_2$, then both
$\ptres{\treechoice{\treecont_1}{\treecont_2}}{\result_1}{\{\treecont_1\}}{\inp}$ and
$\ptres{\treechoice{\treecont_1}{\treecont_2}}{\result_2}{\{\treecont_1\}}{\inp}$ hold.

We extend this definition to ordered lists of trees: $\ptres{\pactive}{\result}{\seen}{\inp}$ means that $\result$ is one possible first result of the list $\pactive$ after removing any number of subtrees in $\seen$.
We can further extend this definition to PikeTree semantics states (where $\seqop{r_1}{r_2}$ is equal to $r_1$ when it is different from $\None$, and $r_2$ otherwise):
\[\ptres{\ptstate{\inp}{\best}{\pactive}{\blocked}{\seen}}{\result}{}{} \defequal
\exists \resulta,\: \resultb.\;
\ptres{\blocked}{\resultb}{\emptyset}{\nextinp{\inp}} \wedge
\ptres{\pactive}{\resulta}{\seen}{\inp} \wedge
\result = \seqop{\resultb}{\seqop{\resulta}{\best}}\]

While this relation is non-deterministic (because we can choose whether to skip subtrees in $\seen$), it turns out that the PikeTree algorithm can only evaluate to a single result.
Our final invariant,  $\piketreeinv{\pts}{\result}$, captures that fact.
We define the invariant below, then prove that the invariant is correctly initialized (\Autoref{thm:piketree_init}) and preserved (\Autoref{thm:piketree}).
\[
\infer[]{\piketreeinv{\ptsfinal{\best}}{\best}}{}\hspace{1cm}
\infer[]
{\piketreeinv{\ptstate{\inp}{\best}{\pactive}{\blocked}{\seen}}{\result}}
{\forall r.\; \ptres{\ptstate{\inp}{\best}{\pactive}{\blocked}{\seen}}{r}{}{} \implies r=\result}
\]

\begin{theorem}[Invariant initialization]
  \label{thm:piketree_init}
  $\forall\: \inp,\: \insub{\treecont}.\;$
  $\piketreeinv{\ptsinit{\treecont}{\inp}}{\firstbranch{\treecont}{\inp}}$.
\end{theorem}
\begin{theorem}[Preservation]
  \label{thm:piketree}
  $\forall\: \pts_1,\: \pts_2,\: \result.\;$
  $\piketreeinv{\pts_1}{\result} \wedge \piketreestep{\pts_1}{\pts_2} \implies \piketreeinv{\pts_2}{\result}$.
\end{theorem}
\begin{proof}
  For \Autoref{thm:piketree_init}, the initial state $\ptsinit{\treecont}{\inp}$, has an empty $\seen$ set.
  As a result, the only possible result is the first accepting branch of $t$.
  The proof for \Autoref{thm:piketree} proceeds by case analysis over $\piketreestep{\pts_1}{\pts_2}$.
  Informally, when the algorithm skips a previously seen tree $\treecont$ (\ruleref{piketree}{Skip}), we know that
  all results of the new state are results of the previous state where $\treecont$ was skipped.
  In other cases, when visiting a new tree $\treecont$, PikeTree not only adds $\treecont$ to the $\seen$ set, but also adds the children of $\treecont$ either in $\pactive$ or in $\blocked$.
  Because of the addition to the  $\seen$ set, PikeTree might skip more branches (corresponding to $\treecont$) in the future.
  However, we prove that skipping these branches will not change the result of the algorithm, since any result in these branches is also a result of the children of $\treecont$.
\end{proof}

\subsection{Correctness of the PikeVM Compiler}
\label{subsec:pikevm_piketree}

We now prove that each execution of the PikeVM algorithm itself, as described on \Autoref{fig:pikevm_smallstep}, corresponds to an execution of the PikeTree algorithm.
Informally, at all times during the execution, each thread of the active or blocked list in the PikeVM state corresponds to a tree in the active or blocked list of the PikeTree algorithm (see \Autoref{subfig:exec}).
We then need to define an equivalence relation between threads and trees.
Whenever PikeVM skips a thread $\thread{\pc}{\gm}{\bo}$ because it has seen its program counter $\pc$ and Boolean $\bo$ already, it means that PikeTree has seen the corresponding tree already and thus is allowed to skip it as well.
Most steps of the PikeVM algorithm directly correspond to a similar step in the PikeTree algorithm, with a few exceptions explained below.

First, we formalize this notion of equivalence between threads and trees on \Autoref{fig:tree_thread}.
The first rule corresponds to the more common case.
A thread $\thread{\pc}{\gm}{\bo}$ is equivalent to a tree $\treecont$ when there exists a list of actions $\cont$, such that $\treecont$ is the tree of $\cont$ for the current input $\inp$ and Boolean $\bo$, and this list of actions is \textit{represented} in the code at label $\pc$.
This representation predicate $\rep{\code}{\cont}{\pc}$ (omitted here for brevity) is inductively defined to state that the bytecode corresponding to each action in $\cont$ starts at label $\pc$ in $\code$, and then points to an \accept~ instruction. The representation of an $\aclose{\gid}$ action is a \setregclose{$\gid$} instruction, and the representation of an $\acheck{\inpcheck}$ action is an \iendloop{} instruction. In-between the bytecode representation of each action, there may be a \jmp{} instruction, in case we are representing the first branch of a disjunction (see \Autoref{fig:pikevm_compile}).

The invariant also contains two more rules to cover cases where the thread does not directly correspond to the tree of a list of actions.
This happens in two cases, when the thread is at a \resetregs{} or a \beginloop~ instruction.
Finally, we naturally extend this equivalence to lists of threads and trees, and we write $\treethread{\pactive}{\pactive_{VM}}{\code}{\inp}$ when all elements of the lists are pairwise equivalent.

\begin{figure}
\mbox{\infer[]
        {\treethread{(\treecont,\gm)}{\thread{\pc}{\gm}{\bo}}{\code}{\inp}}
        {\booltree{\cont}{\inp}{\bo}{\treecont} \newpremise \rep{\code}{\cont}{\pc}}}\hspace{1cm}%
\mbox{\infer[]
        {\treethread{(\treecont,\gm)}{\thread{\pc}{\gm}{\bo}}{\code}{\inp}}
        {\getpc{\code}{\pc} = \beginloop \newpremise    \treethread{(\treecont,\gm)}{\thread{\pc+1}{\gm}{\cannotexit}}{\code}{\inp}}}

\semspace
\mbox{\infer[]
        {\treethread{(\treereset{\gidl}{\treecont},\gm)}{\thread{\pc}{\gm}{\bo}}{\code}{\inp}}
        {\getpc{\code}{\pc} = \resetregs{\gidl} \newpremise \gmreset{\gm}{\gidl} = \gm' \newpremise \treethread{(\treecont,\gm')}{\thread{\pc+1}{\gm'}{\bo}}{\code}{\inp}}}

\caption{Equivalence between trees and PikeVM threads}%
\Description{} 
\label{fig:tree_thread}
\end{figure}

The executions of the PikeVM and the PikeTree algorithms are quite similar.
As one adds seen threads to its $\seen$ set, the other adds equivalent trees.
There is however one exception: the PikeVM algorithm sometimes performs more steps than the PikeTree algorithm.
This happens when executing bytecode instructions that do not correspond to operations recorded in the backtracking tree: \jmp{} and \beginloop.
We refer to these instructions as \textit{stuttering} instructions.
In \Autoref{subfig:exec}, when executing labels 3 and 7, PikeVM takes one more step than PikeTree.
As we execute these instructions, some threads are added to the $\seen$ set of PikeVM before the corresponding tree is added on the PikeTree side.
As a result, an invariant of the executions is that each thread $\thread{\pc}{\gm}{\bo}$ of the $\seen_{VM}$ set of PikeVM is either equivalent to a tree in the $\seen$ set of the PikeTree algorithm, or $\getpc{\code}{\pc}$ is a stuttering instruction and $\thread{\pc}{\gm}{\bo}$ is equivalent to the current active tree of the PikeTree algorithm (the head of the $\pactive$ list).
We write $\seenincl{\seen_{VM}}{\seen}{\pactive}$ to express this property.
Finally, we can define on \Autoref{fig:pikevm_inv} the invariant relating the execution of PikeVM to one execution of the PikeTree algorithm, and show that the invariant is initialized (\Autoref{thm:pikevm_init}) and preserved (\Autoref{thm:pikevm}).

\begin{figure}
\mbox{\infer[]
        {\pikeinv{\ptsfinal{\result}}{\pvsfinal{\result}}{\code}}
        {}}\hspace{.2cm}%
\mbox{\infer[] {\pikeinv{\ptstate{\inp}{\best}{\pactive}{\blocked}{\seen}}{\pvstate{\inp}{\best}{\pactive_{VM}}{\blocked_{VM}}{\seen_{VM}}}{\code}}
        {\treethread{\pactive}{\pactive_{VM}}{\code}{\inp}
         \newpremise \treethread{\blocked}{\blocked_{VM}}{\code}{\nextinp{\inp}}
         \newpremise \seenincl{\seen_{VM}}{\seen}{\pactive}}}
\caption{Equivalence between PikeTree semantic states and PikeVM semantic states}%
\Description{} 
\label{fig:pikevm_inv}
\end{figure}

\begin{theorem}[PikeVM invariant initialization]
  \label{thm:pikevm_init}
  $\forall\: \insub{\subreg},\: \treecont,\: \code,\: \inp.\;$\\
  $\compilation{\subreg} = \code \wedge \booltree{[\areg{\re{\subreg}}]}{\inp}{\canexit}{\treecont} \implies$
  $\pikeinv{\ptsinit{\treecont}{\inp}}{\pvsinit{\inp}}{\code}$.
\end{theorem}
\begin{theorem}[PikeVM invariant preservation]
  \label{thm:pikevm}
  $\forall\: \pts_1,\: \pvs_1,\: \pvs_2,\: \code.\;$\\
  $\pikeinv{\pts_1}{\pvs_1}{\code} \wedge \pikevmstep{\pvs_1}{\pvs_2}{\code} \implies$
  $(\exists \pts_2.\; \piketreestep{\pts_1}{\pts_2} \wedge \pikeinv{\pts_2}{\pvs_2}{\code}) \vee (\pikeinv{\pts_1}{\pvs_2}{\code})$.
\end{theorem}
\begin{proof}
For \Autoref{thm:pikevm_init}, we show that the list of actions $[\areg{\re{\subreg}}]$ is represented at the label 0 of its compiled code $\code$ by induction over $\subreg$.
Then, the proof of \Autoref{thm:pikevm}  proceeds by case analysis over $\pikevmstep{\pvs_1}{\pvs_2}{\code}$.
Most cases are proved by induction over the $\booltree{\cont}{\inp}{\bo}{\treecont}$ predicate relating the current active tree $\treecont$ being executed by PikeTree and the actions $\cont$ represented at the current active thread executed by PikeVM.
An induction is required because this list of actions could start with an arbitrary sequence of \re{\regepsilon} actions before the action responsible for the current operation in the tree and current instruction in the bytecode.
By construction of the Boolean semantics (\Autoref{fig:bool_semantics}), we deduce that progress checks on the PikeTree side match the checks on the PikeVM side.
When PikeVM decides to skip a thread, it means this thread was in its $\seen$ set.
From the inclusion between $\seen$ sets, we know that there exists an equivalent tree that can be skipped on the PikeTree side.
\end{proof}

Finally, we can prove that PikeVM returns the same result as the ECMAScript specification (\Autoref{thm:pikevm_final}) by first proving that PikeVM returns a result of PikeTree (\Autoref{thm:pikevm_trc}), where $\trc$ represents the transitive reflexive closure of a small-step relation.

\begin{theorem}[PikeVM to PikeTree execution]
  \label{thm:pikevm_trc}
  $\forall\: \insub{r},\: \inp,\: \treecont,\: \result.\;$\\
  $\booltree{[\areg{\re{\subreg}}]}{\inp}{\canexit}{\treecont} \wedge {}$
  $\pikevmstar{\pvsinit{\inp}}{\pvsfinal{\result}}{\compilation{\subreg}} \implies$
  $\piketreestar{\ptsinit{\treecont}{\inp}}{\ptsfinal{\result}}$
\end{theorem}
\begin{theorem}[PikeVM correctness theorem]
  \label{thm:pikevm_final}
  $\forall\: \iswf{\subreg_w},\: \inp,\: \result.\;$\\
  $\insub{\tolinden{\subreg_w}} \wedge$
  $\pikevmstar{\pvsinit{\inp}}{\pvsfinal{\result}}{\compilation{\tolinden{\subreg_w}}} \implies$
  $\warblrecompile{\subreg_w}(\strof{\inp},\idx{\inp}) = \towarblre{\result}$.
\end{theorem}
\begin{proof}
  \Autoref{thm:pikevm_trc} is proved by induction over the derivation of $\pikevmstar{\pvsinit{\inp}}{\pvsfinal{\result}}{\compilation{\subreg}}$, applying \Autoref{thm:pikevm} at each step.
  Our final \Autoref{thm:pikevm_final} then follows from previous results.
  From \Autoref{thm:pikevm_trc}, we know that each execution of the PikeVM algorithm on a regex and an input corresponds to an execution of the PikeTree algorithm on the corresponding Boolean tree.
  From \Autoref{thm:boolean}, we know that this corresponding Boolean tree is the backtracking tree of that regex and that input.
  From Theorems~\ref{thm:piketree} and~\ref{thm:piketree_init}, we know that each result of the PikeTree algorithm corresponds to the first accepting branch of that tree.
  Finally, from \Autoref{thm:warblre_equiv}, we know that this first accepting branch is precisely the result defined by the ECMAScript standard.
\end{proof}

Despite the PikeVM algorithm having been used widely for decades, this is the first time that it has been formally proved to return the top priority result according to backtracking semantics.
Reasoning on backtracking trees greatly facilitates the proof as it enables simple ways to express crucial properties.
For instance, proving the correctness of skipping branches would be much harder to prove directly on the PikeVM algorithm: two executions of the same bytecode instruction might yield completely different results when the $\seen$ set is different.

\section{Discussion and Related Work}
\label{sec:related}

\newcommand{\related}[1]{\tightparagraph{#1.}}

\related{Adapting our work to other languages}
While our semantics matches that of the JavaScript regexes, it could be straightforwardly modified for other languages with backtracking semantics.
We chose JavaScript for its wide use and because it is the only modern backtracking semantics regex language to come with a full mechanized specification that we can formally compare ourselves to.

For the PikeVM, we could expect the following changes.
In languages without capture reset, we could simply remove the $\treereset{}{}$ node from the tree type, and the PikeVM would not generate any \resetregs{} instruction.
In languages without the special nullable quantifier semantics of JavaScript, we would remove the $\acheck{}$ action, and make all nodes of the tree check for progress.
The PikeVM would become simpler: we could remove the Boolean from threads entirely, and remove instructions \beginloop~ and \iendloop{}. In the proof, we would remove the Boolean semantics of \Autoref{subsec:boolean} entirely.
In both cases, there would be one less case in \Autoref{fig:tree_thread}.

For contextual equivalence, most of our proofs of \Autoref{sec:rewrite} do not depend on JavaScript specificities (with the exception of anchor rewriting equivalences).
We expect them to hold in any backtracking semantics language.
Counter-examples are easier to generalize as they can be confirmed with tests, we discuss them in \appendixref{app:counterex}.

\subsection{Related Work}
With the exception of Warblre, mechanized semantics and verification work typically exclude backtracking semantics and capture groups.
This is known to be difficult: according to~\citet{lean_lookarounds}, ``\textit{Formalizing even the core aspect of the capture semantics and algorithms [\dots] is therefore a major undertaking with many challenges}''.
Our semantics and proofs provide a solution.

\related{Warblre}
Warblre~\cite{warblre_icfp} was the first mechanization, in an interactive theorem prover, of a modern, general-purpose regex language (ECMAScript 2023~\cite[\S 22.2]{ecma_2023}).  By closely following the paper specification, it traded succintness and ease-of-use for high auditability and faithfulness; accordingly, its authors described it as ``\textit{a foundation for researchers to restate the semantics in a way that better suits their field}''.
Our \autoref{sec:semantics} does that, and provides a new, complete formalization of ECMAScript 2023 \S 22.2.  It preserves the desirable properties of Warblre (in particular, \Autoref{thm:warblre_equiv} shows faithfulness), but without its shortcomings: our tree semantics are more succinct, and much easier to reason about (we provide an induction principle, and we eliminate the burden of reasoning about the error monad).

\related{JavaScript regex formalizations}
Several previous efforts have developed unmechanized formal semantics for JavaScript regexes.
\citet{expose} presented a formalized semantics for ECMAScript 2015 regexes and used it for symbolic execution of programs with regexes.
Similarly, \citet{black_ostrich} augmented the Ostrich constraint solver with a model of ECMAScript regexes, using two-way alternating automata.
Recently, \citet{regex_repair,regex_repair2} presented a formal semantics for JavaScript regexes used for regex repair.
Their semantic statement includes a matching direction and a group map just like ours, and their \textit{continuation regex} resembles our list of actions, but they only formalize the top-priority match.
Finally, \citet{psst} formalized a large subset of JavaScript regex semantics using prioritized streaming string transducers.
While there are mechanizations of the JavaScript language~\cite{jiset,jscert}, none include the regular expression chapter of ECMAScript.
In fact, besides Warblre, our semantics is the only mechanized one for JavaScript regexes, and also the only complete one: other formalizations exclude some JavaScript features such as lookbehinds, regex flags, or backreferences.
Manually defining formal semantics for a modern regex language is a complex and error-prone task.
For instance, careful inspection of the semantics proposed by~\citet[Extended version, Table 5]{black_ostrich} reveals some mistakes:
capture reset is not performed, and the semantics of lookarounds is incorrect.
JavaScript lookarounds are supposed to be atomic, meaning that once a match is found for a lookaround, an engine cannot backtrack and match it in a different way~\cite[Section 22.2.2.4, Note 3]{ecma_2025}.
Such subtle differences are easy to miss (this behavior of lookarounds can only be observed with a backreference to a capture group defined inside a positive lookaround); by contrast, our proof of \Autoref{thm:warblre_equiv} ensures the correctness of the semantics.
Still, mechanization also entails specific challenges:
for instance, some paper formalizations use \textit{non-strict positive occurrences}
by referring to the negation of matching within the definition of the matching relation itself (for instance for negative lookarounds, or for the second branch of the disjunction).
Such definitions would have to be adjusted to be accepted by Rocq~\cite{lean_lookarounds}.

For a different language, the original PCRE documentation~\cite{pcre_manual} notes that
``\textit{The set of strings that are matched by a regular expression can be represented as a tree structure [...] of infinite size}'' and mentions both depth-first and breadth-first traversal of the tree.
We have found that specializing the tree to a particular input provides a finite tree that is more practical for formal verification than a coinductive one, and that the PikeVM algorithm is more than a breadth-first traversal: it alternates between depth-first and breadth-first order to maintain priority.

\related{Formalizations of modern features}
Other work has also verified or formalized individual modern features, but in simplified languages.
For instance, matching algorithms supporting lookarounds have recently been mechanized in Lean~\cite{lean_lookarounds} and Rocq~\cite{verified_lookarounds}.
However, these works include neither capture groups nor backtracking semantics.
Other unmechanized work specify capture groups using prioritized transducers~\cite{groups_transducers} or encode lookarounds as alternating automata~\cite{lookarounds_afa}.

\related{PikeVM verification}
Recently, the \texttt{lean-regex} regex matcher has also been implemented and verified in Lean~\cite{lean_regex}.
This project also includes a PikeVM implementation, along with a proof.
However, the proof differs from our \autoref{sec:pikevm} in some ways.
First, its theorem states that the algorithm returns \textit{one of the possible matches} of the input regex on the input string.
In contrast, in our work we consider that the correctness of the PikeVM requires a stronger property: not only that it returns a result when it should, but also that it returns the top-priority one according to backtracking semantics.
Second, our theorem also relates the result of the PikeVM to a mechanized translation of a real-world regex semantics, Warblre.
On the other hand, the \texttt{lean-regex} formalization can be executed and uses efficient data-structure implementations and optimizations, while our work stops at algorithm verification.

\related{Traditional regular expressions}
Traditional regular expressions and finite state automata, including NFA simulation, have been verified and mechanized many times~\cite{coq_regular_languages, regular_representations, decision_procedure, functional_to_hol, deciding_kat, regular_sets_afp, validating_lr1_parsers, verified_operational_semantics, certified_parsing}.
Regular expression matching algorithms have been formally verified for lexers~\cite{verbatim,verbatim_pp,stainless_lexer,coqlex} and in particular for \textit{leftmost-longest} semantics~\cite{isabelle_posix_lexing,posix_lexing_itp}, but not for backtracking semantics.

\subsection{Future Work}
Our semantics is verified to be equivalent to the latest version of Warblre, based on the 2023 edition of ECMAScript~\cite{ecma_2023}.
As of 2025, two new editions have been released~\cite{ecma_2024,ecma_2025}, introducing new regex features: the \texttt{v} Unicode flag, modifiers, and duplicate named groups.
If Warblre were to be updated, we could adapt accordingly.
First, the new \flag{v} flag adds new features in character descriptors that we could support by adapting the $\inpadvance{\cd}{\inp}{\dir}$ function.
Second, new modifiers allow flags \flag{i}, \flag{m} and \flag{s} to be set locally in a subregex.
Currently, we pass the record containing the value of each flag as parameter of the $\istree{\cont}{\inp}{\gm}{\dir}{\treecont}$ predicate.
To support local modifiers, we would pass this record as an index instead, and add a corresponding type of node in backtracking trees.
Finally, duplicate named groups allow some capture group identifiers to appear several times in distinct disjunction branches.  If Warblre's $\wellformed$ set were updated, we expect that we could accordingly relax our definition of contextual equivalence to allow equivalent regexes to have duplicated named groups.

While we have proved the \textit{correctness} of the PikeVM algorithm, we have not verified an implementation, let alone an efficient one.
A realistic engine would require efficient data structures (e.g., for capture groups and group maps), matching optimizations (e.g., prefix acceleration), and a specification and implementation of top-level APIs (e.g., \texttt{matchAll}).
Other bugs in implementations can lie in their complexity.
Future work could prove that the PikeVM semantics terminates in a number of steps linear in the size of the input.

It would also be interesting to extend the base PikeVM algorithm with additional features,
to obtain a formally verified linear-time algorithm supporting all JavaScript regex features except backreferences.
We expect that adapting our proofs would be straightforward for most missing features.
To support anchors, existing PikeVM implementations compile them to a new bytecode instruction that checks the surroundings of the current string position.
Supporting the question mark quantifiers can be done with only \fork{}{}, \beginloop~ and \iendloop{} instructions.
Then, generic quantifiers \re{\quant{\subreg}{\rmin}{\Delta}{p}} can be handled by duplicating the bytecode of \re{\subreg} (for instance, compiling \re{\quant{\subreg}{1}{1}{\top}} concatenates the bytecodes of \re{\subreg} and \re{\regqm{r}}).
A more challenging extension would consist in supporting lookaorunds, following the algorithm of~\citet{regelk_pldi}.

Finally, other work has also explored other extended features that are not part of JavaScript, such as
\textit{atomic groups}~\cite{memo_atomic_groups}, the \textit{shuffle} operator~\cite{derivatives_enhanced},
or verifying in Idris 2~\cite{idris_tyre} the type-safety of an engine supporting \textit{typed regexes}~\cite{tyre}.
Other work goes beyond capture groups and instead returns \textit{parse trees} for standard regular expressions, indicating the full list of substrings a subexpression has matched~\cite{greedy_matching, bitcoded_parsing, certified_bitcoded_parsing, efficient_parse_tree}.
An interesting future direction would be to extend our work to support these extensions and alternative semantics.

\section{Conclusion}
\label{sec:conclusion}

We have presented a new semantics for JavaScript regexes, that is mechanized in Rocq, complete yet succinct, proven to be faithful to the ECMAScript standard, and practical for formal verification.
We have showed that formalizing not only the first match of a regex, but also lower-priority matches, makes the semantics practical.
We have used it to provide novel proofs of real-world applications: the first verification of the PikeVM linear matching algorithm, used in many deployed engines; and a new formal notion of contextual equivalence that allowed us to prove and disprove regex equivalences from the literature and from the optimizer of a popular JavaScript regex manipulation library.
Thanks to the verified connection with an audited formalization, formal proofs conducted with our semantics can be trusted to correspond exactly to the behavior specified by ECMAScript.
Our work lays the foundation for the development of verified and realistic modern regex engines.

\section*{Acknowledgments}
We thank Yann Herklotz, Martin Odersky, and Marcin Wojnarowski for feedback on this paper, and Eugène Flesselle for preliminary exploration of regex equivalence.
This research was funded in whole or in part by the Swiss National Science Foundation (SNSF), grant number 10003649.

\section*{Artifact Availability}
Our development is free software and can be accessed online: \url{https://github.com/epfl-systemf/Linden}.
We have also packaged all the definitions and proofs presented in this paper in an artifact~\cite{linden_artifact}.
The artifact offers both the standalone Rocq files, and a virtual machine image (along with a script to create this image from scratch).

\bibliographystyle{ACM-Reference-Format}
\bibliography{main}

@article{regelk_pldi,
  author       = {Aur{\`{e}}le Barri{\`{e}}re and
                  Cl{\'{e}}ment Pit{-}Claudel},
  title        = {Linear Matching of JavaScript Regular Expressions},
  journal      = {Proc. {ACM} Program. Lang.},
  volume       = {8},
  number       = {{PLDI}},
  pages        = {1336--1360},
  year         = {2024},
  url          = {https://doi.org/10.1145/3656431},
  doi          = {10.1145/3656431},
  bibsource    = {dblp computer science bibliography, https://dblp.org}
}

@article{warblre_icfp,
  author       = {De Santo, No{\'{e}} and
                  Aur{\`{e}}le Barri{\`{e}}re and
                  Cl{\'{e}}ment Pit{-}Claudel},
  title        = {A Coq Mechanization of JavaScript Regular Expression Semantics},
  journal      = {Proc. {ACM} Program. Lang.},
  volume       = {8},
  number       = {{ICFP}},
  pages        = {1003--1031},
  year         = {2024},
  url          = {https://doi.org/10.1145/3674666},
  doi          = {10.1145/3674666},
  bibsource    = {dblp computer science bibliography, https://dblp.org}
}

@software{linden_artifact,
  author       = {Barrière, Aurèle and
                  Deng, Victor and
                  Pit-Claudel, Clément},
  title        = {Artifact for "Formal Verification for JavaScript
                   Regular Expressions: a Proven Mechanized Semantics
                   and its Applications", at POPL 2026
                  },
  month        = oct,
  year         = 2025,
  publisher    = {Zenodo},
  doi          = {10.5281/zenodo.17305393},
  url          = {https://doi.org/10.5281/zenodo.17305393},
}

@inproceedings{expose,
  author = {Blake Loring and Duncan Mitchell and Johannes Kinder},
  editor = {Kathryn S. McKinley and Kathleen Fisher},
  title = {Sound Regular Expression Semantics for Dynamic Symbolic Execution of
           {JavaScript}},
  booktitle = {Proceedings of the 40th {ACM} {SIGPLAN} Conference on Programming
               Language Design and Implementation, {PLDI} 2019, Phoenix, AZ, USA,
               June 22-26, 2019},
  pages = {425--438},
  publisher = {{ACM}},
  year = {2019},
  doi = {10.1145/3314221.3314645},
  bibsource = {dblp computer science bibliography, https://dblp.org},
}

@article{regex_repair,
  author       = {Nariyoshi Chida and
                  Tachio Terauchi},
  title        = {Repairing Regular Expressions for Extraction},
  journal      = {Proceedings of the {ACM} on Programming Languages},
  volume       = {7},
  number       = {{PLDI}},
  pages        = {1633--1656},
  year         = {2023},
  url          = {https://doi.org/10.1145/3591287},
  doi          = {10.1145/3591287},
  bibsource    = {dblp computer science bibliography, https://dblp.org}
}

@inproceedings{regex_repair2,
  author       = {Nariyoshi Chida and
                  Tachio Terauchi},
  editor       = {Vladimir Filkov and
                  Baishakhi Ray and
                  Minghui Zhou},
  title        = {Repairing Regex-Dependent String Functions},
  booktitle    = {Proceedings of the 39th {IEEE/ACM} International Conference on Automated
                  Software Engineering, {ASE} 2024, Sacramento, CA, USA, October 27
                  - November 1, 2024},
  pages        = {294--305},
  publisher    = {{ACM}},
  year         = {2024},
  url          = {https://doi.org/10.1145/3691620.3695005},
  doi          = {10.1145/3691620.3695005},
  bibsource    = {dblp computer science bibliography, https://dblp.org}
}

@article{psst,
  author       = {Taolue Chen and
                  Alejandro Flores{-}Lamas and
                  Matthew Hague and
                  Zhilei Han and
                  Denghang Hu and
                  Shuanglong Kan and
                  Anthony W. Lin and
                  Philipp R{\"{u}}mmer and
                  Zhilin Wu},
  title        = {Solving string constraints with Regex-dependent functions through
                  transducers with priorities and variables},
  journal      = {Proceedings of the {ACM} on Programming Languages},
  volume       = {6},
  number       = {{POPL}},
  pages        = {1--31},
  year         = {2022},
  url          = {https://doi.org/10.1145/3498707},
  doi          = {10.1145/3498707},
  bibsource    = {dblp computer science bibliography, https://dblp.org}
}

@inproceedings{lean_lookarounds,
  author       = {Ekaterina Zhuchko and
                  Margus Veanes and
                  Gabriel Ebner},
  editor       = {Amin Timany and
                  Dmitriy Traytel and
                  Brigitte Pientka and
                  Sandrine Blazy},
  title        = {Lean Formalization of Extended Regular Expression Matching with Lookarounds},
  booktitle    = {Proceedings of the 13th {ACM} {SIGPLAN} International Conference on
                  Certified Programs and Proofs, {CPP} 2024, London, UK, January 15-16,
                  2024},
  pages        = {118--131},
  publisher    = {{ACM}},
  year         = {2024},
  url          = {https://doi.org/10.1145/3636501.3636959},
  doi          = {10.1145/3636501.3636959},
  bibsource    = {dblp computer science bibliography, https://dblp.org}
}

@article{efficient_lookarounds,
  author       = {Konstantinos Mamouras and
                  Agnishom Chattopadhyay},
  title        = {Efficient Matching of Regular Expressions with Lookaround Assertions},
  journal      = {Proc. {ACM} Program. Lang.},
  volume       = {8},
  number       = {{POPL}},
  pages        = {2761--2791},
  year         = {2024},
  url          = {https://doi.org/10.1145/3632934},
  doi          = {10.1145/3632934},
  bibsource    = {dblp computer science bibliography, https://dblp.org}
}

@inproceedings{verified_lookarounds,
  author       = {Agnishom Chattopadhyay and
                  Wu Angela Li and
                  Konstantinos Mamouras},
  editor       = {Kathrin Stark and
                  Amin Timany and
                  Sandrine Blazy and
                  Nicolas Tabareau},
  title        = {Verified and Efficient Matching of Regular Expressions with Lookaround},
  booktitle    = {Proceedings of the 14th {ACM} {SIGPLAN} International Conference on
                  Certified Programs and Proofs, {CPP} 2025, Denver, CO, USA, January
                  20-21, 2025},
  pages        = {198--213},
  publisher    = {{ACM}},
  year         = {2025},
  url          = {https://doi.org/10.1145/3703595.3705884},
  doi          = {10.1145/3703595.3705884},
  bibsource    = {dblp computer science bibliography, https://dblp.org}
}

@inproceedings{decision_procedure,
  author    = {Thierry Coquand and
               Vincent Siles},
  title     = {A Decision Procedure for Regular Expression Equivalence in Type Theory},
  booktitle = {Certified Programs and Proofs - First International Conference, {CPP}
               2011},
  year      = {2011},
  doi       = {10.1007/978-3-642-25379-9_11},
  bibsource = {dblp computer science bibliography, https://dblp.org}
}

@inproceedings{coq_regular_languages,
  author       = {Christian Doczkal and
                  Jan{-}Oliver Kaiser and
                  Gert Smolka},
  editor       = {Georges Gonthier and
                  Michael Norrish},
  title        = {A Constructive Theory of Regular Languages in Coq},
  booktitle    = {Certified Programs and Proofs - Third International Conference, {CPP}
                  2013, Melbourne, VIC, Australia, December 11-13, 2013, Proceedings},
  series       = {Lecture Notes in Computer Science},
  volume       = {8307},
  pages        = {82--97},
  publisher    = {Springer},
  year         = {2013},
  url          = {https://doi.org/10.1007/978-3-319-03545-1\_6},
  doi          = {10.1007/978-3-319-03545-1\_6},
  bibsource    = {dblp computer science bibliography, https://dblp.org}
}

@article{regular_representations,
  author       = {Christian Doczkal and
                  Gert Smolka},
  title        = {Regular Language Representations in the Constructive Type Theory of
                  Coq},
  journal      = {J. Autom. Reason.},
  volume       = {61},
  number       = {1-4},
  pages        = {521--553},
  year         = {2018},
  url          = {https://doi.org/10.1007/s10817-018-9460-x},
  doi          = {10.1007/S10817-018-9460-X},
  bibsource    = {dblp computer science bibliography, https://dblp.org}
}

@inproceedings{tyre,
  author       = {Gabriel Radanne},
  editor       = {Manuel V. Hermenegildo and
                  Atsushi Igarashi},
  title        = {Typed parsing and unparsing for untyped regular expression engines},
  booktitle    = {Proceedings of the 2019 {ACM} {SIGPLAN} Workshop on Partial Evaluation
                  and Program Manipulation, PEPM@POPL 2019, Cascais, Portugal, January
                  14-15, 2019},
  pages        = {35--46},
  publisher    = {{ACM}},
  year         = {2019},
  url          = {https://doi.org/10.1145/3294032.3294082},
  doi          = {10.1145/3294032.3294082},
  bibsource    = {dblp computer science bibliography, https://dblp.org}
}

@article{idris_tyre,
  author       = {Ohad Kammar and
                  Katarzyna Marek},
  title        = {Idris TyRE: a dependently typed regex parser},
  journal      = {CoRR},
  volume       = {abs/2305.04480},
  year         = {2023},
  url          = {https://doi.org/10.48550/arXiv.2305.04480},
  doi          = {10.48550/ARXIV.2305.04480},
  eprinttype    = {arXiv},
  eprint       = {2305.04480},
  bibsource    = {dblp computer science bibliography, https://dblp.org}
}

@inproceedings{deciding_kat,
  author       = {Thomas Braibant and
                  Damien Pous},
  editor       = {Matt Kaufmann and
                  Lawrence C. Paulson},
  title        = {An Efficient Coq Tactic for Deciding Kleene Algebras},
  booktitle    = {Interactive Theorem Proving, First International Conference, {ITP}
                  2010, Edinburgh, UK, July 11-14, 2010. Proceedings},
  series       = {Lecture Notes in Computer Science},
  volume       = {6172},
  pages        = {163--178},
  publisher    = {Springer},
  year         = {2010},
  url          = {https://doi.org/10.1007/978-3-642-14052-5\_13},
  doi          = {10.1007/978-3-642-14052-5\_13},
  bibsource    = {dblp computer science bibliography, https://dblp.org}
}

@article{isabelle_posix_lexing,
  author       = {Christian Urban},
  title        = {{POSIX} Lexing with Derivatives of Regular Expressions},
  journal      = {J. Autom. Reason.},
  volume       = {67},
  number       = {3},
  pages        = {24},
  year         = {2023},
  url          = {https://doi.org/10.1007/s10817-023-09667-1},
  doi          = {10.1007/S10817-023-09667-1},
  bibsource    = {dblp computer science bibliography, https://dblp.org}
}

@inproceedings{posix_lexing_itp,
  author       = {Fahad Ausaf and
                  Roy Dyckhoff and
                  Christian Urban},
  editor       = {Jasmin Christian Blanchette and
                  Stephan Merz},
  title        = {{POSIX} Lexing with Derivatives of Regular Expressions (Proof Pearl)},
  booktitle    = {Interactive Theorem Proving - 7th International Conference, {ITP}
                  2016, Nancy, France, August 22-25, 2016, Proceedings},
  series       = {Lecture Notes in Computer Science},
  volume       = {9807},
  pages        = {69--86},
  publisher    = {Springer},
  year         = {2016},
  url          = {https://doi.org/10.1007/978-3-319-43144-4\_5},
  doi          = {10.1007/978-3-319-43144-4\_5},
  bibsource    = {dblp computer science bibliography, https://dblp.org}
}

@inproceedings{verbatim,
  author       = {Derek Egolf and
                  Sam Lasser and
                  Kathleen Fisher},
  title        = {Verbatim: {A} Verified Lexer Generator},
  booktitle    = {{IEEE} Security and Privacy Workshops, {SP} Workshops 2021, San Francisco,
                  CA, USA, May 27, 2021},
  pages        = {92--100},
  publisher    = {{IEEE}},
  year         = {2021},
  url          = {https://doi.org/10.1109/SPW53761.2021.00022},
  doi          = {10.1109/SPW53761.2021.00022},
  bibsource    = {dblp computer science bibliography, https://dblp.org}
}

@inproceedings{verbatim_pp,
  author       = {Derek Egolf and
                  Sam Lasser and
                  Kathleen Fisher},
  editor       = {Andrei Popescu and
                  Steve Zdancewic},
  title        = {Verbatim++: verified, optimized, and semantically rich lexing with
                  derivatives},
  booktitle    = {{CPP} '22: 11th {ACM} {SIGPLAN} International Conference on Certified
                  Programs and Proofs, Philadelphia, PA, USA, January 17 - 18, 2022},
  pages        = {27--39},
  publisher    = {{ACM}},
  year         = {2022},
  url          = {https://doi.org/10.1145/3497775.3503694},
  doi          = {10.1145/3497775.3503694},
  bibsource    = {dblp computer science bibliography, https://dblp.org}
}

@article{functional_to_hol,
  author       = {Scott Owens and
                  Konrad Slind},
  title        = {Adapting functional programs to higher order logic},
  journal      = {High. Order Symb. Comput.},
  volume       = {21},
  number       = {4},
  pages        = {377--409},
  year         = {2008},
  url          = {https://doi.org/10.1007/s10990-008-9038-0},
  doi          = {10.1007/S10990-008-9038-0},
  bibsource    = {dblp computer science bibliography, https://dblp.org}
}

@inproceedings{freezing_the_web,
  author    = {Cristian{-}Alexandru Staicu and
               Michael Pradel},
  title     = {Freezing the Web: {A} Study of {ReDoS} Vulnerabilities in {JavaScript}-based Web Servers},
  booktitle = {27th {USENIX} Security Symposium, {USENIX} Security 2018},
  pages     = {361--376},
  publisher = {{USENIX} Association},
  year      = {2018},
  url       = {https://www.usenix.org/conference/usenixsecurity18/presentation/staicu},
  bibsource = {dblp computer science bibliography, https://dblp.org}
}

@misc{re2,
  author = {{Google}},
  title = {{RE2}: {A} fast, safe, thread-friendly alternative to backtracking regular expression engines like those used in PCRE, Perl, and Python},
  year = {2010},
  howpublished = {\url{https://github.com/google/re2}}
  }

@inproceedings{hyperscan,
  author       = {Xiang Wang and
                  Yang Hong and
                  Harry Chang and
                  KyoungSoo Park and
                  Geoff Langdale and
                  Jiayu Hu and
                  Heqing Zhu},
  editor       = {Jay R. Lorch and
                  Minlan Yu},
  title        = {Hyperscan: {A} Fast Multi-pattern Regex Matcher for Modern CPUs},
  booktitle    = {16th {USENIX} Symposium on Networked Systems Design and Implementation,
                  {NSDI} 2019, Boston, MA, February 26-28, 2019},
  pages        = {631--648},
  publisher    = {{USENIX} Association},
  year         = {2019},
  url          = {https://www.usenix.org/conference/nsdi19/presentation/wang-xiang},
  bibsource    = {dblp computer science bibliography, https://dblp.org}
}

@Misc{rust_regex,
  author       = {Andrew Gallant},
  title        = {Crate {regex}: {An} implementation of regular expressions for {Rust}},
  howpublished = {\url{https://docs.rs/regex/latest/regex/}},
  year         = 2014
}

@Misc{go_regexp,
  author       = {Go},
  title        = {GO Regexp package},
  howpublished = {\url{https://pkg.go.dev/regexp}},
  year         = 2025
}

@misc{nphard,
  author = {Mark Jason Dominus},
  title = {Perl Regular Expression Matching is {NP}-Hard},
  howpublished= {\url{https://perl.plover.com/NPC/NPC-3SAT.html}},
  year = {2000}
}

@article{dotnet_pldi,
  author       = {Dan Moseley and
                  Mario Nishio and
                  Jose Perez Rodriguez and
                  Olli Saarikivi and
                  Stephen Toub and
                  Margus Veanes and
                  Tiki Wan and
                  Eric Xu},
  title        = {Derivative Based Nonbacktracking Real-World Regex Matching with Backtracking
                  Semantics},
  journal      = {Proc. {ACM} Program. Lang.},
  volume       = {7},
  number       = {{PLDI}},
  pages        = {1026--1049},
  year         = {2023},
  url          = {https://doi.org/10.1145/3591262},
  doi          = {10.1145/3591262},
}

@misc{non_backtrack_v8,
  author = {{V8}},
  year = {2021},
  title = {An Additional Non-backtracking RegExp Engine},
  howpublished = {\url{https://v8.dev/blog/non-backtracking-regexp}}
  }

@misc{regexp_vm_approach,
  author = {Russ Cox},
  year = {2009},
  title = {Regular Expression Matching: the Virtual Machine Approach},
  howpublished = {\url{https://swtch.com/~rsc/regexp/regexp2.html}}
  }

@article{thompson68,
  author    = {Ken Thompson},
  title     = {Regular Expression Search Algorithm},
  journal   = {Commun. {ACM}},
  volume    = {11},
  number    = {6},
  pages     = {419--422},
  year      = {1968},
  doi       = {10.1145/363347.363387},
  bibsource = {dblp computer science bibliography, https://dblp.org}
}

@misc{lookbehind_priority_v8,
  author = {{V8}},
  title = {RegExp: Inconsistent branch priority in look-behind assertions},
  year = 2025,
  howpublished = {\url{https://issues.chromium.org/issues/388290816}}
  }

@misc{incorrect_result_v8,
  author = {{V8}},
  title = {Mismatch between the Experimental engine and the backtracking engine on empty repetitions},
  year = 2023,
  howpublished = {\url{https://issues.chromium.org/issues/42204037}}
  }

@misc{regexp_tree,
author = {Dmitry Soshnikov},
title = {\texttt{regexp-tree}},
year = 2025,
howpublished = {\url{https://github.com/DmitrySoshnikov/regexp-tree}}
}

@inproceedings{redos_impact,
  author       = {James C. Davis and
                  Christy A. Coghlan and
                  Francisco Servant and
                  Dongyoon Lee},
  title        = {The impact of regular expression denial of service (ReDoS) in practice:
                  an empirical study at the ecosystem scale},
  booktitle    = {Proceedings of the 2018 {ACM} Joint Meeting on European Software Engineering
                  Conference and Symposium on the Foundations of Software Engineering,
                  {ESEC/SIGSOFT} {FSE} 2018},
  pages        = {246--256},
  publisher    = {{ACM}},
  year         = {2018},
  url          = {https://doi.org/10.1145/3236024.3236027},
  doi          = {10.1145/3236024.3236027},
}

@misc{ecma_2023,
  author = {{ECMA}},
  year = {2023},
  month = {June},
  title = {ECMA-262, 14th edition: ECMAScript® 2023 Language Specification},
  url = {https://262.ecma-international.org/14.0/},
}

@misc{ecma_2024,
  author = {{ECMA}},
  year = {2024},
  month = {June},
  title = {ECMA-262, 15th edition: ECMAScript® 2024 Language Specification},
  url = {https://262.ecma-international.org/15.0/},
}

@misc{ecma_2025,
  author = {{ECMA}},
  year = {2025},
  month = {June},
  title = {ECMA-262, 16th edition: ECMAScript® 2025 Language Specification},
  url = {https://262.ecma-international.org/16.0/},
}

@article{derivatives_reexamined,
  title={Regular-expression derivatives re-examined},
  author={Owens, Scott and Reppy, John and Turon, Aaron},
  journal={Journal of Functional Programming},
  volume={19},
  number={2},
  pages={173--190},
  year={2009},
  publisher={Cambridge University Press}
}

@article{derivatives_extended,
  author       = {Ian Erik Varatalu and
                  Margus Veanes and
                  Juhan{-}Peep Ernits},
  title        = {Derivative Based Extended Regular Expression Matching Supporting Intersection,
                  Complement and Lookarounds},
  journal      = {CoRR},
  volume       = {abs/2309.14401},
  year         = {2023},
  url          = {https://doi.org/10.48550/arXiv.2309.14401},
  doi          = {10.48550/ARXIV.2309.14401},
  eprinttype    = {arXiv},
  eprint       = {2309.14401},
  bibsource    = {dblp computer science bibliography, https://dblp.org}
}

@misc{spidermonkey_irregexp,
  author = {Iain Ireland},
  year = {2020},
  title = {A New RegExp Engine in {SpiderMonkey}},
  howpublished = {\url{https://hacks.mozilla.org/2020/06/a-new-regexp-engine-in-spidermonkey/}}
  }

@inproceedings{backref_expr,
  author       = {Taisei Nogami and
                  Tachio Terauchi},
  editor       = {J{\'{e}}r{\^{o}}me Leroux and
                  Sylvain Lombardy and
                  David Peleg},
  title        = {On the Expressive Power of Regular Expressions with Backreferences},
  booktitle    = {48th International Symposium on Mathematical Foundations of Computer
                  Science, {MFCS} 2023, August 28 to September 1, 2023, Bordeaux, France},
  series       = {LIPIcs},
  volume       = {272},
  pages        = {71:1--71:15},
  publisher    = {Schloss Dagstuhl - Leibniz-Zentrum f{\"{u}}r Informatik},
  year         = {2023},
  url          = {https://doi.org/10.4230/LIPIcs.MFCS.2023.71},
  doi          = {10.4230/LIPICS.MFCS.2023.71},
  bibsource    = {dblp computer science bibliography, https://dblp.org}
}

@misc{pcre_manual,
  author = {Philip Hazel},
  year = 2012,
  howpublished = {\url{https://pcre.org/pcre.txt}},
  title = {{PCRE} Library Functions Manual},
}

@misc{pcre,
author = {Philip Hazel},
year = {1997},
howpublished = {\url{https://www.pcre.org/}},
title = {{PCRE} - Perl Compatible Regular Expressions}
}

@article{regular_sets_afp,
  author  = {Alexander Krauss and Tobias Nipkow},
  title   = {Regular Sets and Expressions},
  journal = {Archive of Formal Proofs},
  month   = {May},
  year    = {2010},
  note    = {\url{https://isa-afp.org/entries/Regular-Sets.html},
             Formal proof development},
  ISSN    = {2150-914x},
}

@inproceedings{validating_lr1_parsers,
  author       = {Jacques{-}Henri Jourdan and
                  Fran{\c{c}}ois Pottier and
                  Xavier Leroy},
  editor       = {Helmut Seidl},
  title        = {Validating {LR(1)} Parsers},
  booktitle    = {Programming Languages and Systems - 21st European Symposium on Programming,
                  {ESOP} 2012, Held as Part of the European Joint Conferences on Theory
                  and Practice of Software, {ETAPS} 2012, Tallinn, Estonia, March 24
                  - April 1, 2012. Proceedings},
  series       = {Lecture Notes in Computer Science},
  volume       = {7211},
  pages        = {397--416},
  publisher    = {Springer},
  year         = {2012},
  url          = {https://doi.org/10.1007/978-3-642-28869-2\_20},
  doi          = {10.1007/978-3-642-28869-2\_20},
  bibsource    = {dblp computer science bibliography, https://dblp.org}
}

@inproceedings{unified_decision_procedure,
  author       = {Tobias Nipkow and
                  Dmitriy Traytel},
  editor       = {Gerwin Klein and
                  Ruben Gamboa},
  title        = {Unified Decision Procedures for Regular Expression Equivalence},
  booktitle    = {Interactive Theorem Proving - 5th International Conference, {ITP}
                  2014, Held as Part of the Vienna Summer of Logic, {VSL} 2014, Vienna,
                  Austria, July 14-17, 2014. Proceedings},
  series       = {Lecture Notes in Computer Science},
  volume       = {8558},
  pages        = {450--466},
  publisher    = {Springer},
  year         = {2014},
  url          = {https://doi.org/10.1007/978-3-319-08970-6\_29},
  doi          = {10.1007/978-3-319-08970-6\_29},
  bibsource    = {dblp computer science bibliography, https://dblp.org}
}

@inproceedings{deciding_equivalence_coq,
  author       = {Nelma Moreira and
                  David Pereira and
                  Sim{\~{a}}o Melo de Sousa},
  editor       = {Wolfram Kahl and
                  Timothy G. Griffin},
  title        = {Deciding Regular Expressions (In-)Equivalence in Coq},
  booktitle    = {Relational and Algebraic Methods in Computer Science - 13th International
                  Conference, RAMiCS 2012, Cambridge, UK, September 17-20, 2012. Proceedings},
  series       = {Lecture Notes in Computer Science},
  volume       = {7560},
  pages        = {98--113},
  publisher    = {Springer},
  year         = {2012},
  url          = {https://doi.org/10.1007/978-3-642-33314-9\_7},
  doi          = {10.1007/978-3-642-33314-9\_7},
  bibsource    = {dblp computer science bibliography, https://dblp.org}
}

@inproceedings{equivalence_compact_proof,
  author       = {Andrea Asperti},
  editor       = {Lennart Beringer and
                  Amy P. Felty},
  title        = {A Compact Proof of Decidability for Regular Expression Equivalence},
  booktitle    = {Interactive Theorem Proving - Third International Conference, {ITP}
                  2012, Princeton, NJ, USA, August 13-15, 2012. Proceedings},
  series       = {Lecture Notes in Computer Science},
  volume       = {7406},
  pages        = {283--298},
  publisher    = {Springer},
  year         = {2012},
  url          = {https://doi.org/10.1007/978-3-642-32347-8\_19},
  doi          = {10.1007/978-3-642-32347-8\_19},
  bibsource    = {dblp computer science bibliography, https://dblp.org}
}

@article{pike_sam,
  author       = {Rob Pike},
  title        = {The Text Editor sam},
  journal      = {Softw. Pract. Exp.},
  volume       = {17},
  number       = {11},
  pages        = {813--845},
  year         = {1987},
  url          = {https://doi.org/10.1002/spe.4380171105},
  doi          = {10.1002/SPE.4380171105},
  bibsource    = {dblp computer science bibliography, https://dblp.org}
}

@inproceedings{memo_atomic_groups,
  author       = {Hiroya Fujinami and
                  Ichiro Hasuo},
  editor       = {Stephanie Weirich},
  title        = {Efficient Matching with Memoization for Regexes with Look-around and
                  Atomic Grouping},
  booktitle    = {Programming Languages and Systems - 33rd European Symposium on Programming,
                  {ESOP} 2024, Held as Part of the European Joint Conferences on Theory
                  and Practice of Software, {ETAPS} 2024, Luxembourg City, Luxembourg,
                  April 6-11, 2024, Proceedings, Part {II}},
  series       = {Lecture Notes in Computer Science},
  volume       = {14577},
  pages        = {90--118},
  publisher    = {Springer},
  year         = {2024},
  url          = {https://doi.org/10.1007/978-3-031-57267-8\_4},
  doi          = {10.1007/978-3-031-57267-8\_4},
  bibsource    = {dblp computer science bibliography, https://dblp.org}
}

@misc{stainless_lexer,
      title={Verified invertible lexer using regular expressions and DFAs}, 
      author={Samuel Chassot and Viktor Kunčak},
      year={2024},
      eprint={2412.13581},
      archivePrefix={arXiv},
      primaryClass={cs.PL},
      url={https://arxiv.org/abs/2412.13581}, 
}

@article{coqlex,
  author       = {Wendlasida Ouedraogo and
                  Gabriel Scherer and
                  Lutz Stra{\ss}burger},
  title        = {Coqlex: Generating Formally Verified Lexers},
  journal      = {Art Sci. Eng. Program.},
  volume       = {8},
  number       = {1},
  year         = {2024},
  url          = {https://doi.org/10.22152/programming-journal.org/2024/8/3},
  doi          = {10.22152/PROGRAMMING-JOURNAL.ORG/2024/8/3},
  bibsource    = {dblp computer science bibliography, https://dblp.org}
}

@inproceedings{verified_operational_semantics,
  author       = {Elton Maximo Cardoso and
                  Leonardo Vieira dos Santos Reis and
                  Rodrigo Geraldo Ribeiro},
  title        = {A Verified Operational Semantics for Regular Expression Parsing},
  booktitle    = {Proceedings of the {XXVII} Brazilian Symposium on Programming Languages,
                  {SBLP} 2023, Campo Grande, MS, Brazil, September 25-29, 2023},
  pages        = {82--90},
  publisher    = {{ACM}},
  year         = {2023},
  url          = {https://doi.org/10.1145/3624309.3624317},
  doi          = {10.1145/3624309.3624317},
  bibsource    = {dblp computer science bibliography, https://dblp.org}
}

@article{groups_transducers,
  author       = {Martin Berglund and
                  Brink van der Merwe},
  title        = {On the semantics of regular expression parsing in the wild},
  journal      = {Theor. Comput. Sci.},
  volume       = {679},
  pages        = {69--82},
  year         = {2017},
  url          = {https://doi.org/10.1016/j.tcs.2016.09.006},
  doi          = {10.1016/J.TCS.2016.09.006},
  bibsource    = {dblp computer science bibliography, https://dblp.org}
}

@article{pearl_equivalence,
  author       = {Alexander Krauss and
                  Tobias Nipkow},
  title        = {Proof Pearl: Regular Expression Equivalence and Relation Algebra},
  journal      = {J. Autom. Reason.},
  volume       = {49},
  number       = {1},
  pages        = {95--106},
  year         = {2012},
  url          = {https://doi.org/10.1007/s10817-011-9223-4},
  doi          = {10.1007/S10817-011-9223-4},
  bibsource    = {dblp computer science bibliography, https://dblp.org}
}

@inproceedings{bitcoded_parsing,
  author       = {Lasse Nielsen and
                  Fritz Henglein},
  editor       = {Adrian{-}Horia Dediu and
                  Shunsuke Inenaga and
                  Carlos Mart{\'{\i}}n{-}Vide},
  title        = {Bit-coded Regular Expression Parsing},
  booktitle    = {Language and Automata Theory and Applications - 5th International
                  Conference, {LATA} 2011, Tarragona, Spain, May 26-31, 2011. Proceedings},
  series       = {Lecture Notes in Computer Science},
  volume       = {6638},
  pages        = {402--413},
  publisher    = {Springer},
  year         = {2011},
  url          = {https://doi.org/10.1007/978-3-642-21254-3\_32},
  doi          = {10.1007/978-3-642-21254-3\_32},
  bibsource    = {dblp computer science bibliography, https://dblp.org}
}

@inproceedings{certified_bitcoded_parsing,
  author       = {Rodrigo Geraldo Ribeiro and
                  Andr{\'{e}} Rauber Du Bois},
  editor       = {Fabio Mascarenhas},
  title        = {Certified Bit-Coded Regular Expression Parsing},
  booktitle    = {Proceedings of the 21st Brazilian Symposium on Programming Languages,
                  {SBLP} 2017, Fortaleza, CE, Brazil, September 21-22, 2017},
  pages        = {4:1--4:8},
  publisher    = {{ACM}},
  year         = {2017},
  url          = {https://doi.org/10.1145/3125374.3125381},
  doi          = {10.1145/3125374.3125381},
  bibsource    = {dblp computer science bibliography, https://dblp.org}
}

@article{efficient_parse_tree,
  author       = {Danny Dub{\'{e}} and
                  Marc Feeley},
  title        = {Efficiently building a parse tree from a regular expression},
  journal      = {Acta Informatica},
  volume       = {37},
  number       = {2},
  pages        = {121--144},
  year         = {2000},
  url          = {https://doi.org/10.1007/s002360000037},
  doi          = {10.1007/S002360000037},
  bibsource    = {dblp computer science bibliography, https://dblp.org}
}

@inproceedings{greedy_matching,
  author       = {Alain Frisch and
                  Luca Cardelli},
  editor       = {Josep D{\'{\i}}az and
                  Juhani Karhum{\"{a}}ki and
                  Arto Lepist{\"{o}} and
                  Donald Sannella},
  title        = {Greedy Regular Expression Matching},
  booktitle    = {Automata, Languages and Programming: 31st International Colloquium,
                  {ICALP} 2004, Turku, Finland, July 12-16, 2004. Proceedings},
  series       = {Lecture Notes in Computer Science},
  volume       = {3142},
  pages        = {618--629},
  publisher    = {Springer},
  year         = {2004},
  url          = {https://doi.org/10.1007/978-3-540-27836-8\_53},
  doi          = {10.1007/978-3-540-27836-8\_53},
  bibsource    = {dblp computer science bibliography, https://dblp.org}
}

@inproceedings{certified_parsing,
  author       = {Denis Firsov and
                  Tarmo Uustalu},
  editor       = {Georges Gonthier and
                  Michael Norrish},
  title        = {Certified Parsing of Regular Languages},
  booktitle    = {Certified Programs and Proofs - Third International Conference, {CPP}
                  2013, Melbourne, VIC, Australia, December 11-13, 2013, Proceedings},
  series       = {Lecture Notes in Computer Science},
  volume       = {8307},
  pages        = {98--113},
  publisher    = {Springer},
  year         = {2013},
  url          = {https://doi.org/10.1007/978-3-319-03545-1\_7},
  doi          = {10.1007/978-3-319-03545-1\_7},
  bibsource    = {dblp computer science bibliography, https://dblp.org}
}

@article{lookarounds_afa,
  author       = {Martin Berglund and
                  Brink van der Merwe and
                  Steyn van Litsenborgh},
  title        = {Regular Expressions with Lookahead},
  journal      = {J. Univers. Comput. Sci.},
  volume       = {27},
  number       = {4},
  pages        = {324--340},
  year         = {2021},
  url          = {https://doi.org/10.3897/jucs.66330},
  doi          = {10.3897/JUCS.66330},
  bibsource    = {dblp computer science bibliography, https://dblp.org}
}

@article{simplifying_regular,
  author       = {Stefan Kahrs and
                  Colin Runciman},
  title        = {Simplifying regular expressions further},
  journal      = {J. Symb. Comput.},
  volume       = {109},
  pages        = {124--143},
  year         = {2022},
  url          = {https://doi.org/10.1016/j.jsc.2021.08.003},
  doi          = {10.1016/J.JSC.2021.08.003},
  bibsource    = {dblp computer science bibliography, https://dblp.org}
}

@inproceedings{derivatives_enhanced,
  author       = {Peter Thiemann},
  editor       = {Yo{-}Sub Han and
                  Kai Salomaa},
  title        = {Derivatives for Enhanced Regular Expressions},
  booktitle    = {Implementation and Application of Automata - 21st International Conference,
                  {CIAA} 2016, Seoul, South Korea, July 19-22, 2016, Proceedings},
  series       = {Lecture Notes in Computer Science},
  volume       = {9705},
  pages        = {285--297},
  publisher    = {Springer},
  year         = {2016},
  url          = {https://doi.org/10.1007/978-3-319-40946-7\_24},
  doi          = {10.1007/978-3-319-40946-7\_24},
  bibsource    = {dblp computer science bibliography, https://dblp.org}
}

@article{resharp,
  author       = {Ian Erik Varatalu and
                  Margus Veanes and
                  Juhan P. Ernits},
  title        = {RE{\#}: High Performance Derivative-Based Regex Matching with Intersection,
                  Complement, and Restricted Lookarounds},
  journal      = {Proc. {ACM} Program. Lang.},
  volume       = {9},
  number       = {{POPL}},
  pages        = {1--32},
  year         = {2025},
  url          = {https://doi.org/10.1145/3704837},
  doi          = {10.1145/3704837},
  bibsource    = {dblp computer science bibliography, https://dblp.org}
}

@article{decidability,
  author       = {Dominik D. Freydenberger},
  title        = {Extended Regular Expressions: Succinctness and Decidability},
  journal      = {Theory Comput. Syst.},
  volume       = {53},
  number       = {2},
  pages        = {159--193},
  year         = {2013},
  url          = {https://doi.org/10.1007/s00224-012-9389-0},
  doi          = {10.1007/S00224-012-9389-0},
  bibsource    = {dblp computer science bibliography, https://dblp.org}
}

@inproceedings{black_ostrich,
  author       = {Benjamin Eriksson and
                  Amanda Stjerna and
                  Riccardo De Masellis and
                  Philipp R{\"{u}}mmer and
                  Andrei Sabelfeld},
  editor       = {Weizhi Meng and
                  Christian Damsgaard Jensen and
                  Cas Cremers and
                  Engin Kirda},
  title        = {Black Ostrich: Web Application Scanning with String Solvers},
  booktitle    = {Proceedings of the 2023 {ACM} {SIGSAC} Conference on Computer and
                  Communications Security, {CCS} 2023, Copenhagen, Denmark, November
                  26-30, 2023},
  pages        = {549--563},
  publisher    = {{ACM}},
  year         = {2023},
  url          = {https://doi.org/10.1145/3576915.3616582},
  doi          = {10.1145/3576915.3616582},
  bibsource    = {dblp computer science bibliography, https://dblp.org}
}

@inproceedings{jiset,
  author       = {Jihyeok Park and
                  Jihee Park and
                  Seungmin An and
                  Sukyoung Ryu},
  title        = {{JISET:} JavaScript IR-based Semantics Extraction Toolchain},
  booktitle    = {35th {IEEE/ACM} International Conference on Automated Software Engineering,
                  {ASE} 2020, Melbourne, Australia, September 21-25, 2020},
  pages        = {647--658},
  publisher    = {{IEEE}},
  year         = {2020},
  url          = {https://doi.org/10.1145/3324884.3416632},
  doi          = {10.1145/3324884.3416632},
  bibsource    = {dblp computer science bibliography, https://dblp.org}
}

@inproceedings{jscert,
  author = {Martin Bodin and Arthur Chargu{\'{e}}raud and Daniele Filaretti and
            Philippa Gardner and Sergio Maffeis and Daiva Naudziuniene and Alan
            Schmitt and Gareth Smith},
  editor = {Suresh Jagannathan and Peter Sewell},
  title = {A Trusted Mechanised {JavaScript} Specification},
  booktitle = {The 41st Annual {ACM} {SIGPLAN-SIGACT} Symposium on Principles of
               Programming Languages, {POPL} '14, San Diego, CA, USA, January
               20-21, 2014},
  pages = {87--100},
  publisher = {{ACM}},
  year = {2014},
  doi = {10.1145/2535838.2535876},
  bibsource = {dblp computer science bibliography, https://dblp.org},
}

@misc{lean_regex,
  author = {Yulu Pan},
  howpublished = {https://github.com/pandaman64/lean-regex},
  year = 2025,
  title = {Lean-regex}
}

\newpage
\appendix
\section{Correspondence between the paper and the mechanization}
\label{app:correspondence}

\newcommand{\rocq}[1]{\small\texttt{#1}}
\newcommand{\fname}[1]{\small\texttt{#1}}

The following table presents a one-to-one correspondence between the paper and the code included in the supplementary material.

\begin{longtable}{@{}l@{\hspace{0.3em}}l@{\hspace{0.3em}}l@{}}
  \textbf{Paper Definition} & \textbf{File} & \textbf{Rocq Name}\\
  \hline
  \Autoref{sec:background} & & \\
  \hline
  \Autoref{fig:syntax} & \fname{Semantics/Regex.v} & \rocq{regex, anchor, lookaround}\\
  \Autoref{fig:syntax} & \fname{Semantics/Chars.v} & \rocq{char\_descr}\\
  \hline
  \Autoref{sec:semantics} & & \\
  \hline
  \Autoref{fig:tree_example} & \fname{Semantics/Examples.v} & \rocq{fig2\_tree}\\
  \Autoref{fig:tree_syntax} & \fname{Semantics/Tree.v} & \rocq{tree}\\
  \Autoref{fig:tree_semantics}, $\istree{\cont}{\inp}{\gm}{\dir}{\treecont}$ & \fname{Semantics/Semantics.v} & \rocq{is\_tree $\cont$ $\inp$ $\gm$ $\dir$ $\treecont$} \\
  List of actions $\cont$ & \fname{Semantics/Semantics.v} & \rocq{actions}\\
  Input $\inp$ (zipper) & \fname{Semantics/Chars.v} & \rocq{input}\\
  $\inpgt{\inp_1}{\inp_2}{\dir}$ & \fname{Semantics/StrictSuffix.v} & \rocq{strict\_suffix $\inp_1$ $\inp_2$ $\dir$}\\
  $\idx{\inp}$ & \fname{Semantics/Chars.v} & \rocq{idx $\inp$}\\
  $\nextinpdir{\inpcheck}{\dir}$ & \fname{Semantics/Chars.v} & \rocq{advance\_input $\inpcheck$ $\dir$} \\
  Group map $\gm$ & \fname{Semantics/Groups.v} & \rocq{GroupMap.t} \\
  $\gmempty$ & \fname{Semantics/Groups.v} & \rocq{GroupMap.empty} \\
  $\gmopen{\gm}{\gid}{n}$ & \fname{Semantics/Groups.v} & \rocq{GroupMap.open $n$ $\gid$ $\gm$} \\
  $\gmclose{\gm}{\gid}{n}$ & \fname{Semantics/Groups.v} & \rocq{GroupMap.close $n$ $\gid$ $\gm$} \\
  $\gmreset{\gm}{\gidl}$ & \fname{Semantics/Groups.v} & \rocq{GroupMap.reset $\gidl$ $\gm$} \\
  $\inpadvance{\cd}{\inp}{\dir}$ & \fname{Semantics/Chars.v} & \rocq{read\_char $\cd$ $\inp$ $\dir$} \\
  $\checkanchor{\anc}{\inp}$ & \fname{Semantics/Semantics.v} & \rocq{anchor\_satisfied $\anc$ $\inp$} \\
  $\readbackref{\gm}{\gid}{\inp}{\dir}$ & \fname{Semantics/Semantics.v} & \rocq{read\_backref $\gm$ $\gid$ $\inp$ $\dir$} \\
  $\defgroups{r}$ & \fname{Semantics/Regex.v} & \rocq{def\_groups $r$} \\
  $\lkresult{\lk}{\treelook}{\gm}{i}$ & \fname{Semantics/Semantics.v} & \rocq{lk\_group\_map $\lk$ $\treelook$ $\gm$ $i$} \\
  $\firstbranch{\treecont}{\inp}$ & \fname{Semantics/Tree.v} & \rocq{first\_leaf $\treecont$ $\inp$} \\
  \Autoref{thm:tree_det} & \fname{Semantics/Semantics.v} & \rocq{is\_tree\_determ} \\
  $\computetreefuel{\actions}{\inp}{\gm}{\dir}{n}$ & \fname{Semantics/FunctionalSemantics.v} & \rocq{compute\_tree $\actions$ $\inp$ $\gm$ $\dir$ $n$} \\
  $\fuel{\actions}{\inp}{\dir}$ & \fname{Semantics/FunctionalSemantics.v} & \rocq{actions\_fuel $\actions$ $\inp$ $\dir$} \\
  $\fuel{r}{\inp}{\dir}$ & \fname{Semantics/FunctionalSemantics.v} & \rocq{regex\_fuel $r$ $\inp$ $\dir$} \\
  $\inpsize{\inp}{\dir}$ & \fname{Semantics/FunctionalSemantics.v} & \rocq{max\_iter $\inp$ $\dir$} \\
  $\worstinp{\lk}{\inp}$ & \fname{Semantics/FunctionalSemantics.v} & \rocq{worst\_input $\inp$ $\dir$} \\
  $\lkdir{\lk}$ & \fname{Semantics/Regex.v} & \rocq{lk\_dir $\lk$} \\
  \Autoref{thm:termination} & \fname{Semantics/FunctionalSemantics.v} & \rocq{functional\_terminates} \\
  $\computetree{\actions}{\inp}{\gm}{\dir}$ & \fname{Semantics/FunctionalUtils.v} & \rocq{compute\_tr $\actions$ $\inp$ $\gm$ $\dir$} \\
  \Autoref{thm:functional_correctness} & \fname{Semantics/ComputeIsTree.v} & \rocq{compute\_is\_tree} \\
  \hline
  \Autoref{sec:warblre_equiv} & & \\
  \hline
  $\tolinden{\subreg_w}$ & \fname{WarblreEquiv/RegexpTranslation.v} & \rocq{warblre\_to\_linden} \\
  $\towarblre{\result}$ & \fname{WarblreEquiv/ResultTranslation.v} & \rocq{to\_MatchState} \\
  \Autoref{thm:warblre_equiv} & \fname{WarblreEquiv/EquivMain.v} & \rocq{equiv\_main\_reconstruct} \\
  Equivalence relation & \fname{WarblreEquiv/EquivDef.v} & \rocq{equiv\_cont} \\
  between continuations & & \\
  and lists of actions & & \\
  \hline
  \Autoref{sec:rewrite} & & \\
  \hline
  $\regeq{\subreg_1}{\subreg_2}$ & \fname{Rewriting/Equivalence.v} & \rocq{observ\_equiv} \\
  Regex contexts $\ctx$ & \fname{Rewriting/Equivalence.v} & \rocq{regex\_ctx} \\
  $\plug{\ctx}{\subreg}$ & \fname{Rewriting/Equivalence.v} & \rocq{plug\_ctx $\ctx$ $\subreg$} \\
  Type of context $\ctx$ & \fname{Rewriting/Equivalence.v} & \rocq{ctx\_dir $\ctx$} \\
  $\leaves{\treecont}{\inp}{\dir}$ & \fname{Semantics/Tree.v} & \rocq{tree\_leaves} $\treecont$ $\gmempty$ $\inp$ $\dir$ \\
  $\leaveseq{\ell_1}{\ell_2}$ & \fname{Rewriting/LeavesEquivalence.v} & \rocq{leaves\_equiv [] $\ell_1$ $\ell_2$} \\
  $\leafeqdir{\re{\subreg_1}}{\re{\subreg_2}}{\dir}$ & \fname{Rewriting/Equivalence.v} & \rocq{tree\_equiv\_dir} \\
  $\leafeqdir{\re{\subreg_1}}{\re{\subreg_2}}{\both}$ & \fname{Rewriting/Equivalence.v} & \rocq{tree\_equiv} \\
  \Autoref{thm:same_equiv} & \fname{Rewriting/Equivalence.v} & \rocq{regex\_equiv\_ctx\_samedir} \\
  \Autoref{thm:forward_equiv} & \fname{Rewriting/Equivalence.v} & \rocq{regex\_equiv\_ctx\_forward} \\
  \Autoref{thm:backward_equiv} & \fname{Rewriting/Equivalence.v} & \rocq{regex\_equiv\_ctx\_backward} \\
  \Autoref{thm:ctx_observ_equiv} & \fname{Rewriting/Equivalence.v} & \rocq{observe\_equivalence} \\
  \Autoref{fig:distr_counterex} & \fname{Semantics/Example.v} & \rocq{different\_results} \\
  \hline
  \Autoref{fig:correct_rewrites}: & & \\
  \re{\disjunction{\subreg_1}{\noncap{\disjunction{\subreg_2}{\subreg_3}}}} $\botheq$ \re{\disjunction{\noncap{\disjunction{\subreg_1}{\subreg_2}}}{\subreg_3}} & \fname{Rewriting/Associativity.v} & \rocq{disj\_assoc} \\
  \re{\subreg_1\noncap{\subreg_2\subreg_3}} $\botheq$ \re{\noncap{\subreg_1\subreg_2}\subreg_3} & \fname{Rewriting/Associativity.v} & \rocq{seq\_assoc} \\
  \re{\noncap{\disjunction{\subreg_2}{\subreg_3}}\subreg_1 \fwdeq \disjunction{\noncap{\subreg_2\subreg_1}}{\noncap{\subreg_3\subreg_1}}} & \fname{Rewriting/Distributivity.v} & \rocq{factored\_expanded\_} \\
  when \re{\subreg_1} has no group & & \rocq{right\_equiv} \\
  \re{\subreg_1\noncap{\disjunction{\subreg_2}{\subreg_3}} \bwdeq \disjunction{\noncap{\subreg_1\subreg_2}}{\noncap{\subreg_1\subreg_3}}} & \fname{Rewriting/Distributivity.v} & \rocq{factored\_expanded\_} \\
  when \re{\subreg_1} has no group & & \rocq{left\_equiv} \\
  Anchors as lookarounds & \fname{Rewriting/Anchors.v} & \rocq{desugar\_anchor\_correct} \\
  \re{\regall} & \fname{Semantics/Chars.v} & \rocq{CdAll} \\
  \hline
  $\leafeqdir{\re{\quant{\subreg}{\rmin}{0}{\top}}}{\re{\quant{\subreg}{\rmin}{0}{\bot}}}{\both}$ & \fname{Rewriting/ForcedQuant.v} & \rocq{forced\_equiv} \\
  \Autoref{fig:quantifier_merge_correct}: & & \\
  \re{\quant{\subreg}{\rmin_1}{0}{\greedy} \quant{\subreg}{\rmin_2}{0}{\greedy}} $\botheq$ & \fname{Rewriting/RegexpTree.v} & \rocq{bounded\_bounded\_equiv} \\
  \re{\quant{\subreg}{\rmin_1+\rmin_2}{0}{\greedy}} & & \\
  \re{\quant{\subreg}{\rmin_1}{0}{\greedy} \quant{\subreg}{0}{\Delta_2}{\top}} $\fwdeq$ & \fname{Rewriting/RegexpTree.v} & \rocq{bounded\_atmost\_equiv} \\
  \re{\quant{\subreg}{\rmin_1}{\Delta_2}{\top}} & & \\
  \re{\quant{\subreg}{\rmin_1}{0}{\greedy} \quant{\subreg}{0}{\Delta_2}{\bot}} $\fwdeq$ & \fname{Rewriting/RegexpTree.v} & \rocq{bounded\_atmost\_lazy\_equiv} \\
  \re{\quant{\subreg}{\rmin_1}{\Delta_2}{\bot}} & & \\
  \re{\quant{\subreg}{0}{\Delta_1}{\top} \quant{\subreg}{\rmin_2}{0}{\greedy}} $\bwdeq$ & \fname{Rewriting/RegexpTree.v} & \rocq{atmost\_bounded\_equiv} \\
  \re{\quant{\subreg}{\rmin_2}{\Delta_1}{\top}} & & \\
  \re{\quant{\subreg}{0}{\Delta_1}{\bot} \quant{\subreg}{\rmin_2}{0}{\greedy}} $\bwdeq$ & \fname{Rewriting/RegexpTree.v} & \rocq{atmost\_bounded\_lazy\_equiv} \\
  \re{\quant{\subreg}{\rmin_2}{\Delta_1}{\bot}} & & \\
  \re{\quant{\subreg}{0}{\Delta_1}{\top} \quant{\subreg}{0}{\Delta_2}{\top}} $\botheq$ & \fname{Rewriting/RegexpTree.v} & \rocq{atmost\_atmost\_equiv} \\
  \re{\quant{\subreg}{0}{\Delta_1 + \Delta_2}{\top}} & &\\
  \hline
  Chain of forward equivalences & \fname{Rewriting/Chain.v} & \rocq{equivalence\_chain}\\
  \hline
  \Autoref{sec:pikevm} & & \\
  \hline
  \Autoref{fig:syntax_pikevm} & \fname{Engine/PikeSubset.v} & \rocq{pike\_regex} \\
  Subset of actions & \fname{Engine/PikeSubset.v} & \rocq{pike\_action} \\
  Subset of trees & \fname{Engine/PikeSubset.v} & \rocq{pike\_subtree} \\
  \Autoref{fig:bytecode} & \fname{Engine/NFA.v} & \rocq{bytecode}, \rocq{code} \\
  \Autoref{fig:pikevm_compile} & \fname{Engine/NFA.v} & \rocq{compile} \\
  Label~$\lbl$ & \fname{Engine/NFA.v} & \rocq{label} \\
  \accept~ instruction appended & \fname{Engine/NFA.v} & \rocq{compilation} \\
  Thread $\thread{\pc}{\gm}{\bo}$ & \fname{Engine/PikeVM.v} & \rocq{thread} \\
  \Autoref{fig:pikevm_smallstep} & \fname{Engine/PikeVM.v} & \rocq{pike\_vm\_step} \\
  States of PikeVM & \fname{Engine/PikeVM.v} & \rocq{pike\_vm\_state} \\
  $\pvsinit{\inp}$ & \fname{Engine/PikeVM.v} & \rocq{pike\_vm\_initial\_state} \\
  \Autoref{fig:pike_ex} & \fname{Engine/FunctionalPikeVM.v} & \rocq{paper\_regex} \\
  \Autoref{subfig:bytecode} & \fname{Engine/FunctionalPikeVM.v} & \rocq{paper\_bytecode} \\
  \Autoref{subfig:tree} & \fname{Engine/FunctionalPikeVM.v} & \rocq{paper\_tree} \\
  \Autoref{fig:bool_semantics} & \fname{Engine/BooleanSemantics.v} & \rocq{bool\_tree} \\
  \Autoref{fig:encodes}, $\encodes{\actions}{\inp}{\bo}$ & \fname{Engine/BooleanSemantics.v} & \rocq{bool\_encoding $\bo$ $\inp$ $\actions$} \\
  \Autoref{thm:boolean} & \fname{Engine/BooleanSemantics.v} & \rocq{encode\_equal} \\
  \Autoref{thm:boolean_correct} & \fname{Engine/BooleanSemantics.v} & \rocq{booltree\_istree\_equiv} \\
  \Autoref{fig:piketree_smallstep} & \fname{Engine/PikeTree.v} & \rocq{pike\_tree\_step} \\
  States of PikeTree & \fname{Engine/PikeTree.v} & \rocq{pike\_tree\_state} \\
  $\ptsinit{\treecont}{\inp}$ & \fname{Engine/PikeTree.v} & \rocq{pike\_tree\_initial\_state} \\
  $\piketreeinv{\pts}{\result}$ & \fname{Engine/PikeTree.v} & \rocq{piketreeinv $\pts$ $\result$} \\
  $\ptres{(\treecont,\gm)}{\result}{\seen}{\inp}$ & \fname{Engine/PikeTree.v} & \rocq{tree\_nd $\treecont$ $\gm$ $\inp$ $\seen$ $\result$} \\
  $\ptres{\pactive}{\result}{\seen}{\inp}$ & \fname{Engine/PikeTree.v} & \rocq{list\_nd $\pactive$ $\inp$ $\seen$ $\result$} \\
  $\ptres{\ptstate{\inp}{\best}{\pactive}{\blocked}{\seen}}{\result}{}{}$ & \fname{Engine/PikeTree.v} & \rocq{state\_nd $\inp$ $\pactive$ $\best$ $\blocked$ $\seen$ $\result$} \\
  \Autoref{thm:piketree_init} & \fname{Engine/PikeTree.v} & \rocq{init\_piketree\_inv} \\
  \Autoref{thm:piketree} & \fname{Engine/PikeTree.v} & \rocq{pts\_preservation} \\
  \Autoref{fig:tree_thread} & \fname{Engine/PikeEquiv.v} & \rocq{tree\_thread} \\
  $\rep{\code}{\cont}{\pc}$ & \fname{Engine/NFA.v} & \rocq{actions\_rep} \\
  $\treethread{\pactive}{\pactive'}{\code}{\inp}$ & \fname{Engine/PikeEquiv.v} & \rocq{list\_tree\_thread} \\
  $\seenincl{\seen_{VM}}{\seen}{\pactive}$ & \fname{Engine/PikeEquiv.v} & \rocq{seen\_inclusion} \\
  \Autoref{fig:pikevm_inv} & \fname{Engine/PikeEquiv.v} & \rocq{pike\_inv} \\
  \Autoref{thm:pikevm_init} & \fname{Engine/PikeEquiv.v} & \rocq{initial\_pike\_inv} \\
  \Autoref{thm:pikevm} & \fname{Engine/PikeEquiv.v} & \rocq{invariant\_preservation} \\
  \Autoref{thm:pikevm_trc} & \fname{Engine/Correctness.v} & \rocq{pike\_vm\_to\_pike\_tree} \\
  \Autoref{thm:pikevm_final} & \fname{Engine/Correctness.v} & \rocq{pikevm\_warblre} \\
\end{longtable}

\newpage
\section{Counterexamples for incorrect equivalences}
\label{app:counterex}

In this section, we provide counter-examples for every distributivity or quantifier merging equivalence we have not proved in \Autoref{sec:rewrite}.
We use the character $\square$ to denote the hole in the context.
For counter examples of backward equivalences, we present a backward context where the $\square$ is inside a lookbehind.
We also provide equivalent JavaScript syntax that can be copy-pasted directly into any JavaScript engine to check that the two regexes return different results.

\subsection{Distributivity}
Even when $\subreg_1$ does not contain capture groups, there are two incorrect distributivity equivalences.

\counterexample
    {$\leafeqdir{\re{\subreg_1\noncap{\disjunction{\subreg_2}{\subreg_3}}}}{\re{\disjunction{\noncap{\subreg_1\subreg_2}}{\noncap{\subreg_1\subreg_3}}}}{\forward}$}
    {$\subreg_1 = \re{\noncap{\disjunction{a}{\noncap{ab}}}}$\hfill
      $\subreg_2 = \re{c}$\hfill
      $\subreg_3 = \re{b}$}
    {\re{\square}}
    {abc}
    {\jscode{"abc".match(/(?:a|(?:ab))(?:c|b)/);}}
    {\jscode{"abc".match(/(?:(?:a|(?:ab))c)(?:(?:a|(?:ab))b)/);}}

\counterexample
    {$\leafeqdir{\re{\noncap{\disjunction{\subreg_2}{\subreg_3}}\subreg_1}}{\re{\disjunction{\noncap{\subreg_2\subreg_1}}{\noncap{\subreg_3\subreg_1}}}}{\backward}$ }
    {$\subreg_1 = \re{\noncap{\disjunction{c}{\noncap{bc}}}}$\hfill
      $\subreg_2 = \re{a}$\hfill
      $\subreg_3 = \re{b}$}
    {\re{abc\lookbehind{\group{1}{\square}}\backref{1}}}
    {abcabc}
    {\jscode{"abcabc".match(/abc(?<=((?:a|b)(?:c|(?:bc))))\\1/);}}
    {\jscode{"abcabc".match(/abc(?<=((?:a(?:c|(?:bc))|(?:b(?:c|(?:bc))))))\\1/);}}

\subsection{Quantifier Merging}

\counterexample
    {$\leafeqdir{\re{\quant{\subreg}{0}{\Delta_1}{\top}\quant{\subreg}{\rmin_2}{0}{\greedy}}}{\re{\quant{\subreg}{\rmin_2}{\Delta_1}{\top}}}{\forward}$}
    {$\subreg = \re{\noncap{\disjunction{a}{\noncap{ab}}}}$\hfill
      $\Delta_1 = 1$\hfill
      $\rmin_2 = 1$}
    {\re{\square}}
    {aba}
    {\jscode{"aba".match(/(?:a|(?:ab))\{0,1\}(?:a|(?:ab))\{1\}/);}}
    {\jscode{"aba".match(/(?:a|(?:ab))\{1,2\}/);}}

\counterexample
    {$\leafeqdir{\re{\quant{\subreg}{0}{\Delta_1}{\bot}\quant{\subreg}{\rmin_2}{0}{\greedy}}}{\re{\quant{\subreg}{\rmin_2}{\Delta_1}{\bot}}}{\forward}$}
    {$\subreg = \re{\noncap{\disjunction{\noncap{ab}}{\noncap{aba}}}}$\hfill
      $\Delta_1 = 1$\hfill
      $\rmin_2 = 1$}
    {\re{\square\noncap{\disjunction{\noncap{bcd}}{c}}}}
    {ababcd}
    {\jscode{"ababcd".match(/(?:(?:ab)|(?:aba))\{0,1\}?(?:(?:ab)|(?:aba))\{1\}?(?:(?:bcd)|c)/);}}
    {\jscode{"ababcd".match(/(?:(?:ab)|(?:aba))\{1,2\}?(?:(?:bcd)|c)/);}}

\counterexample
    {$\leafeqdir{\re{\quant{\subreg}{\rmin_1}{0}{\greedy}\quant{\subreg}{0}{\Delta_2}{\top}}}{\re{\quant{\subreg}{\rmin_1}{\Delta_2}{\top}}}{\backward}$}
    {$\subreg = \re{\noncap{\disjunction{a}{\noncap{ba}}}}$\hfill
      $\rmin_1 = 1$\hfill
      $\Delta_2 = 1$}
    {\re{caba\lookbehind{\noncap{\disjunction{\noncap{cab}}{c}}\group{1}{\square}}\backref{1}}}
    {cabaa}
    {\jscode{"cabaa".match(/caba(?<=(?:(?:cab)|c)((?:a|(?:ba))\{1\}(?:a|(?:ba))\{0,1\}))\\1/);}}
    {\jscode{"cabaa".match(/caba(?<=(?:(?:cab)|c)((?:a|(?:ba))\{1,2\}))\\1/);}}

\counterexample
    {$\leafeqdir{\re{\quant{\subreg}{\rmin_1}{0}{\greedy}\quant{\subreg}{0}{\Delta_2}{\bot}}}{\re{\quant{\subreg}{\rmin_1}{\Delta_2}{\bot}}}{\backward}$}
   {$\subreg = \re{\noncap{\disjunction{\noncap{ba}}{\noncap{aba}}}}$\hfill
      $\rmin_1 = 1$\hfill
      $\Delta_2 = 1$}
    {\re{cbaba\lookbehind{\noncap{\disjunction{c}{\noncap{cb}}}\group{1}{\square}}\backref{1}}}
    {cbabababa}
    {\jscode{"cbabababa".match(/cbaba(?<=(?:c|(?:cb))((?:(?:ba)|(?:aba))\{1\}(?:(?:ba)|(?:aba))\{0,1\}?))\\1/);}}
    {\jscode{"cbabababa".match(/cbaba(?<=(?:c|(?:cb))((?:(?:ba)|(?:aba))\{1,2\}?))\\1/);}}

\counterexample
    {$\leafeqdir{\re{\quant{\subreg}{0}{\Delta_1}{\bot}\quant{\subreg}{0}{\Delta_2}{\bot}}}{\re{\quant{\subreg}{0}{\Delta_1+\Delta_2}{\bot}}}{\forward}$}
    {$\subreg = \re{\noncap{\disjunction{\noncap{ab}}{\noncap{aba}}}}$\hfill
      $\Delta_1 = 1$\hfill
      $\Delta_2 = 1$}
    {\re{\square\noncap{\disjunction{b}{c}}}}
    {ababc}
    {\jscode{"ababc".match(/(?:(?:ab)|(?:aba))\{0,1\}?(?:(?:ab)|(?:aba))\{0,1\}?(?:b|c)/);}}
    {\jscode{"ababc".match(/(?:(?:ab)|(?:aba))\{0,2\}?(?:b|c)/);}}

\counterexample
    {$\leafeqdir{\re{\quant{\subreg}{0}{\Delta_1}{\bot}\quant{\subreg}{0}{\Delta_2}{\bot}}}{\re{\quant{\subreg}{0}{\Delta_1+\Delta_2}{\bot}}}{\backward}$}
   {$\subreg = \re{\noncap{\disjunction{\noncap{ba}}{\noncap{aba}}}}$\hfill
      $\Delta_1 = 1$\hfill
      $\Delta_2 = 1$}
    {\re{cbaba\lookbehind{\noncap{\disjunction{c}{\noncap{cb}}}\group{1}{\square}}\backref{1}}}
    {cbabababa}
    {\jscode{"cbabababa".match(/cbaba(?<=(?:c|(?:cb))((?:(?:ba)|(?:aba))\{0,1\}?(?:(?:ba)|(?:aba))\{0,1\}?))\\1/);}}
    {\jscode{"cbabababa".match(/cbaba(?<=(?:c|(?:cb))((?:(?:ba)|(?:aba))\{1,2\}?))\\1/);}}

\subsection{Checking these counter-examples in other languages}
The distributivity and quantifier merging counter-examples above should hold for any regex language following backtracking semantics, when their feature set supports it.
We have considered the following popular regex languages: \textbf{.NET, PCRE2, PCRE, Python, Go, Rust, and Java}.

\paragraph{Forward counter-examples}
All of these regex languages support all features used in our forward counter-examples.
We checked that all our forward counter-examples are indeed counter-examples for all these languages.

\paragraph{Backward counter-examples}
Our backward counter-examples use features that are not supported in all languages.
First, Go and Rust support neither backreferences nor lookarounds.
Then, PCRE2, PCRE and Python disallow lookbehinds \re{\lookbehind{\subreg}} when the length of all strings matched by $\subreg$ is not known to be constant (which is not the case in our counter-examples).
As a result, our backward counter-examples are not applicable in Go, Rust, PCRE2, PCRE, and Python.

.NET supports all of these features, and we checked that all our backward counter-examples are indeed counter-examples for .NET.

Finally, while Java regexes support lookbehinds, they do not follow backtracking semantics when \textit{inside lookbehinds} (but they do inside lookaheads).
Instead of matching in reverse like JavaScript or .NET, when matching a lookbehind \re{\lookbehind{\subreg}} from a position $p$ in a string, the matching algorithm implemented in OpenJDK first tries to match $\subreg$ from position $p$, then from position $p-1$, then $p-2$\dots~
As a result, the \textit{shortest} match in a lookbehind has priority in Java regexes, and our backward counter-examples are not applicable.
To illustrate these differences in semantics, consider matching the regex \re{abc\lookbehind{\group{1}{\disjunction{abc}{\disjunction{c}{bc}}}}x} on string \str{abcx} in Java: it returns that the capture group 1 matched the substring \str{c} (the middle branch of the disjunction).

\newpage
\section{Small-step semantics of the PikeTree algorithm}
\label{app:smallstep_piketree}

\Autoref{fig:piketree_smallstep_full} describes the full small-step semantics rules of the PikeTree algorithm (of which a few selected rules were depicted on \Autoref{fig:piketree_smallstep}).

\begin{figure}[h]
  \mbox{\infer[\ruledef{piketree:full}{Final}]{\piketreestep{\ptstate{\inp}{\best}{[]}{[]}{\seen}}{\ptsfinal{\best}}}{}}

  \semspace
  \mbox{\infer[\ruledef{piketree:full}{NextChar}]
    {\piketreestep
      {\ptstate{\inp}{\best}{[]}{\blocked}{\seen}}
      {\ptstate{\inp'}{\best}{\blocked}{[]}{\emptyset}}}
    {\nextinp{\inp} = \Some{\inp'} \newpremise \blocked \neq []}}
  
  \semspace
  \mbox{\infer[\ruledef{piketree:full}{Skip}]{\piketreestep
      {\ptstate{\inp}{\best}{(\treecont,\gm) :: \pactive}{\blocked}{\seen}}
      {\ptstate{\inp}{\best}{\pactive}{\blocked}{\seen}}}
    {\intseen{\treecont}{\seen}}}

  \semspace
  \mbox{\infer[\ruledef{piketree:full}{Match}]{\piketreestep
      {\ptstate{\inp}{\best}{(\treematch,\gm) :: \pactive}{\blocked}{\seen}}
      {\ptstate{\inp}{\Some{(\inp,\gm)}}{[]}{\blocked}{\seen'}}}
    {\seen' = \addtseen{\seen}{\treematch}}}

  \semspace
  \mbox{\infer[\ruledef{piketree:full}{Blocked}]{\piketreestep
      {\ptstate{\inp}{\best}{(\treeread{\rchar}{\treecont},\gm)::\pactive}{\blocked}{\seen}}
      {\ptstate{\inp}{\best}{\pactive}{\blocked\app[(\treecont,\gm)]}{\seen'}}}
    {\seen' = \addtseen{\seen}{\treeread{\rchar}{\treecont}}}}

  \semspace
  \mbox{\infer[\ruledef{piketree:full}{Choice}]{\piketreestep
      {\ptstate{\inp}{\best}{(\treechoice{\treecont_1}{\treecont_2},\gm)::\pactive}{\blocked}{\seen}}
      {\ptstate{\inp}{\best}{(\treecont_1,\gm)::(\treecont_2,\gm)::\pactive}{\blocked}{\seen'}}}
    {\seen' = \addtseen{\seen}{\treechoice{\treecont_1}{\treecont_2}}}}

  \semspace
  \mbox{\infer[\ruledef{piketree:full}{Progress}]{\piketreestep
      {\ptstate{\inp}{\best}{(\treeprogress{\treecont},\gm) :: \pactive}{\blocked}{\seen}}
      {\ptstate{\inp}{\best}{(\treecont,\gm) :: \pactive}{\blocked}{\seen'}}}
    {\seen' = \addtseen{\seen}{\treeprogress{\treecont}}}}

  \semspace
  \mbox{\infer[\ruledef{piketree:full}{Open}]{\piketreestep
      {\ptstate{\inp}{\best}{(\treeopen{\gid}{\treecont},\gm) :: \pactive}{\blocked}{\seen}}
      {\ptstate{\inp}{\best}{(\treecont,\gm') :: \pactive}{\blocked}{\seen'}}}
    {\seen' = \addtseen{\seen}{\treeopen{\gid}{\treecont}} \newpremise \gmopen{\gm}{\gid}{\idx{\inp}} = \gm'}}

  \semspace
  \mbox{\infer[\ruledef{piketree:full}{Close}]{\piketreestep
      {\ptstate{\inp}{\best}{(\treeclose{\gid}{\treecont},\gm) :: \pactive}{\blocked}{\seen}}
      {\ptstate{\inp}{\best}{(\treecont,\gm') :: \pactive}{\blocked}{\seen'}}}
    {\seen' = \addtseen{\seen}{\treeclose{\gid}{\treecont}} \newpremise \gmclose{\gm}{\gid}{\idx{\inp}} = \gm'}}

  \semspace
  \mbox{\infer[\ruledef{piketree:full}{Reset}]{\piketreestep
      {\ptstate{\inp}{\best}{(\treereset{\gidl}{\treecont},\gm) :: \pactive}{\blocked}{\seen}}
      {\ptstate{\inp}{\best}{(\treecont,\gm') :: \pactive}{\blocked}{\seen'}}}
    {\seen' = \addtseen{\seen}{\treereset{\gidl}{\treecont}} \newpremise \gmreset{\gm}{\gidl} = \gm'}}
  
\caption{PikeTree small-step semantics}%
\Description{} 
\label{fig:piketree_smallstep_full}
\end{figure}

\end{document}